\documentclass[11pt,letter]{article}
\usepackage{bbding}
\usepackage{framed, xcolor}
\usepackage{booktabs}
\usepackage{epsfig,amssymb,amsfonts,amsthm}
\usepackage{amsmath}

\usepackage{tabularx,multirow}
\usepackage[font=scriptsize,bf]{caption}
\usepackage{tikz,subfigure}
\usepackage[active]{srcltx}
\usepackage[margin=1in]{geometry}

\usepackage{algorithm}
\usepackage{algorithmic}
\usepackage{epsfig,amssymb,amsfonts,amsmath}
\usepackage{multirow}

\usepackage[numbers,sort&compress,sectionbib]{natbib}

\usepackage{amsfonts}
\usepackage{adjustbox}
\newcommand{\minitab}[2][l]{\begin{tabular}{#1}#2\end{tabular}}

\definecolor{ocre}{rgb}{0.72,0,0} 
\definecolor{newblue}{rgb}{0.0,0.0,0.5} 

\usepackage[pagebackref,citecolor=newblue,colorlinks=true,linkcolor=ocre]{hyperref}

\newcommand{\N}{\mathbb{N}}

\newcommand{\ce}{\mathrm{e}}
\newcommand{\pat}{\mathcal{P}}

\newcommand{\dist}{\operatorname{dist}}
\newcommand{\disc}{\operatorname{disc}}
\newcommand{\cont}{\operatorname{cont}}

\newcommand{\poly}{\operatorname{poly}}
\newcommand{\polylog}{\operatorname{polylog}}
\renewcommand{\deg}{d}
\newcommand{\diam}{\operatorname{diam}}
\newcommand{\Odd}{\mathsf{Odd}}
\newcommand{\eps}{\epsilon}

\renewcommand{\leq}{\leqslant}
\renewcommand{\geq}{\geqslant}
\renewcommand{\le}{\leqslant}
\renewcommand{\ge}{\geqslant}

\newcommand{\Oh}{\mathcal{O}}

\newcommand{\Hmax}[2]{\mathsf{H}_{\max}}

\renewcommand{\P}{\mathbf{P}}
\newcommand{\bP}{\mathbf{P}}

\newcommand{\M}{\mathbf{M}}
\newcommand{\bM}{\mathbf{M}}

\newcommand{\coni}{\mbox{\textsc{Cond1}}}
\newcommand{\conii}{\mbox{\textsc{Cond2}}}
\newcommand{\coniii}{\mbox{\textsc{Cond3}}}

\newcommand{\thmref}[1]{Theorem~\ref{thm:#1}}
\newcommand{\thmmref}[1]{Theorem~\ref{thm:#1}}

\newcommand{\lemref}[1]{Lemma~\ref{lem:#1}}

\newcommand{\corref}[1]{Corollary~\ref{cor:#1}}
\newcommand{\obsref}[1]{Observation~\ref{obs:#1}}

\newcommand{\figref}[1]{Figure~\ref{fig:#1}}

\newcommand{\tabref}[1]{Table~\ref{tab:#1}}
\newcommand{\secref}[1]{Section~\ref{sec:#1}}

\newcommand{\eq}[1]{\eqref{eq:#1}}

\newcommand{\Pro}[1]{\mathbf{Pr} \left[\,#1\,\right]}

\newcommand{\Ex}[1]{\mathbf{E} \left[\,#1\,\right]}

\newcommand{\EXX}[2]{\mathbf{E}_{#1} \left[\,#2\,\right]}

\renewcommand{\diam}{\operatorname{diam}}

\renewcommand{\tilde}{\widetilde}
\renewcommand{\epsilon}{\varepsilon}

\def\mod{\operatorname{mod}}

\newtheorem{theorem}{Theorem}  

\newtheorem{lemma}[theorem]{Lemma}
\newtheorem{obs}[theorem]{Observation}

\newtheorem{corollary}[theorem]{Corollary}

\newtheorem{rem}[theorem]{Remark}

\theoremstyle{definition}

\newtheorem{defi}[theorem]{Definition}

\renewcommand{\tilde}{\widetilde}

\title{\textbf{Tight Bounds for Randomized Load Balancing on \\Arbitrary Network Topologies}\footnote{A preliminary version of this work appeared in Proceedings of the 53rd Annual IEEE Symposium on Foundations of Computer Science (FOCS), 2012.}}

\author{Thomas Sauerwald\\
University of Cambridge\\ Cambridge, UK\\
\texttt{thomas.sauerwald@cl.cam.ac.uk}
\and
He Sun\\
Max Planck Institute for Informatics\\
Saarbr\"{u}cken, Germany\\
\texttt{hsun@mpi-inf.mpg.de}\\
}

\begin{document}

\maketitle

\begin{abstract}
We consider the problem of balancing load items (tokens) in networks.
Starting with an arbitrary load distribution,
we allow  nodes to exchange tokens with their neighbors in each round.
The goal is to achieve a distribution where all nodes have nearly the same number of tokens.

For the continuous case where tokens are arbitrarily divisible, most load balancing schemes correspond to Markov chains, whose convergence is fairly well-understood in terms of their spectral gap. However, in many applications, load items cannot be divided arbitrarily, and we need to deal with the discrete case where the load is composed of indivisible tokens. This discretization entails a non-linear behavior due to its rounding errors, which makes this analysis much harder than in the continuous case.

We investigate several randomized protocols for different communication models in the discrete case. As our main result, we prove that for any regular network in the matching model, all nodes have the same load up to an additive constant in (asymptotically) the same number of rounds as required in the continuous case. This generalizes and tightens the previous best result, which only holds for expander graphs \cite{FS09}, and demonstrates that there is almost no difference between the discrete and continuous cases. Our results also provide a positive answer to the question of how well discrete load balancing can be approximated by (continuous) Markov chains, which has been posed by many researchers, e.g., \cite{RSW98,MGS98,LW95}.

\end{abstract}

\numberwithin{theorem}{section}
\numberwithin{equation}{section}
%
%
%
%
%

\newpage

\thispagestyle{empty}

\tableofcontents

\setcounter{page}{0}

\newpage

\section{Introduction}

Consider an application running on a parallel network with $n$ processors. Every processor has initially a certain amount of \emph{tokens} (jobs) and the processors are connected by an arbitrary graph. The goal of load balancing is to reallocate the tokens by transferring them along the edges so that eventually every processor has almost the same number of tokens.

Load balancing is a well-studied problem in distributed systems and has manifold applications in scheduling~\citep{Surana06}, hashing~\cite{M04},
routing~\citep{Cyb89},
numerical computation such as solving partial differential equations~\citep{Zhanga09,Wil91,SS94}, and simulation of dynamics \cite{Boi91}. This trend has been reinforced by the flattening of processor speeds, leading to an increasing usage of multi-core processors
and the emergence of large decentralized networks like P2P networks. 
 Especially for large-scale networks, it is desirable to use local and iterative load balancing protocols, where every processor only needs to know its current and neighboring processors' loads, and based on that, decides how many tokens should be sent or received.


A common classification of iterative load balancing schemes is the distinction between \emph{diffusion} and \emph{matching} models. In the case of diffusion, every node balances its load concurrently with \emph{all} neighbors in each round.
In the matching model, by contrast, every node balances its load with at most \emph{one} neighbor in each round.

We measure the smoothness of the load distribution by the so-called {\em discrepancy}, which is the difference between the maximum and minimum load among all nodes. In view of more complex scenarios where jobs may be removed or injected, the discrepancy
is a more precise measure than the {\em makespan}, which only takes the maximum load into account.

Many studies on load balancing assume that the load is arbitrarily divisible. In this so-called {\em continuous case}, the diffusion scheme corresponds to a Markov chain on the graph, and one can resort to a wide range of techniques to analyze the convergence rate \cite{Boillat90,GLM99,MGS98}. In particular, the {\em spectral gap} captures the time to reach a small discrepancy fairly accurately~\cite{SJ89,RSW98}. This relation continues to hold for the matching model, even if the sequence of matchings is generated randomly, which might be necessary for graphs with no canonical sequence~\cite{BGPS06,MG96}.

However, in many applications a processor's load may consist of jobs which are not further divisible, which is why the continuous case is also referred to as the ``idealized case''~\cite{RSW98}. A common way to model indivisible jobs is the {\em unit-size token model} where one assumes a smallest load entity, the unit-size token, and load is always represented by a multiple of this smallest entity. In the following, we will refer to the unit-size token model as the {\em discrete case}. It is well-known that the discrete case can be roughly modeled by Markov chains, but this approximation may be quite inaccurate~\cite{RSW98}. Hence, the question of a precise quantitative relationship between Markov chains and load balancing algorithms has been posed in many references~\cite{RSW98,MGS98,MG96,LW95,GLM99,SS94}. Unfortunately, the discrete case is much more difficult to analyze due to its nonlinearity that is caused by the roundings to whole tokens in each round.

\citet{MGS98} proved the first rigorous result for the discrete case in the diffusion model. They assume that the number of tokens sent along each edge is
obtained by rounding down the amount of load that would be sent in the continuous case. Using this approach, they established that the discrepancy is at most $\Oh\big(\frac{dn}{1-\lambda}\big)$ after $\Oh\big(\frac{\log (Kn)}{1-\lambda}\big)$ rounds, where $d$ is the degree, $K$ is the discrepancy of the initial load vector, and $1-\lambda$ is the spectral gap of this diffusion matrix.
Similar results for the matching model were shown by \citet{MG96}.

Further progress was made by \citet{RSW98} who introduced the so-called {\em local divergence}, which is a natural parameter that essentially aggregates the sum of load differences over all edges in all rounds. For both the diffusion and matching models, they proved that the local divergence yields an upper bound on the maximum deviation between the continuous and discrete case for the aforementioned rounding down approach. They also computed the local divergence for different networks~(e.g., torus graphs) and proved a general upper bound  that translates into a discrepancy bound of $\Oh\big(\frac{ d \log n}{1-\lambda}\big)$ after $\Oh\big(\frac{ \log (Kn)}{1-\lambda}\big)$ rounds for any $d$-regular graph.

While always rounding down may lead to quick stabilization, the discrepancy can be substantial, i.e., as large as the diameter of the graph (in case of diffusion, it can even be even the diameter times the degree). Therefore, \citet{RSW98} suggested to use randomized rounding in order to get a better approximation of the continuous case.

\citet{FS09} derived several results for this randomized protocol in the matching model. By analyzing the $\ell_2$-version of the local divergence, the so-called {\em local $2$-divergence}, they proved that on many networks, the randomized protocol yields a square root improvement in terms of the achieved discrepancy compared to the deterministic protocol from \cite{RSW98}. Later \citet{BCFFS11} extended some of these results to the diffusion model.

\citet{AKU05} studied the stability of the random matching model on arbitrary networks. In addition to balancing operations in each round, an adversary can place jobs on the processors, and every processor can execute one job in each round. The main result of \cite{AKU05} states that unless the adversary trivially overloads the network, the system is stable, regardless of how the adversary injects new jobs and the balancing operations are rounded.

Closely related to our problem are balancing networks \cite{AHS94}, which are siblings of sorting networks with comparators replaced with balancers. \citet{KP92} gave the first construction of  a balancing network of depth $\Oh( \log n)$ that achieves a discrepancy of one.  \citet{RSW98} established results
for other networks, but these involve a much larger depth. All of these results~\cite{AHS94,KP92,RSW98} require each balancer to be initialized in a special way, while our randomized protocols do not require any specific initialization, and,
therefore making them more robust and practical.

There are also studies in which nodes are equipped with additional abilities. For instance, \citet{EM05} analyzed a load balancing model where every node knows the average load. \citet{ES10} analyzed an algorithm
that uses random-walk based routing of positive and negative tokens to minimize the makespan. Finally, in \citet{ABS12} nodes simulate the continuous process in parallel to decide how many tokens should be sent.

While all aforementioned load balancing protocols can send an arbitrarily large number of tokens along edges,
several studies considered an alternative model in which only a single token can traverse an each edge in each round. Obviously, the convergence is much slower and at least linear in the initial discrepancy (for concrete results, see \citet{AAMR93,GLM99}).


\paragraph{Our Results} We analyze several natural randomized protocols for indivisible, unit-size tokens. All protocols use the idea of randomized rounding in order to ``imitate'' the behavior of the continuous case in each round.

Our main result for the matching model is as follows:

\begin{theorem}\label{thm:main}
Let $G$ be a regular graph with $n$ nodes and $K$ be the discrepancy of the initial load vector. There are constants $c_1, c_2 > 0$ independent of $G$ and $K$ so that with probability at least $1-\exp(-(\log n)^{c_1})$, the discrepancy is at most $c_2$ in $\Oh\left(\frac{\log (Kn)}{1-\lambda(\bP)}\right)$ rounds in the random matching model. This also holds in $\Oh\left(\frac{\log (Kn)}{1-\lambda(\bM)}\right)$ rounds in the balancing circuit model if $d$ is constant. Furthermore,  the expected discrepancy is constant in the same number of rounds in both models
\footnote{For precise definitions of both models, $\lambda(\bP)$, $\lambda(\M)$, and $d$, see \secref{matching}.}.
\end{theorem}

The two bounds on the runtime in \thmref{main} match the ones from the continuous case up to a constant factor (see \thmref{circuitcontinuous}, \thmref{contmatching}, and the lower bound in \thmref{continuouslowermatching}). The previous best result for this protocol holds only for expander graphs and the number of rounds is a factor $(\log \log n)^3$ larger than ours \cite{FS09}. For expander graphs and $K=\poly(n)$, our algorithm needs only $\Theta(\log n)$ rounds, which would be necessary for any centralized algorithm. For general graphs, all previous bounds on the discrepancy include the spectral gap $1-\lambda$. Therefore, especially for graphs which have small expansion like Torus graphs, our main result represents a vast improvement (\tabref{discrepancy}).

We further analyze the matching model on non-regular graphs. In \thmref{logg}, we show that a discrepancy of $\Oh( \log \log n)$ can be achieved within a number of rounds only an $\Oh(\log \log n)$ factor larger than in the continuous case. Together with \thmref{main}, these results demonstrate that for {\em arbitrary} networks, there is essentially no difference between the discrete and continuous cases.

Finally, we also study two natural diffusion-based protocols in the discrete case \cite{BCFFS11,FGS10}. For these protocols our discrepancy bounds depend only polynomially on the maximum degree $\Delta$ and logarithmically on $n$, while again all previous results include the spectral gap or are restricted to special graph classes~\cite{RSW98,BCFFS11,MGS98,FGS10}. For the precise results, we refer the reader to \secref{diffusion}.

\begin{table}
\caption{Comparison with Previous Work\label{tab:discrepancy}}{%
\begin{tabular*}{1\textwidth}{p{1.4cm}p{4.5cm}p{3cm}p{3cm}p{2.5cm}}
\toprule \textbf{\footnotesize ~~Graphs} & {\textbf{\footnotesize ~\;Rounds}}  &\textbf{\footnotesize \!\!Discrepancy}
&\textbf{\footnotesize \!\!Model}
    &\textbf{\footnotesize \!\!Reference}   \\\hline
        \multirow{4}{*}{\minitab[c]{\footnotesize Expander} } &
         \multirow{4}{*}{\minitab[c]{\footnotesize $\Oh(\log (Kn))$} }   & $\Oh(\log n)$   & det. (BC) & \cite{RSW98} \\\cmidrule(l){3-5}
         & & $\Oh(\log \log n)$ &{rand. (BC \& RM)} & {\cite{FS09}}  \\
         & & $\Oh(1)$ & rand. (BC \& RM) &  \thmmref{main} \\
         \midrule
                  \multirow{6}{*}{\minitab[c]{\footnotesize $r$-dim. \\ Torus} } &
                  \multirow{6}{*}{\minitab[c]{\footnotesize $\Oh\big(\log (Kn) \, n^{2/r}\big)$ } }
                  &  $\Oh(n^{1/r})$   & det. (BC) & \cite{RSW98} \\\cmidrule(l){3-5}
         &    & $\Oh(n^{1/(2r)} \sqrt{\log n})$ & rand. (BC) &  \cite{FS09}  \\
             &    & $\Oh(n^{1/(2r)} \log n)$ & rand. (RM) & \cite{FS09}  \\
         &    & $\Oh(1)$ & rand. (BC \& RM) & \thmmref{main}
                  \\
            \midrule
               \multirow{8}{*}{\minitab[c]{\footnotesize $d$-Regular \\  Graphs} } &
                \multirow{8}{*}{\minitab[c]{\footnotesize $\Oh\left(\frac{\log (Kn)}{1-\lambda}\right)$ } }
               & $\Oh\left(\frac{d \log n}{1-\lambda}\right)$   & det. (BC) & \cite{RSW98} \\\cmidrule(l){3-5}
     &    &  {$\Oh\left( \sqrt{\frac{d \log n}{1-\lambda}}\right)$} & rand. (BC) &  \cite{FS09}  \\
         &    & $\Oh(1)$ & rand. (BC) \ &  \thmmref{main}  \\\cmidrule(l){3-5}
          &  & $\Oh\left( \sqrt{\frac{(\log n)^3}{1-\lambda}}\right)$ & {rand. (RM)} &  \cite{FS09}  \\
                         & & $\Oh(1)$ & rand. (RM) &  \thmmref{main}  \\
            \midrule
            \multirow{6}{*}{\minitab[c]{\footnotesize  Arbitrary \\  Graphs} } &
            \multirow{4}{*}{\minitab[c]{\footnotesize  $\Oh\left(\frac{d\cdot \log (Kn)}{1-\lambda}\right)$ }}
             & $\Oh\left(\frac{d \cdot \log n}{1-\lambda}\right)$   & det. (BC) &  \cite{RSW98} \\\cmidrule(l){3-5}
            & & $\Oh\left( \sqrt{\frac{d \cdot \log n}{1-\lambda}}\right)$   & rand. (BC) &  \cite{FS09} \\\cmidrule(l){2-5}
                    &~~$\Oh(\tau_{\cont}(K,n^{-2}))$  & $\Oh( (\log n)^{\eps} )$ & rand. (BC \& RM) & \thmmref{logg}  \\
         &  ~~$\Oh(\tau_{\cont}(K,n^{-2}) \cdot \log \log n)$& $\Oh(\log \log n)$ & rand. (BC \& RM) &  \thmmref{logg}  \\
        \bottomrule
       \end{tabular*}}
            \caption{
            A comparison of our results for the matching model with the previous best results.
       The initial discrepancy is denoted by $K$, and $1-\lambda$ denotes the spectral gap.
       Here, det.~and rand.~refer to deterministic and randomized rounding, respectively. BC stands for the balancing circuit model, and RM stands for the random matching model.
       Note that $\tau_{\cont}(K,n^{-2})$ is the time for the continuous process to reach a discrepancy of $n^{-2}$ with probability $1-n^{-1}$. For the precise definitions, see~\secref{matching}.}
\end{table}

\paragraph{Our Techniques} Our main results are based on the combination of two novel techniques, which may have further applications to other problems. First, instead of analyzing the rounding errors for each edge directly \cite{RSW98,MGS98,MG96,FS09,BCFFS11}, we adopt a token-based viewpoint and relate the movements of tokens to parallel random walks. This establishes a nice analogy between the distribution of tokens and the well-studied balls-and-bins model (see \corref{verysparse}), and constitutes an important application of distributed random walks \cite{AKK11,ANP13} to load balancing.
Second, we employ potential functions to reduce the task of balancing an arbitrary load vector to the task of balancing a {\em sparse} load vector, i.e., a load vector that contains much fewer tokens than $n$. Especially for these sparse load vectors, the token-based viewpoint yields much stronger concentration inequalities than the previously known results \cite{FS09,BCFFS11}.

All our discrepancy bounds make use of the  local $2$-divergence, which is one of the most important tools to quantify the deviation between the continuous and the discrete cases~\cite{RSW98,FS09,BCFFS11}. We prove that for {\em any} graph and {\em any} sequence of matchings, the local $2$-divergence is between $1$ and $\sqrt{2}$, while all previous bounds include graph parameters such as the spectral gap or the (maximum) degree. For the diffusion model, the local $2$-divergence is basically given by $\Theta( \sqrt{\Delta})$, where $\Delta$ is the maximum degree. All previous bounds on the local $2$-divergence in the diffusion model depend on the size and expansion of the graph, or are restricted to certain graph classes.


%
%

\paragraph{Organization}
The remainder of this paper is organized as follows. \secref{matching} introduces the matching-based model
and provides some basic notations and definitions.
In \secref{divergence}, we derive a tight bound on the local $2$-divergence that is further applied to bound the maximum deviation between the continuous and discrete cases. In \secref{randomwalk}, we introduce a novel technique that relates the movement of tokens to independent random walks. Based on this method, we derive bounds on the discrepancy that hold for arbitrary graphs (see \secref{randomwalkresults}). The proof of our main result (\thmref{main}) is given in \secref{main}. Finally, in \secref{diffusion}, we apply similar techniques as in \secref{divergence} to bound the local $2$-divergence for the diffusion model.

\paragraph{Basic Notations}
We assume that  $G=(V,E)$ is an undirected, connected  graph with  $n$ nodes, indexed from $1$ to $n$.  For any node $u$, let $N(u)$ be the set of neighbors of node $u$ and $\deg(u):=|N(u)|$ its degree.
The maximum degree of $G$ is denoted by $\Delta:=\max_{u}\deg(u)$. By $\diam(G)$, we denote the diameter of $G$. Following \cite{RSW98}, we use the notation $[u:v]$ for an edge $\{u,v\} \in E$ with $u < v$.
For any vector $x=(x_1,\dots,x_n)$, the $p$-norm of $x$ is defined by $\|x\|_p:=\left(\sum_{i=1}^n |x_i|^p\right)^{1/p}$. In particular,  $\|x\|_{\infty}:=\max_{1\leq i\leq n}|x_i|$. We also use $x_{\max}:=\max_{1\leq i \leq n} x_i$, $x_{\min}:=\min_{1 \leq i \leq n} x_{i}$, and $\disc(x):=\max_{i,j}|x_i-x_j|=x_{\max}-x_{\min}$. The $n$-dimensional vector $(1,\ldots,1)$ is also denoted by $\mathbf{1}$ and,  similarly, $(0,\ldots,0)$ is denoted by $\mathbf{0}$.
Throughout the paper, all vectors are represented as row vectors.
The $n$ by $n$ identity matrix is denoted by $\mathbf{I}$.
For any $n$ by $n$ real symmetric matrix $\M$, let $\lambda_1(\M)\geq\dots\geq\lambda_n(\M)$ be the $n$ eigenvalues of matrix $\M$. For simplicity, let $\lambda(\M):=\max\{|\lambda_2(\M)|, |\lambda_n(\M)|\}$.
For a non-symmetric matrix $\M$, we define an associated symmetric matrix $\tilde{\M}:=\M\cdot\M^{\mathrm{T}}$, and let $\lambda(\M):=\max\{|\lambda_2(\tilde{\M})|,|\lambda_n(\tilde{\M})|\}$.

By $\log(\cdot)$, we denote the natural logarithm. For any condition $\mathcal{E}$, let $\chi_{\mathcal{E}}$ be the indicator variable which is $1$ if $\mathcal{E}$ holds and $0$, otherwise.

Several inequalities in this paper require that $n$ is
lower bounded by a sufficiently large constant.

%
%

\section{The Matching Model}\label{sec:matching}

In the {\em matching model} (also known as  {\em dimension exchange model} reflecting its seminal application to hypercubes),  every two matched nodes in round $t$ balance their loads as evenly as possible. This can be expressed by a symmetric $n$ by $n$ matching matrix $\M^{(t)}$, where we,  with slight abuse of notation, use the same symbol for the matching and the corresponding matching matrix. Matrix $\M^{(t)}$ is defined as follows. For every $\{u,v\} \in \M^{(t)}$, we have
$\M_{u,u}^{(t)}:=1/2$, $\M_{v,v}^{(t)}:=1/2$, and $\M_{u,v}^{(t)}=\M_{v,u}^{(t)}:=1/2$. If $u$ is not matched, then $\mathbf{M}^{(t)}_{u,u}=1$, and $\mathbf{M}^{(t)}_{u,v}=0$ for $u \neq v$. We will frequently consider the product of consecutive matching matrices and denote this by $\mathbf{M}^{[t_1,t_2]} = \prod_{s=t_1}^{t_2} \mathbf{M}^{(s)}$ for two rounds $t_1 \leq t_2$. If $t_1 = t_2 + 1$, then $\mathbf{M}^{[t_1,t_2]}=\mathbf{I}$. For any matrix $\mathbf{M}$, let $\M_{u,\cdot}$ be the row of $\M$ corresponding to node $u$ and $\M_{\cdot,u}$ be the column of $\M$ corresponding to node $u$. Note that, in general, the matrix $\mathbf{M}^{[t_1,t_2]}$ may not be symmetric, but as a product of doubly-stochastic matrices, it is doubly-stochastic.

\subsection{Balancing Circuit and Random Matching Model} In the {\em balancing circuit model}, a certain sequence of matchings is applied periodically.
More precisely, let $\mathbf{M}^{(1)},\dots,\mathbf{M}^{(d)}$ be a sequence of $d$ matching matrices\footnote{Traditionally, the variable $d$ is
used for the number of matchings \cite{RSW98,FS09}. There may not exist a direct relation between $d$ and the (maximum) degree of the underlying graph $G$. However, the graph {\em induced} by the union of the $d$ matchings has maximum degree at most $d$.}.
Then in round $t \geq 1$, we apply the matching matrix  \[\mathbf{M}^{(t)} := \mathbf{M}^{(((t-1) \mod d)+1)}.\]
  Following \cite{RSW98}, we define the {\em round matrix} $\M:= \prod_{s=1}^{d} \M^{(s)}$. A natural choice for the $d$ matching matrices is given by an edge coloring of graph $G$. There are various efficient distributed
edge coloring algorithms (see, for example, \cite{PS92,PS97}).

  We notice that the round matrix $\M$ is the product of doubly stochastic matrices and, therefore, also doubly stochastic. Since we further assume that $\M$ is ergodic, the uniform distribution is the (unique) stationary distribution of the associated Markov chain. However, the matrix $\M$ may not be symmetric. As discussed in \cite{RSW98},  we can relate the convergence of $\M$ in this case to that of an associated symmetric matrix $\widetilde{\M}:=\M \cdot \M^{\mathrm{T}}$, the so-called multiplicative reversibilization \cite{F91}. Since $\M$ is ergodic, we have $\lambda(\M) < 1$ in either case.


The alternative to  the balancing circuit model is the \emph{random matching model}, where one generates a random matching in each round. There are several simple and distributed randomized protocols to generate such matchings. For instance, \citet{MG96} analyzed a two-stage protocol for $d$-regular graphs where in the first stage every node picks an incident edge independently with probability $\Theta(1/d)$. In the second stage, we consider the matching formed by all edges that are not incident to any other edge chosen in the first stage. A similar protocol  studied in \cite{BGPS06}  also works for non-regular graphs.
These protocols have two natural properties  sufficient for our analysis. First, for \[
p_{\min} := \min_{t \in \mathbb{N}} \min_{\{u,v\} \in E} \Pro{ \{u,v\} \in \M^{(t)} },\]
 we have $p_{\min} \geq c_{\min} \cdot \frac{1}{\Delta}$ for some constant $c_{\min} > 0$. Second, random matchings generated in different rounds are mutually independent.



\subsection{The Continuous Case}
In the continuous case, load is arbitrarily divisible. Let $\xi^{(0)} \in \mathbb{R}^n$ be the initial load vector, and two matched nodes balance their loads perfectly in every round. It is easy to see that  this process corresponds to a linear system and the load vector $\xi^{(t)},$ $t\geq 1$, can be expressed as
$\xi^{(t)} = \xi^{(t-1)} \,\M^{(t)}$,
which results in $\xi^{(t)} = \xi^{(0)} \, \M^{[1,t]}$. Moreover,
\begin{align*}
\xi_u^{(t)} &=
    \xi_u^{(t-1)} +   \sum_{v \colon \{u,v\} \in E} \left( \xi_v^{(t-1)} \M_{v,u}^{(t)} -
    \xi_u^{(t-1)} \M_{u,v}^{(t)} \right) \\
    & = \xi_u^{(t-1)} +  \hspace{-3pt}  \sum_{v \colon \{u,v\} \in \M^{(t)}} \left( \frac{1}{2} \xi_v^{(t-1)} - \frac{1}{2}
    \xi_u^{(t-1)} \right)
    \ .
\end{align*}
We define the average load by $\overline{\xi}:=\sum_{w \in V}\xi^{(0)}_w/n$, which is invariant of the round $t$. Note that the convergence in the continuous case  depends on the randomly chosen matchings in the random matching model, while it is ``deterministic'' in the balancing circuit model (for any fixed initial load vector).

\begin{defi}Let $G$ be any graph. Fix any pair $(K,\epsilon)$ with $K \geq \epsilon > 0$.
For any pair of integers $0 \leq t_1 < t_2$, we call a time-interval $[t_1,t_2]$ associated with a sequence
of matchings $\langle \M^{(t_1+1)},\ldots,\M^{(t_2)} \rangle$
$(K,\epsilon)$--smoothing if for any $\xi^{(t_1)} \in \mathbb{R}^n$ with
$
\disc\left(\xi^{(t_1)}\right) \leq K  $, the load vector at the end of round $t_2$ satisfies
 $ \disc\left(\xi^{(t_2)}\right) \leq \epsilon$.
\begin{itemize}
 \item For the balancing circuit model, define
\[
  \tau_{\cont}(K,\epsilon) := \min \left\{  t \in \mathbb{N} \colon \mbox{$[0,t]$ is $(K,\epsilon)$--smoothing} \right\},
  \]
 i.e., $\tau_{\cont}(K,\epsilon)$ is the minimum number of rounds in the continuous case to reach discrepancy $\epsilon$ for any initial vector $\xi^{(0)}$ with discrepancy at most $K$.
\item
 For the random matching model, define
\[
                            \tau_{\cont}(K,\epsilon)
:= \min \left\{  t \in \mathbb{N} \colon \Pro{ \mbox{$[0,t]$ is $(K, \eps)$--smoothing}} \geq 1 - n^{-1} \right\},
\]
 i.e., $\tau_{\cont}(K,\epsilon)$ is the minimum number of rounds in the continuous case so that with probability at least $1-n^{-1}$, we reach a discrepancy of $\epsilon$ for any initial vector $\xi^{(0)}$ with discrepancy at most $K$. Note that the probability space is taken over  $t$ randomly chosen matchings $\M^{(1)},\ldots,\M^{(t)}$.
\end{itemize}
\end{defi}

We point out that in both the balancing circuit and the random matching models, $\tau_{\cont}(K,\epsilon)$ is {\em not} a random variable. In addition, by scaling the load vector $\xi^{(t)}$, it is immediate that $\tau_{\cont}(\alpha \cdot K, \alpha \cdot \eps)=\tau_{\cont}(K,\eps)$ for any $\alpha \in \mathbb{R}^{+}$.

The  next lemma provides several results that relate the convergence time $\tau_{\cont}(.,.)$ to the entries of the matrix $\M^{[1,\tau_{\cont}(.,.)]}$.
\begin{lemma}\label{lem:boundbymixinglemma}
Fix any sequence of matchings $\mathcal{M}=\langle \M^{(1)},\M^{(2)},\ldots \rangle$, and consider the continuous case.  Then, the following statements hold:
\begin{itemize}\itemsep 0pt
\item If the time-interval $[0,t]$ is $(K,\epsilon)$--smoothing, then for any non-negative vector $y$ with $\|y\|_1=1$, it holds that  \[\left| \sum_{w \in V} y_w\xi^{(t)}_w - \overline{\xi} \right| \leq \eps.\]
\item If the time-interval $[0,t]$ is $(1,\epsilon)$--smoothing, then for every $u,v \in V$, we have \[\left| \M_{u,v}^{[1,t]} - \frac{1}{n} \right| \leq \epsilon.\]
\item Conversely, if $\left| \M_{u,v}^{[1,t]} - \frac{1}{n} \right| \leq \epsilon$ for every $u,v \in V$, then the time-interval $[0,t]$ is $(1,2 \epsilon n)$-smoothing.
\end{itemize}
\end{lemma}

\begin{proof}
For the first statement, let $\xi^{(0)}$ be any initial load vector with discrepancy at most $K$. Since $[0,t]$ is $(K,\eps)$-smoothing, $\mathrm{disc}(\xi^{(t)})\leq \eps$. Hence, for all nodes $w \in V$,
$
 \left| \xi_{w}^{(t)} - \overline{\xi} \right| \leq \epsilon.
$
Using this along with the triangle inequality, we get
\begin{align*}
 \left| \sum_{w \in V} y_w\xi_w^{(t)} - \overline{\xi} \right| &= \left| \sum_{w \in V} \left( y_w\xi_w^{(t)} - y_w\overline{\xi} \right) \right| \leq  \sum_{w \in V} \left| y_w \xi_w^{(t)} - y_w\overline{\xi} \right| = \sum_{w \in V} y_w \cdot  \left| \xi_{w}^{(t)} - \overline{\xi} \right| \leq \eps.
\end{align*}
For the second statement, let $y$ be the unit-vector that is one at vertex $v$ and $\xi^{(0)}$ be the unit-vector that is one at vertex $u$. Then, $\overline{\xi}=\frac{1}{n}$, and
\[ \xi_{v}^{(t)} = \sum_{w \in V} \xi_{w}^{(0)} \M_{w,v}^{[1,t]} = \M_{u,v}^{[1,t]}.\] Hence,
$ \sum_{w \in V} y_w \xi_w^{(t)} = \xi_v^{(t)} = \M_{u,v}^{[1,t]}$, and by the first statement, $\left| \M_{u,v}^{[1,t]} - \frac{1}{n} \right| \leq \epsilon$.

To prove the third statement, assume that it holds for all nodes $u,v \in V$ that $\left| \M_{v,u}^{[1,t]} - \frac{1}{n} \right| \leq \epsilon$ and satisfies $\xi^{(0)}$ satisfies $\disc(\xi^{(0)}) \leq 1$. Then,
\begin{align*}
 \xi_u^{(t)} &= \sum_{v \in V} \xi_v^{(0)} \cdot \M_{v,u}^{[1,t]} \\
           &= \sum_{v \in V}  \overline{\xi} \cdot \M_{v,u}^{[1,t]} + \sum_{v \in V} (\xi_v^{(0)} - \overline{\xi}) \cdot \M_{v,u}^{[1,t]} \\
           &= \overline{\xi} + \sum_{v \in V}  (\xi_v^{(0)} - \overline{\xi}) \cdot  \frac{1}{n} + \sum_{v \in V} (\xi_v^{(0)} - \overline{\xi}) \cdot \left( \M_{v,u}^{[1,t]} - \frac{1}{n} \right) \\
           &\leq \overline{\xi} + 0 +\sum_{v \in V} 1 \cdot \epsilon \leq \overline{\xi} + \epsilon \cdot n.
\end{align*}
Similarly, $\xi_{u}^{(t)} \geq \overline{\xi} - \epsilon n$. Hence, the discrepancy of $\xi^{(t)}$ is at most $2\epsilon n$.
\end{proof}

 Following previous works \cite{MGS98,RSW98}, we adopt the view that the continuous case ($\tau_{\cont}(K,\epsilon)$) is well-understood, and our goal is to analyze the discrete case. For the balancing circuit model, there is indeed a natural bound on $\tau_{\cont}(K,\epsilon)$ depending on  $\lambda(\M)$. Recall that if $\M$ is not symmetric, then $\lambda(\M)$ is defined in terms of $\tilde{\M}$.

\begin{theorem}[{\cite[Theorem~1]{RSW98}}]\label{thm:circuitcontinuous}
    Consider the balancing circuit model with $d$  matchings $\M^{(1)},\ldots,\M^{(d)}$. Then, for any $\epsilon > 0$, it holds that
   \[
     \tau_{\cont}(K,\epsilon) \leq d \cdot \frac{ 4}{1-\lambda(\M)}\cdot \log\left( \frac{Kn}{\epsilon}  \right).
    \]
\end{theorem}

The next lemma is easily derived from the results in \cite{F91}.
\begin{lemma}\label{lem:wellknown}
Consider the balancing circuit model for $d$ matchings $\M^{(1)},\ldots,\M^{(d)}$.
Let $\M=\prod_{i=1}^{d} \M^{(i)}$. Then, it holds
 for any pair of nodes $u,v \in V$ that
\[
  \left| \M_{u,v}^{t} - \frac{1}{n} \right|
  \leq \left( \lambda(\M) \right)^{t/2}.
\]
\end{lemma}
\begin{proof}
We first look at the case where $\M$ is symmetric. Let $v_1,\ldots v_n$ be orthonormal eigenvectors of the matrix $\M$ with corresponding eigenvalues $\lambda_1,\ldots,\lambda_n$. Let $\xi^{(0)}$ be the unit vector which is one for vertex $u$, and zero, otherwise. We  decompose $\xi^{(0)}$ as $\sum_{i=1}^n \pi_i$, where $\pi_i$ is a constant multiplicity of $v_i$. Since $\xi^{(0)}$ is a probability distribution, we have $\pi_1=(1/n,\ldots, 1/n)$.
Therefore,
\begin{align*}
\left\| \xi^{(1)}-\pi_1\right\|^2_2 &=\left\| \xi^{(0)}\M-\pi_1\right\|^2_2 \\
&= \left\| \pi_1\M+\ldots +\pi_n\M-\pi_1\right\|^2_2 \\
& = \left\| \lambda_2\pi_2+\cdots + \lambda_n\pi_n\right\|^2_2 \\
& \leq \lambda(\M)^2\cdot \left\|\xi^{(0)}-\pi_1\right\|^2_2,
\end{align*}
and it holds by induction that
\[
\left\| \xi^{(t)}-\pi_1\right\|^2_2 \leq \lambda(\M)^{2t}\cdot \left\| \xi^{(0)}-\pi_1\right\|^2_2 \leq \lambda(\M)^t.
\]
Since $\xi_v^{(t)}=\mathbf{M}_{u,v}^{t}$, we have
\[
\left(\xi_v^{(t)} -\frac{1}{n} \right)^2 =
\left(\mathbf{M}_{u,v}^{t} -\frac{1}{n} \right)^2
\leq
\lambda(\M)^t,
\]
and
hence, the claim holds.

For the case of a non-symmetric matrix $\M$, let $\xi^{(0)}$ again be the unit vector which is one for vertex $u$ and zero, otherwise.
By \cite[Eq.~2.11]{F91}, we have that
\begin{align*}
  \sum_{v \in V} \left( \xi_v^{(t)} - \frac{1}{n}  \right)^2 &\leq \left(  \max_{2\leq i\leq n} |\lambda_i(\M)| \right)^{t} \cdot \sum_{v \in V} \left( \xi_v^{(0)} - \frac{1}{n}  \right)^2 \leq \lambda(\M)^{t}.
\end{align*}
Since $\xi_v^{(t)}=\mathbf{M}_{u,v}^{t}$, the claim follows. Combining the two cases completes the proof.
\end{proof}

For matchings with $\lambda(\M)$ close to $1$, the following lemma gives a better bound.

\begin{lemma}\label{lem:circuitlocalexpansion}
Consider the balancing circuit model with $d=\Oh(1)$. Then, for any $\sigma$ with $\frac{2}{n} \leq \sigma^{-1} \leq 1$, there is a constant $c=128 d^2$ such that for any $u,v \in V$, \begin{align*}
  \bM_{u,v}^{c \sigma^{4}} \leq \frac{1}{n} + \sigma^{-1/2}.
\end{align*}
\end{lemma}
\begin{proof}
Fix any node $u \in V$. We first look at the case where there is no number $s$ satisfying $\| \bM_{u, \cdot}^{s} \|_2^2 \geq \sigma^{-1}$. Then, for any node $v$ we have that
\[ \left(\M_{u,v}^{c\sigma^4} -\frac{1}{n}\right)^2\leq\left\|\M_{u,.}^{c\sigma^4} \right\|^2_2\leq \sigma^{-1},
\]
which implies that $\M^{c\sigma^4}_{u,v}\leq \frac{1}{n}+\sigma^{-1/2}$.

Next, assume that $s$ is  any number for which $\| \bM_{u, \cdot}^{s} \|_2^2 \geq \sigma^{-1}$.  Then,
\[
  \left\| \bM_{u,\cdot}^{s}   \right\|_{\infty} \geq \frac{\left\| \bM_{u,\cdot}^{s}   \right\|_{2}^2}{ \left\|\bM_{u,\cdot}^{s}   \right\|_{1}   } \geq \sigma^{-1}.
\]
Let $v$ be any node with $\M_{u,v}^{s} \geq \sigma^{-1}$.
Since $\frac{1}{2} \sigma^{-1} \geq \frac{1}{n}$, there is at least one node $w \neq v$ with $\M_{u,w}^{s} \leq \frac{1}{2} \sigma^{-1}$.
Next, define $G'$ as the undirected (multi-)graph formed by the union of the $d$ matchings $\M^{(1)},\M^{(2)},\ldots,\M^{(d)}$. Since the round matrix $\M$ is ergodic, the graph $G'$ is connected.
Hence, there is a shortest path $P=(u_1=v,\ldots,u_{\ell}=w)$ from $v$ to such a node $w \in V$. In particular, we can find a node $w \in V$ (and associated path $P$) with the additional property that $w$ is the {\em only} node on $P$ with $\M_{u,w}^{s} \leq \frac{1}{2} \sigma^{-1}$. Therefore, $\ell-1 \leq 1 / (\frac{1}{2} \sigma^{-1}) = 2 \sigma$. Consequently, there must be at least one edge $\{f,g\} \in E(G')$ along the path $P$ so that
 \[
   \bM_{u,f}^{s} - \bM_{u,g}^{s} \geq
   \frac{ \bM_{u,v}^{s} - \bM_{u,w}^{s}}{ \ell - 1 } \geq
   \frac{ \frac{1}{2} \sigma^{-1} }{ 2 \sigma } = \frac{1}{4} \sigma^{-2}.
 \]
Let $\M^{(j)}$ with $1 \leq j \leq d$ be the first matching to include the edge $\{f,g\} \in E(G)$.
Now, observe that in at least one of the rounds $t$ with $d \cdot s + 1 \leq t \leq d \cdot s + j$, node $f$ sends a load amount of at least $\frac{1}{16d} \sigma^{-2}$ to one of its neighbors, or node $g$ receives a load amount of
at least $\frac{1}{16d} \sigma^{-2}$ from one of its neighbors.
(To see this, assume that the statement does not hold in all rounds $t$ with $d\cdot s+1\leq t\leq d\cdot s+j-1$. ) Then, the load of $f$ at the end of  round $d \cdot s+j-1$ would be at least
\[
\bM_{u,f}^{[1,d \cdot s+j-1]} \geq \bM_{u,f}^{[1,d \cdot s]} - (j-1) \cdot \frac{1}{16d} \sigma^{-2} \geq \bM_{u,f}^{[1,d \cdot s]} - \frac{1}{16} \sigma^{-2},
\]
and similarly, $\bM_{u,g}^{[1,d \cdot s+j-1]} \leq \bM_{u,g}^{[1,d \cdot s]} + \frac{1}{16} \sigma^{-2}$, which would imply
\[
\bM_{u,f}^{[1,d \cdot s+j-1]} - \bM_{u,g}^{[1,d \cdot s+j-1]} \geq \frac{1}{4} \sigma^{-2} - 2 \cdot \frac{1}{16} \sigma^{-2} = \frac{1}{8} \sigma^{-2}.
\]
 Hence in round $d \cdot s+ j$ node $f$ would send a load amount of at least
\[
 \frac{1}{2} \cdot \left( \bM_{u,f}^{[1,d \cdot s+j-1]} - \bM_{u,g}^{[1,d \cdot s+j-1]} \right) \geq \frac{1}{16} \sigma^{-2}
 \]
 to node $g$). Hence, there exists a pair of nodes $p$ and $q$ (noticing that at least one of these nodes is $f$ or $g$), which satisfies
 \[
 \bM_{u,p}^{[1,t-1]} - \bM_{u,q}^{[1,t-1]} \geq \frac{1}{8d} \cdot \sigma^{-2}.
 \]
Since $\left\| \M_{u,\cdot}^{[1,s]} \right\|_2^2$ is non-increasing in $s$, we conclude that
\begin{align}
\left\| \bM_{u,\cdot}^{[1,d \cdot s +d]} \right\|_{2}^{2} &\leq
\left\| \bM_{u,\cdot}^{[1,t]} \right\|_{2}^{2} \notag \\
&\leq \sum_{k \in V, k \not\in \{p,q\} } \left( \bM_{u,k}^{[1,t-1]} \right)^2 + 2 \cdot \left( \frac{ \bM_{u,p}^{[1,t-1]} + \bM_{u,q}^{[1,t-1]}}{2} \right)^2   \notag \\
 &=\left\| \bM_{u,.}^{[1,t-1]}  \right\|_{2}^2
 - \left( \bM_{u,p}^{[1,t-1]} \right)^2 - \left(  \bM_{u,q}^{[1,t-1]} \right)^2 + \frac{1}{2} \left( \bM_{u,p}^{[1,t-1]}
 + \bM_{u,q}^{[1,t-1]} \right)^2  \notag \\
&= \left\| \bM_{u,.}^{[1,t-1]}  \right\|_{2}^2
-  \frac{1}{2} \left( \bM_{u,p}^{[1,t-1]} - \bM_{u,q}^{[1,t-1]} \right)^2  \notag \\
&\leq \left\| \bM_{u,.}^{[1,t-1]}  \right\|_{2}^2
-  \frac{1}{2} \left( \frac{1}{8d} \sigma^{-2} \right)^2  \notag.
 \end{align}
Hence,
\begin{align*}
  \left\| \M_{u,\cdot}^{[1,d \cdot s + d]} \right\|_2^2 \leq \left\| \M_{u,\cdot}^{[1,t-1]} \right\|_2^2 - \frac{1}{2} \left( \frac{1}{8d} \sigma^{-2}   \right)^{2} \leq \left\| \M_{u,\cdot}^{[1,d \cdot s]} \right\|_2^2 - \frac{1}{128 d^2} \sigma^{-4}.
\end{align*}
To summarize,  as long as $\| \M_{u,\cdot}^{s} \|_2^2 \geq \sigma^{-1}$, it holds that
\[
  \left\| \M_{u,\cdot}^{s+1} \right\|_2^2 \leq \left\| \M_{u,\cdot}^{s} \right\|_2^2 - \frac{1}{128 d^2} \sigma^{-4}.
\]
Hence, for $s = 128 d^2 \cdot \sigma^4$, it holds that
\[
 \left\| \M_{u,\cdot}^{s} \right\|_2^2 \leq \sigma^{-1}.
\]
This implies that for any $v \in V$,
$
  \left( \M_{u,v}^{s} - \frac{1}{n} \right)^2 \leq \sigma^{-1},
$
and the claim follows.
\end{proof}

For the random matching model, the convergence depends  on $p_{\min}$ and the spectral gap of the {\em diffusion matrix} $\P$,
defined as $\P_{u,v} := \frac{1}{2 \Delta}$ if $\{u,v\} \in E$,
$\bP_{u,v} := 1 - \frac{\deg(u)}{2 \Delta} $ if $v=u$, and $\bP_{u,v}:=0$, otherwise. Note that $\lambda(\bP)=\lambda_2(\bP)$.

\begin{theorem}[{cf.~\cite[Theorem~1]{MG96}}]\label{thm:muthu}
Let $G$ be any graph with maximum degree $\Delta$. Consider the random matching model in the continuous case,
and let $p_{\min}$ be the minimum probability for an edge to be included in the matching. Let $\xi^{(0)}$ be an arbitrary initial load vector, and define the quadratic potential function as
$
  \Phi^{(t)} := \sum_{u \in V} \big( \xi_u^{(t)} - \overline{\xi} \big)^2.
$
Then, for any round $t$, we have
\[
 \Ex{ \Phi^{(t)}   } \leq \left(1 -  \Delta \cdot p_{\min} \cdot (1-\lambda_2(\P))   \right)^{t} \cdot \Phi^{(0)}.
\]
\end{theorem}

\begin{proof}
The proof essentially follows that of \cite[Theorem~1]{MG96} and is given for the sake of completeness. By definition, we have that
\begin{align*}
\Phi^{(t-1)} - \Phi^{(t)} & = \sum_{ u\in V} \left[
\left( \xi_u^{(t-1)} -\bar{\xi} \right)^2 -
\left( \xi_u^{(t)} -\bar{\xi} \right)^2
\right]   \\
& = \sum_{ \{u,v\}\in \M^{(t)}} \left[
\left( \xi_u^{(t-1)} -\bar{\xi} \right)^2 +
\left( \xi_v^{(t-1)} -\bar{\xi} \right)^2 -
\left( \xi_u^{(t)} -\bar{\xi} \right)^2 -\left( \xi_u^{(t)} -\bar{\xi} \right)^2
\right],
\end{align*}
since $\xi^{(t)}_u=\xi^{(t-1)}_u$ for every unmatched node $u$ in round $t$. Notice that for $\{u,v\}\in \M^{(t)}$, it holds that $\xi^{(t)}_u=\xi^{(t)}_v = \left( \xi^{(t-1)}_u + \xi^{(t-1)}_v \right)\big/2$, and hence,
\[
\left( \xi_u^{(t-1)} -\bar{\xi} \right)^2 +
\left( \xi_v^{(t-1)} -\bar{\xi} \right)^2 -
\left( \xi_u^{(t)} -\bar{\xi} \right)^2 -\left( \xi_u^{(t)} -\bar{\xi} \right)^2= \frac{1}{2}\cdot\left( \xi^{(t-1)}_u -\xi^{(t-1)}_v \right)^2.
\]
Therefore,
\begin{align*}
\Ex{ \Phi^{(t-1)} -\Phi^{(t)} } &= \sum_{\{u,v\}\in E(G)} \Pro{ \{u,v\} \in \M^{(t)} }\cdot\frac{1}{2}\cdot\left( \xi^{(t-1)}_u -\xi^{(t-1)}_v \right)^2\\
&\geq \frac{p_{\min}}{2}\cdot \sum_{\{u,v\}\in E(G)}\left( \xi^{(t-1)}_u -\xi^{(t-1)}_v \right)^2.
\end{align*}
Notice that we may assume for the remainder of the proof that $\xi^{(t-1)} \neq \mathbf{1} \cdot \bar{\xi}$, for otherwise the claim of the lemma holds trivially since $\Phi^{(t)}=\Phi^{(t-1)}=0$. Hence, $\Phi^{(t-1)} \neq 0$, and
\begin{align*}
 \Ex{ \frac{ \Phi^{(t-1)}-\Phi^{(t)} }{\Phi^{(t-1)} } ~\Big|~ \xi^{(t-1)}} & = \frac{\Ex{ \Phi^{(t-1)}-\Phi^{(t)}}}{\Phi^{(t-1)} }\\
 &\geq \frac{\frac{p_{\min}}{2}\cdot \sum_{\{u,v\}\in E(G)}\left( \xi^{(t-1)}_u -\xi^{(t-1)}_v \right)^2}{\sum_{u \in V} \big( \xi_u^{(t-1)} - \overline{\xi} \big)^2}.
 \end{align*}
By introducing the vector $y\in \mathbb{R}^n$ where $y_u=
\xi_u^{(t-1)}-\bar{\xi}$, we can rewrite the formula above as
\[\Ex{ \frac{ \Phi^{(t-1)}-\Phi^{(t)} }{\Phi^{(t-1)} } ~\Big|~ \xi^{(t-1)}} \geq \frac{p_{\min}}{2}\cdot
\frac{y\mathbf{L}y^{\mathrm{T}}}{y\cdot y^{\mathrm{T}}},
\]
where $\mathbf{L}$ is the Laplacian matrix of $G$ defined by $\mathbf{L}_{u,u} = \deg(u)$, $\mathbf{L}_{u,v} = -1$ for $\{u,v\} \in E(G)$, and $\mathbf{L}_{u,v}=0$, otherwise. Since $\xi^{(t-1)} \neq \mathbf{1} \cdot \bar{\xi}$, we have $y \neq \mathbf{0}, y\perp\mathbf{1}$, and the Min-Max characterization of eigenvalues yields\[\Ex{ \frac{ \Phi^{(t-1)}-\Phi^{(t)} }{\Phi^{(t-1)} } ~\Big|~ \xi^{(t-1)}} \geq \frac{p_{\min}}{2}\cdot
\lambda_{n-1}(\mathbf{L}).
\]
Since  $\mathbf{P} = \mathbf{I} - \frac{1}{2 \Delta} \cdot \mathbf{L}$, we have that  $\lambda_{2}(\P)=1 - \frac{\lambda_{n-1}(\mathbf{L})}{2 \Delta}$ and $\lambda_{n-1}(\mathbf{L}) = 2 \Delta \cdot (1-\lambda_2(\P))$. This implies
\begin{align*}
  \Ex{ \Phi^{(t)} ~\big|~ \xi^{(t-1)}} &\leq \left(1 - \frac{\lambda_{n-1}(\mathbf{L})}{2} \cdot p_{\min} \right) \cdot \Phi^{(t-1)} = \left( 1 - \Delta \cdot p_{\min} \cdot \left(1 - \lambda_2(\mathbf{P}) \right) \right) \cdot \Phi^{(t-1)},
\end{align*}
and the claim follows by induction on $t$.
\end{proof}

\thmref{muthu} implies the following corollary.
\begin{corollary}\label{cor:muthu}
Consider the random matching model. Then,
for any node $v \in V$ and round $t$, we have
\[
  \Pro{ \bigcap_{u \in V} \left\{ \left| \M_{v,u}^{[1,t]} - \frac{1}{n} \right| \leq  \ce^{- c\cdot (1-\lambda_2(\P) )\cdot t} \right\} } \geq 1 - \ce^{-c\cdot(1-\lambda_2(\P) )\cdot t},
\]
where $c > 0$ is the constant given by $c=\frac{p_{\min} \cdot \Delta}{4}$.
\end{corollary}
\begin{proof}
Let $\xi^{(0)}$ be the initial load vector with $\xi_v^{(0)}:=1$ and $\xi_{u}^{(0)}:=0$ for $u \neq v$. Then, it holds that
\[
  \xi_u^{(t)} = \sum_{w \in V} \xi_w^{(0)} \cdot \M_{w,u}^{[1,t]} = 1 \cdot \M_{v,u}^{[1,t]},
\]
and $\overline{\xi}=1/n$.
Define the quadratic potential function as
$
  \Phi^{(t)} := \sum_{u \in V} \big( \xi_u^{(t)} - \overline{\xi} \big)^2.
$
By \thmref{muthu},
\[
\Ex{ \Phi^{(t)}   } \leq \left(1 - \Delta \cdot p_{\min} \cdot (1-\lambda_2(\P))   \right)^{t} \cdot \Phi^{(0)}.
\]
Define $c':=p_{\min} \cdot \Delta$, which is a positive constant since $p_{\min} \geq c_{\min} / \Delta$. Using
\[
\Phi^{(0)} = \left( 1-\frac{1}{n} \right)^2+ (n-1)\cdot \left(\frac{1}{n}\right)^2 =  1 - \frac{1}{n} \leq  1,
\]
 we obtain
\[
\Ex{ \Phi^{(t)}   } \leq \left(1 - c'\cdot (1-\lambda_2(\P)) \right)^{t} \cdot 1 \leq \ce^{- c' \cdot (1-\lambda_2(\P) )\cdot t  }.
\]
Hence, by Markov's inequality, it holds that
\[
\Pro{ \Phi^{(t)} \geq \ce^{-c'/2 \cdot (1-\lambda_2(\P) )\cdot t  } } \leq \ce^{- c'/2 \cdot (1-\lambda_2(\P) )\cdot t }.
\]
Assuming that the event $\Phi^{(t)} \leq \ce^{- c'/2 \cdot (1-\lambda_2(\P) )\cdot t  }$ occurs, it holds for any node $u \in V$ that
\[
 \left( \M_{v,u}^{[1,t]} - \frac{1}{n} \right)^2 = \left( \xi_u^{(t)} - \frac{1}{n}  \right)^2 \leq \ce^{- c'/2 \cdot (1-\lambda_2(\P) )\cdot t  },
\]
which implies that
\[
 \left| \M_{v,u}^{[1,t]} - \frac{1}{n} \right| \leq  \ce^{- c'/4 \cdot (1-\lambda_2(\P))\cdot t }.
\]
By setting $c:=c'/4$, we finish the proof.
\end{proof}

For graphs with $\lambda(\P)$ close to $1$, the following lemma gives a better bound.
\begin{lemma}\label{lem:matchinglocalexpansion}
Let $G$ be any $d$-regular graph and consider the random matching model. Fix any node $u \in V$. Then, for any $\sigma$ with $\frac{2}{n} \leq \sigma^{-1} \leq 1$, there is a constant $c>0$ such that
\begin{align*}
  \Pro{ \left\| \bM_{u,\cdot}^{[1,c \sigma^{6}]}  \right\|_{2}^2 \geq \sigma^{-1} } \leq \ce^{- \sigma^{4}}.
\end{align*}
Furthermore, it holds that
\begin{align*}
  \Pro{ \left\| \bM_{\cdot,u}^{[1,c \sigma^{6}]}  \right\|_{2}^2 \geq \sigma^{-1} } \leq \ce^{- \sigma^{4}}.
\end{align*}
\end{lemma}
\begin{proof}
The proof is very similar to the proof of \lemref{circuitlocalexpansion}. Again,
consider a round $s$ for which $\left\| \bM_{u,\cdot}^{[1,s]}  \right\|_{2}^{2} \geq \sigma^{-1}$ holds, and we have
\[
  \left\| \bM_{u,\cdot}^{[1,s]}   \right\|_{\infty} \geq \frac{\left\| \bM_{u,\cdot}^{[1,s]}   \right\|_{2}^2}{ \left\|\bM_{u,\cdot}^{[1,s]}   \right\|_{1}   } \geq \sigma^{-1}.
\]
Let $v$ be any node with
$ \bM_{u,v}^{[1,s]} \geq \sigma^{-1}$. We continue with a case distinction on the degree of $G$.

\medskip
\textbf{Case 1:}
  $d \geq 4 \sigma^{2}$. Since $\sum_{w \in N(v)} \bM_{u,w}^{[1,s]} \leq 1$, at least $d/2 \geq 2 \sigma^2$ of the neighbors $w \in N(v)$ satisfy $\bM_{u,w}^{[1,s]} \leq \frac{1}{2} \sigma^{-2}$.
Hence, an edge $\{v,w\}$ with $\bM_{u,w}^{[1,s]} \leq \frac{1}{2} \sigma^{-2}$ is included in the random matching in round $s+1$
  with probability at least
  \[
  \frac{d}{2} \cdot p_{\min} \geq \frac{d}{2} \cdot c_{\min} \cdot \frac{1}{d} = \frac{1}{2} \cdot c_{\min} > 0.
\]
Assuming that this event occurs,
\begin{align}
\left\| \bM_{u,\cdot}^{[1,s+1]} \right\|_{2}^{2} &\leq
 \sum_{k \in V, k \not\in \{v,w\} } \left( \bM_{u,k}^{[1,s]} \right)^2 + 2 \cdot \left( \frac{ \bM_{u,v}^{[1,s]} + \bM_{u,w}^{[1,s]}}{2} \right)^2   \notag \\
 &=\left\| \bM_{u,.}^{[1,s]}  \right\|_{2}^2
 - \left( \bM_{u,v}^{[1,s]} \right)^2 - \left(  \bM_{u,w}^{[1,s]} \right)^2 + \frac{1}{2} \left( \bM_{u,v}^{[1,s]}
 + \bM_{u,w}^{[1,s]} \right)^2  \notag \\
&= \left\|  \bM_{u,.}^{[1,s]}  \right\|_{2}^2
-  \frac{1}{2} \left( \bM_{u,v}^{[1,s]} - \bM_{u,w}^{[1,s]} \right)^2 \label{eq:improvement} \\
&\leq \left\|  \bM_{u,.}^{[1,s]}  \right\|_{2}^2
-  \frac{1}{2} \left( \sigma^{-1} - \frac{1}{2} \sigma^{-2} \right)^2  \notag \leq \left\|  \bM_{u,.}^{[1,s]}  \right\|_{2}^2
- \frac{1}{8} \sigma^{-2}. \notag
 \end{align}

\textbf{Case 2.} $d \leq 4 \sigma^{2}$.
Since $\frac{1}{2}\sigma^{-1}\geq\frac{1}{n}$ by the assumption on $\sigma$, there is at least one node $w \neq v$ with $\M_{u,w}^{[1,s]}\leq \frac{1}{2}\sigma^{-1}$. Consider now a shortest path $P=(u_1=v,\ldots,u_{\ell}=w)$ from $v$ to such a node $w$ with the property that $w$ is the \emph{only} node on $P$ with $\bM_{u,w}^{[1,s]} \leq \frac{1}{2} \sigma^{-1}$. By construction of $P$, the length $\ell$ satisfies $\ell - 1 \leq 1/(\frac{1}{2} \sigma^{-1}) = 2 \sigma$. Hence, there must be at least one edge $\{f,g\} \in E$ along the path $P$ so that
 \[
   \bM_{u,f}^{[1,s]} - \bM_{u,g}^{[1,s]} \geq
   \frac{ \bM_{u,v}^{[1,s]} - \bM_{u,w}^{[1,s]}				 }{ \ell - 1 } \geq
   \frac{ \frac{1}{2} \sigma^{-1} }{ 2 \sigma } = \frac{1}{4} \sigma^{-2}.
 \]
Recall that the edge $\{f,g\}$ is included in the matching in round $s+1$ with probability at least $p_{\min}=c_{\min} \cdot \frac{1}{d} \geq c_{\min} \cdot \frac{1}{4\sigma^2}$.  If the edge $\{f,g\}$ is part of the matching, then we conclude from \eq{improvement} that
 \begin{align*}
  \left\| \bM_{u,\cdot}^{[1,s+1]} \right\|_{2}^{2} &\leq \left\| \bM_{u,\cdot}^{[1,s]} \right\|_{2}^{2} - \frac{1}{32} \sigma^{-4}.
 \end{align*}

Summarizing both cases, we can upper bound the minimum round $\tau$ when $\left\| \bM_{u,\cdot}^{[1,\tau]} \right\|_{2}^{2} \leq \sigma^{-1}$ occurs by the sum of $32 \sigma^{4}$ independent random geometric variables with success probability $p_{\min}$ each. Using \lemref{chernoffgeo}, we obtain that the sum of  these geometric variables is larger than $c\cdot \sigma^{6}$ with probability at most $  \ce^{-\sigma^{4} }
$, if $c>0$ is a sufficiently large constant.

To prove the second statement, we use a symmetry argument.
Recall that
\[
\M^{[1,c \sigma^6]} = \prod_{s=1}^{c \sigma^6} \M^{(s)},
\]
 where $\M^{(s)}$ is a matching chosen randomly according to some distribution $Z^{(s)}$. Consider now the matrix $\P^{[1,c \sigma^6]}=\prod_{s=1}^{c \sigma^6} \P^{(s)}$, where $\P^{(s)}$ is a matching chosen randomly according to some distribution $Z^{(c \sigma^6-s+1)}$. Using the same analysis as above for $\P$ instead of $\M$, we conclude that
\[
  \Pro{ \left\| \P_{u,\cdot}^{[1,c \sigma^6]} \right\|_2^2 \geq \sigma^{-1} } \leq \ce^{-\sigma^{4} }.
\]
Note that we can couple $\P$ and $\M$ so that
\[
\P^{[1,c \sigma^6]} = \prod_{s=1}^{c \sigma^6} \M^{(c \sigma^6-s+1)}.
\]
Taking the transpose of both sides yields
\begin{align*}
 \left( \P^{[1,c \sigma^6]}\right)^{\mathrm{T}} &= \left(
 \prod_{s=1}^{c \sigma^6} \M^{(c \sigma^6-s+1)} \right)^{\mathrm{T}} =
 \prod_{s=1}^{c \sigma^6} \left( \M^{(s)} \right)^{\mathrm{T}} =  \prod_{s=1}^{c \sigma^6}  \M^{(s)} = \M^{[1,c \sigma^6]},
\end{align*}
and the second statement follows.
\end{proof}

\thmref{muthu} also implies the following upper bound on $\tau_{\cont}(K,\epsilon)$.

\begin{theorem}\label{thm:contmatching}
Let $G$ be any graph with the maximum degree $\Delta$, and consider the random matching model.
Then, for any $\epsilon > 0$, it holds that
\[
 \tau_{\cont}(K,\epsilon) \leq \frac{8}{\Delta \cdot p_{\min}} \cdot \frac{1}{1-\lambda_2(\P)}\cdot\log \left( \frac{2 K n}{\epsilon}  \right) .
\]
\end{theorem}

Hence, for $p_{\min}=\Theta(1/\Delta)$, we obtain essentially the same convergence as for the first-order-diffusion scheme (cmp.~\thmref{continuous}), although the communication is restricted to a single matching in each round. A similar result for the $\ell_2$-norm can be found in~\cite[Theorem~5]{BGPS06}. We now prove \thmref{contmatching}.

\begin{proof}
We apply \corref{muthu} with
\[
t=\frac{1}{c} \cdot \frac{ \log(n^2) + \log(2Kn/{\epsilon})}{1-\lambda_2(\P)}
\]
 for all nodes $v \in V$ to conclude that
\[
 \Pro{ \bigcap_{v \in V} \bigcap_{u \in V} \left\{ \left| \M_{v,u}^{[1,t]} - \frac{1}{n} \right| \leq \frac{\epsilon}{2 K n} \right\} } \geq 1 - n \cdot n^{-2} = 1 - n^{-1}.
\]
Assuming that \[\left| \M_{v,u}^{[1,t]} - \frac{1}{n} \right| \leq \frac{\epsilon}{2 K n}\]  for all nodes $u,v \in V$, the third statement of \lemref{boundbymixinglemma} implies that  the time-interval $[0,t]$ is $(1, \frac{\epsilon}{K})$-smoothing, or equivalently, $(K,\epsilon)$-smoothing.
\end{proof}

We now provide an almost matching lower bound. For the sake of simplicity, we focus on $d$-regular graphs and consider the following distributed protocol for generating a matching in each round (see also~\cite{BGPS06}). First, each node becomes active or passive with probability $1/2$. Then, every active node chooses a random neighbor. If this neighbor is passive and receives only one request from an active neighbor, then the two nodes are matched.
\begin{theorem}\label{thm:continuouslowermatching}
Let $G$ be any $d$-regular graph, and consider the random matching model using the aforementioned procedure for generating matchings. For any $K \geq 8n$ and $\epsilon > 0$, there is an initial load vector with discrepancy $K$ so that after
$\frac{1}{4} \cdot \frac{1}{1 - \lambda(\P)} \cdot \log \left( \frac{K}{8 \epsilon n} \right) -1$ rounds, the expected discrepancy is at least $\epsilon$.
\end{theorem}
A similar lower bound was derived in \cite[Corollary~3]{BGPS06}, but that lower bound holds only with small probability (depending on $\epsilon/K$).
\begin{proof}
The lower bound will be established by proving that the expectation of the load of a vertex is equal to the load of a vertex in a corresponding diffusion-based load balancing process. First, we observe that for any round $t$ and any edge $\{u,v\} \in E(G)$,
\[
  \Pro{ \{u,v\} \in \M^{(t)} } = 2 \cdot \frac{1}{2} \cdot \frac{1}{2} \cdot \frac{1}{d} \cdot \left(1 - \frac{1}{2} \cdot \frac{1}{d} \right)^{d-1} := p;
\]
in particular, every edge is included with the same probability. Consider now a random walk on $G$ defined as follows. If the random walk is at a vertex $u$ at the end of round $t-1$, it moves to a neighbor $v \in N(u)$ with $\{u,v\} \in \M^{(t)}$ with probability $1/2$; if there exists no such neighbor, then the random walk stays at vertex $u$. This process is the same as a random walk that with probability $\frac{1}{2} \cdot d \cdot p$ decides to move to a neighbor, which is then chosen uniformly at random among all $d$ neighbors, and otherwise stays at $u$.
The transition matrix of this random walk is given by \[
\mathbf{Q} = \left(1 - \frac{1}{2} \cdot d \cdot p \right) \cdot \mathbf{I} + \frac{1}{2} \cdot p \cdot \mathbf{A},
\]
where $\mathbf{A}$ is the adjacency matrix of graph $G$.
Since $\frac{1}{2} \cdot d \cdot p \leq \frac{1}{4}$, it follows that
\begin{align*}
\lambda_n(\mathbf{Q}) &= \left(1-\frac{1}{2}\cdot d\cdot p\right)\cdot \lambda_n(\mathbf{I}) + \frac{1}{2}\cdot p\cdot\lambda_n(\mathbf{A})  \geq \frac{3}{4} + \frac{1}{4d}\cdot (-d)  = \frac{1}{2},
\end{align*}
and $\lambda(\mathbf{Q}) \geq \frac{1}{2}$.
Further, recall that the diffusion matrix $\mathbf{P}$ is defined by
\[
  \mathbf{P} = \frac{1}{2} \cdot \mathbf{I} + \frac{1}{2} \cdot \frac{1}{d} \cdot \mathbf{A}.
\]
 Since $\frac{1}{2} \cdot p \leq \frac{1}{4}$, we conclude that $\lambda_{i}(\mathbf{P}) \leq \lambda_{i}(\mathbf{Q})$ for every $1 \leq i \leq n$, and hence,
 \[
 \lambda(\mathbf{P}) = \lambda_2(\mathbf{P}) \leq \lambda_2(\mathbf{Q}) = \lambda(\mathbf{Q}).
 \]

%
Let us now fix a vertex $w \in V$. Our next claim is that if we define the initial load vector $\xi^{(0)}$ as $\xi_u^{(0)}=1$ for $u=w$ and $\xi_u^{(0)}=0$, otherwise, then
\begin{align}
  \Ex{ \xi_v^{(t)}  } = \mathbf{Q}_{w,v}^{t}. \label{eq:claimm}
\end{align}
We prove this equality by induction on $t$.
Since $\mathbf{Q}^{(0)}=\mathbf{I}$, the claim holds for $t=0$. For $t \geq 1$ and fixed load vector $\xi^{(t-1)}$, we conclude by using the induction hypothesis that
\begin{align*}
 \Ex{ \xi_v^{(t)}} &= \Ex{ \Ex{ \xi_v^{(t)} \, \Big| \, \xi^{(t-1)} }} \\ &= \Ex{
 \xi_v^{(t-1)} + \sum_{u \in N(v)} \frac{\xi_u^{(t-1)}-\xi_v^{(t-1)} }{2} \cdot p } \\
 &= \mathbf{Q}_{w,v}^{t-1} + \sum_{u \in N(v)} \frac{\mathbf{Q}_{w,u}^{t-1}-\mathbf{Q}_{w,v}^{t-1} }{2} \cdot p \\
 &= \left(1 - \frac{1}{2} \cdot d \cdot p \right) \cdot \mathbf{Q}_{w,v}^{t-1} + \sum_{u \in N(v)} \frac{1}{2} \cdot p \cdot \mathbf{Q}_{w,u}^{t-1} \\
 &=  \mathbf{Q}_{w,v}^{t-1} \cdot \mathbf{Q}_{v,v} + \sum_{u \in N(v)} \mathbf{Q}_{w,u}^{t-1} \cdot \mathbf{Q}_{u,v} = \mathbf{Q}_{w,v}^{t},
\end{align*}
which completes the induction step and the proof of  \eq{claimm}.

We now introduce some additional notation. For a Markov chain with doubly-stochastic transition matrix $\mathbf{Q}$, we define the {\em variation distance} at time $t$ with the initial vertex $w$ as
\[
  \Delta_w(t) = \frac{1}{2} \sum_{v \in V} \left| \mathbf{Q}^{t}_{w,v} - \frac{1}{n} \right|.
\]
For any vertex $u \in V$ and $\tilde{\epsilon} \in (0,1)$, the rate of convergence is defined by
\[
  \tau_{w}(\tilde{\epsilon}) := \min \left\{ t \colon \Delta_w(t') \leq \tilde{\epsilon} \mbox{ for all $t' \geq t$}  \right\} =  \min \left\{ t \colon \Delta_w(t) \leq \tilde{\epsilon}  \right\},
\]
where the last equality comes from the fact that the variation distance is non-increasing in time (cf.~\cite[Chapter~11]{MU05}). By
\cite[Proposition~1, Part (ii)]{Si92}, it holds that
\[
  \tau := \max_{w \in V} \tau_{w}(\tilde{\epsilon}) \geq \frac{1}{2} \cdot \frac{\lambda(\mathbf{Q})}{1-\lambda(\mathbf{Q})} \cdot \log \left( \frac{1}{2 \tilde{\epsilon}} \right).
\]
As derived earlier, $\lambda(\mathbf{Q}) \geq \frac{1}{2}$ and $\lambda(\mathbf{Q})\geq \lambda(\mathbf{P})$, and, therefore, we have that \begin{align*}
 \tau &\geq \frac{1}{4} \cdot \frac{1}{1 - \lambda(\P)} \cdot \log \left( \frac{1}{2 \tilde{\epsilon}} \right).
\end{align*}
Let $\tau':=\frac{1}{4} \cdot \frac{1}{1 - \lambda(\P)} \cdot \log \left( \frac{1}{2 \tilde{\epsilon}} \right)$. Hence, we have  in $\tau'-1$ iterations  that
\[
  \frac{1}{2} \sum_{v \in V} \left| \mathbf{Q}^{\tau'-1}_{w,v} - \frac{1}{n} \right| > \tilde{\epsilon},
\]
which implies that there is a vertex $v \in V$ with
\begin{align}
   \mathbf{Q}^{\tau'-1}_{w,v} - \frac{1}{n}  > \frac{\tilde{\epsilon}}{4 n}. \label{eq:lb}
\end{align}
Since in every round there is at least one vertex with load at most $\frac{1}{n}$,
\begin{align*}
 \Ex{ \disc( \xi^{(\tau'-1)}) } &\geq \Ex{ \xi_v^{(\tau'-1)} - \frac{1}{n} } = \mathbf{Q}^{\tau'-1}_{w,v} - \frac{1}{n} \geq \frac{\tilde{\epsilon}}{4n}.
\end{align*}
Hence, we have shown that for any given $\tilde{\epsilon}$, there are load vectors with initial discrepancy $1$ so that the expected discrepancy after $\tau'-1$ rounds is at least $\frac{\tilde{\epsilon}}{4n}$. Scaling the load vector by a factor of $K$ and replacing $\tilde{\epsilon}$ by $\epsilon \cdot \frac{4n}{K}$, we conclude that there are load vectors with initial discrepancy $K$ so that the expected discrepancy after $\tau'-1$ rounds is at least $\epsilon$.
\end{proof}

\subsection{The Discrete Case}
Let us now turn to the discrete case with indivisible unit-size tokens. Let $x^{(0)} \in \mathbb{Z}^n$ be the initial load vector with average load $\overline{x}:=\sum_{w \in V} x_w^{(0)}/n$
 and $x^{(t)}$ be the load vector at the end of round $t$. If the sum of tokens of two matched nodes is odd, we have to decide which of the matched  nodes should get the excess token. To this end, we employ the so-called {\em random orientation}~(\cite{RSW98,FS09}), in the spirit of randomized rounding. More precisely, for any two matched nodes $u$ and $v$ in round $t$, node $u$ gets $\Big\lceil \frac{x^{(t-1)}_u+x^{(t-1)}_v}{2} \Big\rceil$ or $\Big\lfloor \frac{x^{(t-1)}_u+x^{(t-1)}_v}{2} \Big\rfloor $
tokens, with probability $1/2$ each. The remaining tokens are assigned to node $v$.
 We can also think of this as first assigning $\Big\lfloor \frac{x^{(t-1)}_u+x^{(t-1)}_v}{2} \Big\rfloor $ tokens to both $u$ and $v$ and then assigning the excess token (if there is one) to $u$ or $v$ with probability $1/2$ each.  We use a uniform random  variable $\Phi_{u,v}^{(t)} \in \{-1,1\}$ to specify the orientation for an edge $\{u,v\} \in \M^{(t)}$, indicating the node that the excess token (if any) is assigned to.
If $\Phi_{u,v}^{(t)}=1$ the excess token is assigned to~$u$, and if
$\Phi_{u,v}^{(t)}=-1$, the excess token is assigned to~$v$. In particular, we have $\Phi_{u,v}^{(t)}=-\Phi_{v,u}^{(t)}$. Further, we point out that the deterministic orientation of \citet{RSW98} corresponds to $\Phi_{u,v}^{(t)}=1$ for $x_u^{(t-1)} \geq x_v^{(t-1)}$ and $\Phi_{u,v}^{(t)}=-1$, otherwise. In other words, the excess token is always kept at the node with the larger load.

For every edge $\{u,v\} \in \M^{(t)}$, we define a corresponding error term by
\[
e^{(t)}_{u,v} := \frac{1}{2} \Odd (x_u^{(t-1)}+x_v^{(t-1)} )\cdot
\Phi_{u,v}^{(t)},
\]
where $\Odd(x):= x\mod 2$ for any integer $x$. Observe that $\Ex{ e^{(t)}_{u,v}}=0$, and
\begin{align}
x_u^{(t)} &=\frac{1}{2}x_u^{(t-1)}+\frac{1}{2}x_v^{(t-1)}
+e_{u,v}^{(t)} \label{eq:error}
\end{align}
for any matched nodes $\{u,v\}$ in round $t$.
Moreover, for any round $t$ we define an error vector $e^{(t)}$ with
$
    e_u^{(t)}:= \sum_{v \colon \{u,v\} \in \M^{(t)}} e_{u,v}^{(t)}
$.
With this notation, the load vector in round $t$ is
$
    x^{(t)} = x^{(t-1)} \M^{(t)} + e^{(t)}.
$
Solving this recursion (cf.~\citep{RSW98}) yields
\begin{align*}
    x^{(t)}
    &= x^{(0)} \M^{[1,t]} + \sum_{s=1}^{t} e^{(s)} \M^{[s+1,t]} = \xi^{(t)} + \sum_{s=1}^{t } e^{(s)}  \M^{[s+1,t]},
\end{align*}
where $\xi^{(t)}$ is the corresponding load vector in the continuous case initialized with $\xi^{(0)}=x^{(0)}$.
Hence, for any node $w \in V$, we have
\begin{align}
  x_w^{(t)} - \xi_w^{(t)}
& &= \sum_{s=1}^{t} \sum_{u\in V}
     \sum_{v \colon \{u,v\} \in \M^{(s)}} e_{u,v}^{(s)} \M^{[s+1,t]}_{u,w}
   = \sum_{s=1}^{t} \sum_{ [u:v] \in \M^{(s)}}
     e_{u,v}^{(s)} \, \left( \M^{[s+1,t]}_{u,w} - \M^{[s+1,t]}_{v,w} \right),
  \label{eq:std}
\end{align}
where the last equality used $e_{u,v}^{(s)}=-e_{v,u}^{(s)}$ (recalling that $[u:v] \in \M^{(s)}$ means $\{u,v\} \in \M^{(s)}$ and $u < v$).

Occasionally, it is convenient to assume that $\overline{x}\in [0,1)$ by adding or subtracting the same number of tokens to each node. Although this may lead to negative entries in the load vector, the above formulas still hold.

\begin{obs}\label{obs:scaling}
Fix a sequence of matchings $\mathcal{M}=\langle \M^{(1)}, \M^{(2)}, \ldots \rangle$ and orientations $\Phi_{u,v}^{(t)}$, $[u:v] \in \M^{(t)}, t \in \mathbb{N}$. Consider two executions of the discrete load balancing protocol with the same matchings and orientations, but with different initial load vectors, $x^{(0)}$ and $\widetilde{x}^{(0)}$. Then, the following statements hold:
\begin{itemize}\itemsep 0pt
 \item If $\widetilde{x}^{(0)} = x^{(0)} + \alpha \cdot \mathbf{1}$ for some $\alpha \in \mathbb{Z}$, then $\widetilde{x}^{(t)} = x^{(t)} + \alpha \cdot \mathbf{1}$ for all $t\geq 1$.
 \item If $\widetilde{x}_u^{(0)} \leq x_u^{(0)}$ for all $u \in V$, then  $\widetilde{x}_u^{(t)} \leq x_u^{(t)}$ for all $u \in V$ and $t \geq 1$.
\end{itemize}
\end{obs}

The next lemma demonstrates that upper bounding the maximum load is essentially equivalent to lower bounding the minimum load.

\begin{lemma}\label{lem:maxminrelation}
Fix a sequence of matchings $\mathcal{M}=\langle \M^{(1)}, \M^{(2)}, \ldots \rangle$. For any triple of positive integers $K$, $\alpha$, and $t$, it holds that
\begin{align*}
 \max_{ \substack{ y \in \mathbb{Z}^n \colon \\ \disc(y) \leq K}} \left\{ \Pro{ x_{\max}^{(t)} \geq \lfloor \overline{x} \rfloor + \alpha  \, \Big| \, x^{(0)} = y } \right\} &\leq \max_{ \substack{ y \in \mathbb{Z}^n \colon \\ \disc(y) \leq K}} \left\{ \Pro{x_{\min}^{(t)} \leq \lfloor \overline{x} \rfloor - \alpha + 2  \, \Big| \, x^{(0)} = y }\right\},
\end{align*}
and
\begin{align*}
 \max_{ \substack{ y \in \mathbb{Z}^n \colon \\ \disc(y) \leq K}} \left\{ \Pro{ x_{\min}^{(t)} \leq \lfloor \overline{x} \rfloor - \alpha  \, \Big| \, x^{(0)} = y } \right\} &\leq \max_{ \substack{ y \in \mathbb{Z}^n \colon \\ \disc(y) \leq K}} \left\{ \Pro{x_{\max}^{(t)} \geq \lfloor \overline{x} \rfloor + \alpha  \, \Big| \, x^{(0)} = y }\right\}.
\end{align*}
\end{lemma}
\begin{proof}
We define a coupling between two executions of the load balancing protocol. For every round $t$, the two executions use the same matching. In the first execution, we start with $x^{(0)}=z \in \mathbb{Z}^n$
that maximizes $ \Pro{ x_{\max}^{(t)} \geq \lfloor \overline{x} \rfloor + \alpha  \, \Big| \, x^{(0)} = z }$ and satisfies $\disc(z) \leq K$.
  The load vector of the second execution is denoted by $\tilde{x}^{(t)}$ with average load $\overline{\widetilde{x}}$, and is initialized by \[
  \tilde{x}^{(0)}_u := \lfloor \overline{z} \rfloor - (z_u - \lfloor \overline{z} \rfloor) = 2 \cdot \lfloor \overline{z} \rfloor - z_u, \qquad \mbox{for any $u \in V$.}
  \]
  Note that $\disc(\tilde{x}^{(0)}) = \disc(z) \leq K$ and $\overline{x} \geq \overline{\widetilde{x}} \geq \overline{x} - 2$.
  We couple the random orientations of the two executions by setting $
 {\tilde{\Phi}}_{u,v}^{(s)} = -{\Phi}_{u,v}^{(s)},
$  for every $[u:v] \in \M^{(s)}$, $1 \leq s \leq t$,
where $\tilde{\Phi}$ denotes the orientations of the second execution.
We now claim that for every round $s \geq 0$,
\begin{align}
  \tilde{x}_u^{(s)} = 2 \cdot \lfloor \overline{x} \rfloor - x_u^{(s)}, \qquad \mbox{for any $u \in V$.}\label{eq:easyinduction}
\end{align}
This claim is shown by induction on $s$. The base case $s=0$ holds by definition. For the induction step, consider first a node $u$  matched with a node $v$ in round $s \geq 1$. By \eq{error},
\begin{align*}
 \tilde{x}_u^{(s)} &= \frac{1}{2} \cdot \tilde{x}_u^{(s-1)} + \frac{1}{2} \cdot \tilde{x}_v^{(s-1)} + \frac{1}{2} \cdot \mathsf{Odd}(\tilde{x}_u^{(s-1)} + \tilde{x}_v^{(s-1)}) \cdot \widetilde{\Phi}_{u,v}^{(s)},
\end{align*}
and using the induction hypothesis yields
\begin{align*}
 \tilde{x}_u^{(s)} &= \frac{1}{2} \cdot (2 \lfloor \overline{x} \rfloor - x_u^{(s-1)} ) + \frac{1}{2} \cdot (2 \lfloor \overline{x} \rfloor - x_v^{(s-1)} ) -\frac{1}{2} \cdot \mathsf{Odd}(2 \lfloor \overline{x} \rfloor - x_u^{(s-1)} + 2 \lfloor \overline{x} \rfloor - x_v^{(s-1)}) \cdot \Phi_{u,v}^{(s)} \\
 &= 2 \cdot \lfloor \overline{x} \rfloor - \left( \frac{1}{2} \cdot x_u^{(s-1)} + \frac{1}{2} \cdot  x_v^{(s-1)}  + \frac{1}{2} \cdot \mathsf{Odd}( x_u^{(s-1)} + x_v^{(s-1)})\cdot \Phi_{u,v}^{(s)}              \right) \\
 &= 2 \cdot \lfloor \overline{x} \rfloor - x_u^{(s)}.
\end{align*}

If a node $u$ is not matched in round $s$, then the claim follows directly by the induction hypothesis. Hence, \eq{easyinduction} holds, which implies $x_{\max}^{(s)} = 2 \lfloor \overline{x} \rfloor - \widetilde{x}_{\min}^{(s)}$;  therefore, for any $\alpha \geq 1$,
\begin{align*}
  \max_{ \substack{ y \in \mathbb{Z}^n \colon \\ \disc(y) \leq K}} \left\{ \Pro{ x_{\max}^{(t)} \geq \lfloor \overline{x} \rfloor + \alpha  \, \Big| \, x^{(0)} = y } \right\}   &= \Pro{ x_{\max}^{(t)} \geq \lfloor \overline{x} \rfloor + \alpha \, \Big| \, x^{(0)} = z } \\ &= \Pro{ \tilde{x}_{\min}^{(t)} \leq \lfloor \overline{x} \rfloor - \alpha \, \Big| \, \widetilde{x}^{(0)} = 2 \cdot \lfloor \overline{z} \rfloor\cdot\mathbf{1} - z } \\
  &\leq  \Pro{ \widetilde{x}_{\min}^{(t)} \leq \lfloor \overline{\widetilde{x}} \rfloor  - \alpha + 2 \, \Big| \, \widetilde{x}^{(0)} = 2 \cdot \lfloor \overline{z} \rfloor\cdot\mathbf{1} - z } \\
   &\leq \max_{ \substack{ y \in \mathbb{Z}^n \colon \\ \disc(y) \leq K}} \left\{ \Pro{x_{\min}^{(t)} \leq \lfloor \overline{x} \rfloor  - \alpha + 2 \, \Big| \, x^{(0)} = y }\right\}.
\end{align*}

The second inequality is shown in the same way. We define $x^{(0)}=z\in\mathbb{Z}^{n}$ that maximizes
$\Pro{x^{(t)}_{\min}\leq \lfloor \overline{x} \rfloor -\alpha \, \big|\, x^{(0)}= z }$. Then,
\begin{align*}
  \max_{ \substack{ y \in \mathbb{Z}^n \colon \\ \disc(y) \leq K}} \left\{ \Pro{ x_{\min}^{(t)} \leq \lfloor \overline{x} \rfloor - \alpha  \, \Big| \, x^{(0)} = y } \right\}   &= \Pro{ x_{\min}^{(t)} \leq \lfloor \overline{x} \rfloor - \alpha \, \Big| \, x^{(0)} = z } \\ &= \Pro{ \tilde{x}_{\max}^{(t)} \geq \lfloor \overline{x} \rfloor + \alpha \, \Big| \, \widetilde{x}^{(0)} = 2 \cdot \lfloor \overline{z} \rfloor\cdot\mathbf{1} - z } \\
  &\leq  \Pro{ \widetilde{x}_{\max}^{(t)} \geq \lfloor \overline{\widetilde{x}} \rfloor + \alpha  \, \Big| \, \widetilde{x}^{(0)} = 2 \cdot \lfloor \overline{z} \rfloor\cdot\mathbf{1} - z } \\
   &\leq \max_{ \substack{ y \in \mathbb{Z}^n \colon \\ \disc(y) \leq K}} \left\{ \Pro{x_{\max}^{(t)} \geq \lfloor \overline{x} \rfloor  +\alpha  \, \Big| \, x^{(0)} = y }\right\}.
\end{align*}
\end{proof}



\begin{rem}
When referring to the random matching model (discrete case), the probability space is over the randomly generated matchings {\em and} the randomized orientations of the matchings. For the balancing circuit model (discrete case), the probability space is over the randomized orientations of the (deterministic) matchings.
\end{rem}

\section{Local Divergence and Discrepancy}\label{sec:divergence}

To bound the deviation between the discrete and continuous cases,
 we consider $\max\limits_{w\in V}\big| x_w^{(t)}- \xi_w^{(t)} \big|$ in \eq{std} at all rounds $t$ and define the local divergence for the matching model.

\begin{defi}[Local $p$-Divergence for Matchings]
For any graph $G$, $p\in \mathbb{N} \setminus \{0\}$ and an arbitrary sequence of matchings
$\mathcal{M}=\langle\mathbf{M}^{(1)},\mathbf{M}^{(2)},\ldots\rangle$,  the local $p$-divergence is defined by
\[
\Psi_p(\mathcal{M}):=\max_{w\in V}\left( \sup_{t\in \mathbb{N}} \sum_{s=1}^{t} \sum_{[u:v]\in \M^{(s)}}\left|\M^{[s+1,t]}_{u,w}-\M^{[s+1,t]}_{v,w}\right|^p  \right)^{1/p}.
\]
\end{defi}
It was pointed out in \cite{RSW98} that $\Psi_1(\mathcal{M})$
is a natural quantity that measures the sum of load differences across all edges in the network aggregated over time. The difference $\left|\M^{[s+1,t]}_{u,w}-\M^{[s+1,t]}_{v,w}\right| $ is the weight of the rounding error $e_{u,v}^{(s)}$ when computing the load vector $x^{(t)}$ as in \eq{std}. The earlier the round $s$, the less the impact of the rounding error $e_{u,v}^{(s)}$, due to the smoothing of the load vector that happens between the rounds $s+1$ and $t$. In general, the local $p$-divergence $\Psi_p(\mathcal{M})$ is the $p$th norm of $\Psi_1(\mathcal{M})$. Of particular interest is the local $2$-divergence, as this value turns out to be closely related to the variance of $x_w^{(t)} - \xi_w^{(t)}$ (see \lemref{martingalebound} on page~\pageref{lem:martingalebound}).

Next we turn to our upper bound on the local $2$-divergence.

\begin{theorem}\label{thm:divergencebalancing}
Let $G$ be any graph with an arbitrary sequence of matchings $\mathcal{M}=\langle \M^{(1)},\M^{(2)},\ldots \rangle$. Then, the following statements hold:
\begin{itemize}
 \item For an arbitrary node $w \in V$ and any pair of rounds $t_1 \leq t$, it holds that
\begin{align*}
 \sum_{s=1}^{t_1} \sum_{[u:v] \in \M^{(s)}} \left( \M_{u,w}^{[s+1,t]} - \M_{v,w}^{[s+1,t]} \right)^2 &\leq 2 \cdot \sum_{k\in V} \left( \M_{k,w}^{[t_1+1,t]} - \frac{1}{n} \right)^2 .
\end{align*}
 \item It holds that  $
 \Psi_2(\mathcal{M}) \leq \sqrt{2 - 2/n}.
$
Moreover, if there is a matching $\mathbf{M}^{(t)}$ in
 $\mathcal{M}$ which contains at least one edge, then
$
 \Psi_2(\mathcal{M}) \geq 1.
$
\end{itemize}
\end{theorem}

While all previous upper bounds on the local divergence are increasing in the expansion, the degree, and/or the number of nodes \cite{RSW98,FS09,BCFFS11},\thmref{divergencebalancing} reveals  that the local $2$-divergence is essentially independent of any graph parameter. On the other hand, the local $1$-divergence is always lower bounded by the diameter of graph $G$ (cf.~\cite{FS09}). We continue with the proof of \thmref{divergencebalancing}.

\begin{proof}
Fix any pair of node $w \in V$ and round $t$. For any $0 \leq s \leq t$, define the following potential function:\[
  \Phi^{(s)} := \sum_{k\in V} \left( \M_{k,w}^{[s+1,t]} - \frac{1}{n} \right)^2.
\]
Notice that since $\M_{\cdot, w}^{[s+1,t]}$ is a probability vector,
$\Phi^{(s)} \leq 1 - \frac{1}{n}$ with equality holding for $s=t$ since  $\M^{[t+1,t]}=\mathbf{I}$. Consider now any round $1 \leq s \leq t$, and let $u,v$ be nodes with $[u:v] \in \M^{(s)}$. Let $y_u :=  \M_{u,w}^{[s+1,t]}$ and
$y_v :=  \M_{v,w}^{[s+1,t]}$. Note that
\[
  \M_{u,w}^{[s,t]} = \sum_{z \in V} \M_{u,z}^{[s,s]} \cdot \M_{z,w}^{[s+1,t]} = \frac{y_u + y_v}{2},
\]
and similarly,
$
 \M_{v,w}^{[s,t]} = \frac{y_u + y_v}{2}.
$
Therefore, the contribution of $u$ and $v$ to
$\Phi^{(s)} - \Phi^{(s-1)}$ equals
\begin{align*}
\lefteqn{ \left(	 \M_{u,w}^{[s+1,t]} - \frac{1}{n}		 \right)^2 + \left(	 \M_{v,w}^{[s+1,t]} - \frac{1}{n}		 \right)^2
- \left(	\M_{u,w}^{[s,t]} - \frac{1}{n}		 \right)^2 - \left(	 \M_{v,w}^{[s,t]} - \frac{1}{n}		 \right)^2 }  \\
&= \left(y_u - \frac{1}{n} \right)^2 + \left(y_v - \frac{1}{n} \right)^2 - \left(\frac{y_u + y_v}{2} - \frac{1}{n} \right)^2
 - \left(\frac{y_u + y_v}{2} - \frac{1}{n} \right)^2  \\
 &= y_u^2 - \frac{2}{n} y_u + \frac{1}{n^2} + y_v^2 - \frac{2}{n} y_v + \frac{1}{n^2} - 2 \cdot \left( \frac{(y_u+y_v)^2}{4} - \frac{y_u+y_v}{n} + \frac{1}{n^2} \right) \\
 &= y_u^2 + y_v^2 - \frac{y_u^2 + 2 y_u y_v + y_v^2}{2} \\
 &= \frac{y_u^2}{2} - y_u y_v + \frac{y_v^2}{2} = \frac{1}{2} \cdot \left( y_u - y_v \right)^2.
\end{align*}
If a node is not matched in round $s$, then its contribution to $\Phi^{(s)}-\Phi^{(s-1)}$ equals zero. Accumulating the contribution of all nodes  yields
\begin{align}
 \Phi^{(s)} - \Phi^{(s-1)} = \sum_{[u:v] \in \M^{(s)}} \frac{1}{2} \cdot \left( \M_{u,w}^{[s+1,t]} - \M_{v,w}^{[s+1,t]} \right)^2\ . \label{eq:divergenceproof}
\end{align}
Therefore,
\begin{align*}
\sum_{s=1}^{t_1}   \sum_{[u:v] \in \M^{(s)}} \left( \M_{w,u}^{[s+1,t]} - \M_{w,v}^{[s+1,t]} \right)^2 &= 2 \sum_{s=1}^{t_1} \left( \Phi^{(s)} - \Phi^{(s-1)} \right) = 2 \Phi^{(t_1)} - 2 \Phi^{(0)} \leq 2\Phi^{(t_1)},
\end{align*}
and the first statement follows.

To prove the second statement, we use the inequality above with $t_1=t$ and $\Phi^{(t)} = 1-1/n$ to obtain
\begin{align}
  \sum_{s=1}^t \sum_{[u:v] \in \M^{(s)}} \left( \M_{u,w}^{[s+1,t]} - \M_{v,w}^{[s+1,t]} \right)^2 &= 2\Phi^{(t)} - 2 \Phi^{(0)} \leq 2 \cdot \left(1 - \frac{1}{n} \right), \label{eq:divergenceequality}
\end{align}
which directly implies that $\Psi_2(\mathcal{M})\leq\sqrt{2-2/n}$.
For the lower bound, consider any round $t$ with $[u:v] \in \M^{(t)}$. Clearly, as $\M^{[t+1,t]}=\mathbf{I}$,
\[
  \M_{u,u}^{[t+1,t]} = \M_{v,v}^{[t+1,t]} = 1,
\quad \mbox{ and } \quad
  \M_{u,v}^{[t+1,t]} = \M_{v,u}^{[t+1,t]} = 0,
\]
and hence,
\begin{align*}
  \Psi_2(\mathcal{M}) &\geq \sqrt{  \left( \M_{u,u}^{[t+1,t]} - \M_{v,u}^{[t+1,t]}			 \right)^{2} } \geq
  1\enspace. 
\end{align*}
\end{proof}

The following corollary shows that the upper bound $\Psi_2(\mathcal{M}) \leq \sqrt{2 - 2/n}$ is actually tight if the sequence of matchings should balance load vectors with arbitrarily large discrepancies.
\begin{corollary}\label{rem:divergencebound}
Suppose that there is a $K$ with $K \geq \sqrt{2n}$ and integer $t=t(K)$ so that the time-interval $[0,t]$ is $(K,1)$-smoothing. Then, $\Psi_2(\mathcal{M}) \geq \sqrt{2 - \frac{2}{n} - \frac{2n}{K^2}}$.
\end{corollary}
\begin{proof}
Recall that the time-interval $[0,t]$ is $(K,1)$-smoothing if and only if the time-interval is $(1,1/K)$-smoothing. By the second statement of \lemref{boundbymixinglemma}, we have that $\left| \M_{u,v}^{[1,t]} - \frac{1}{n} \right| \leq \frac{1}{K}$ for any $u, v \in V$.  Hence, \[\Phi^{(0)}
= \sum_{k\in V} \left( \M_{k,w}^{[1,t]} -\frac{1}{n} \right)^2
 \leq \frac{n}{K^2}\] for any fixed $w \in V$, and it follows by the equality in \eq{divergenceequality} that \[
 \Psi_2(\mathcal{M}) \geq \sqrt{2 - \frac{2}{n} - 2 \Phi^{(0)}} \geq \sqrt{2 - \frac{2}{n} - \frac{2n}{K^2} }.\qedhere
 \]
\end{proof}


In the next lemma, we derive a Chernoff-type concentration inequality for bounding a sum of rounding errors.

\begin{lemma}\label{lem:martingalebound}
Consider an arbitrary sequence of matchings $\mathcal{M}=\langle \M^{(1)},\M^{(2)},\ldots \rangle$.
Fix two rounds $0 \leq  t_1 < t_2$ and the load
vector $x^{(t_1)}$ at the end of round $t_1$. For any family of  numbers $g^{(s)}_{u,v}\ ([u:v]\in \M^{(s)}, t_1 +1 \leq s \leq t_2)$, define the random variable
\[
   Z := \sum_{s=t_1+1}^{t_2} \sum_{[u:v] \in \M^{(s)}} g_{u,v}^{(s)}\cdot e_{u,v}^{(s)}.
\]
Then, $\Ex{Z}=0$, and for any $\delta > 0$, it holds that
\begin{align*}
 \Pro{ \left| Z - \Ex{Z}           \right|  \geq \delta       }
 \leq
2 \exp \left(-  \frac{ \delta^2}{2 \sum_{s=t_1+1}^{t_2} \sum_{[u:v] \in \M^{(s)}} \left(g^{(s)}_{u,v}\right)^2 }         \right).
\end{align*}
\end{lemma}
\begin{proof}
The proof of this lemma is similar to the one of \cite[Theorem~1.1, first statement]{BCFFS11}.

Since $\Ex{ e_{u,v}^{(s)} }=0$ for all $\{u,v\} \in \M^{(s)}$, it follows that $\Ex{ Z } = 0$. Our goal is now to prove that $Z$ is concentrated around its mean by applying the concentration inequality in \thmref{martingale}. Observe that $Z$
depends on at most $(n/2) \cdot (t_2-t_1)$ random variables $e_{u,v}^{(s)}$, and each of them corresponds to one of the orientations of at most $n/2$ matching edges in each of the $t_2-t_1$ rounds. Let us denote this sequence by $Y_{\ell}$ with $(t_1+1)\cdot (n/2) + 1 \leq \ell \leq (t_2+1) \cdot (n/2)$, where $Y_{\ell}$ with $\ell = \alpha \cdot (n/2) + \beta$, $t_1+1 \leq \alpha \leq t_2$, $1 \leq \beta \leq n/2$, describes the orientation of the $\beta$th matching edge $[u':v']$ in round $\alpha$, so $Y_{\ell} = \Phi_{u',v'}^{(s)}$. Here, we take an arbitrary ordering of the matching edges in round $\alpha$, and if there are less than $\beta$ matching edges in round $\alpha$, then $Y_{\ell} = 0$.

In order to apply \thmref{martingale}, we first verify that for every $(t_1+1)\cdot (n/2) + 1 \leq \ell \leq (t_2+1) \cdot (n/2)$ with $\ell = \alpha \cdot (n/2) + \beta$
and $[u':v']$ being the $\beta$th matching edge in round $\alpha$,
\begin{align}
 \left| \Ex{ Z \, \mid \, Y_{(t_1+1) \cdot (n/2)+1}, \ldots, Y_{\ell} }
 %
 - \Ex{ Z \, \mid \, Y_{(t_1+1) \cdot (n/2)+1}, \ldots, Y_{\ell-1} } \right| &\leq \left|g^{(\alpha)}_{u',v'}\right|. \label{eq:azuma}
\end{align}
In order to simplify the notation,
let $\mathcal{Y}_{\ell}:=(Y_{(t_1+1) \cdot (n/2)+1},\ldots,Y_{\ell} )$ for any $\ell$.

To prove \eq{azuma}, we split the sum of $Z$ into three parts: $s < \alpha$, $s=\alpha$, and $s > \alpha$.

\medskip
{\bf Case 1: $t_1+1 \leq s \leq \alpha-1$.}
For every $[u:v] \in \M^{(s)}$, $e_{u,v}^{(s)}$ is determined by
$\mathcal{Y}_{\ell-1}$. Hence,
\begin{align*}
 \left| \Ex{  \sum_{s=t_1+1}^{\alpha-1} \sum_{[u:v] \in \M^{(s)}} g^{(s)}_{u,v}\cdot e_{u,v}^{(s)} \, \bigg| \, \mathcal{Y}_{\ell} }
  - \Ex{  \sum_{s=t_1+1}^{\alpha-1} \sum_{[u:v] \in \M^{(s)}} g^{(s)}_{u,v} \cdot e_{u,v}^{(s)} \, \bigg| \, \mathcal{Y}_{\ell-1} } \right| &= 0.
\end{align*}

{\bf Case 2: $s = \alpha$.}
Then,
\begin{align*}
\lefteqn{ \left| \Ex{ \sum_{[u:v] \in \M^{(\alpha)}} g^{(\alpha)}_{u,v} \cdot  e_{u,v}^{(\alpha)} \, \Big| \, \mathcal{Y}_{\ell} }
    -  \Ex{   \sum_{[u:v] \in \M^{(\alpha)}}
g^{(\alpha)}_{u,v} \cdot e_{u,v}^{(\alpha)} \, \Big| \, \mathcal{Y}_{\ell-1} } \right| } \\
&\leq \sum_{\substack{[u:v] \in \M^{(\alpha)}\\ [u:v] \neq [u':v'] } } \left|
\Ex{  g^{(\alpha)}_{u,v} \cdot  e_{u,v}^{(\alpha)} \, \Big| \, \mathcal{Y}_{\ell}  }
    -  \Ex{
g^{(\alpha)}_{u,v} \cdot e_{u,v}^{(\alpha)} \, \Big| \, \mathcal{Y}_{\ell-1} } \right| \\
&\quad\quad\quad\quad\quad\, + \left| \Ex{  g^{(\alpha)}_{u',v'} \cdot  e_{u',v'}^{(\alpha)} \, \Big| \, \mathcal{Y}_{\ell}  }
    -  \Ex{
g^{(\alpha)}_{u',v'} \cdot e_{u',v'}^{(\alpha)} \, \Big| \, \mathcal{Y}_{\ell-1} } \right|.
\end{align*}
Note for any edge $[u:v] \in \M^{(\alpha)}, [u:v] \neq [u':v'] $ whose orientation is determined by $\mathcal{Y}_{\ell-1}$, it holds that
\[
 \left|
\Ex{  g^{(\alpha)}_{u,v} \cdot  e_{u,v}^{(\alpha)} \, \Big| \, \mathcal{Y}_{\ell}  }
    -  \Ex{
g^{(\alpha)}_{u,v} \cdot e_{u,v}^{(\alpha)} \, \Big| \, \mathcal{Y}_{\ell-1} } \right| = 0,
\]
since both expectations are determined by $\mathcal{Y}_{\ell-1}$ and $\mathcal{Y}_{\ell}$. For any edge $[u:v] \in \M^{(\alpha)},$ $[u:v] \neq [u':v'] $ whose orientation is not determined by $\mathcal{Y}_{\ell-1}$ (and thus also not determined by $\mathcal{Y}_{\ell}$), the rounding error $e_{u,v}^{(\alpha)}$ is independent of all other rounding errors in round $\alpha$, so in particular, it is independent of $e_{u',v'}^{(\alpha)}$. Therefore,
\[
 \left|
\Ex{  g^{(\alpha)}_{u,v} \cdot  e_{u,v}^{(\alpha)} \, \Big| \, \mathcal{Y}_{\ell}  }
    -  \Ex{
g^{(\alpha)}_{u,v} \cdot e_{u,v}^{(\alpha)} \, \Big| \, \mathcal{Y}_{\ell-1} } \right| = 0 - 0 = 0.
\]
Consequently,
\begin{align*}
\lefteqn{ \left| \Ex{ \sum_{[u:v] \in \M^{(\alpha)}} g^{(\alpha)}_{u,v} \cdot  e_{u,v}^{(\alpha)} \, \Big| \, \mathcal{Y}_{\ell} }
    -  \Ex{   \sum_{[u:v] \in \M^{(\alpha)}}
g^{(\alpha)}_{u,v} \cdot e_{u,v}^{(\alpha)} \, \Big| \, \mathcal{Y}_{\ell-1} } \right| } \\
&\leq  \left| \Ex{
g^{(\alpha)}_{u',v'} \cdot e_{u',v'}^{(\alpha)} \, \Big| \, \mathcal{Y}_{\ell} } - \Ex{
g^{(\alpha)}_{u',v'} \cdot e_{u',v'}^{(\alpha)} \, \Big| \, \mathcal{Y}_{\ell-1} } \right| \\
&\leq \left|g^{(\alpha)}_{u',v'}\right|,
\end{align*}
where the last inequality holds, since $e_{u',v'}^{(\alpha)} \in \{-1/2,0,+1/2\}$.

{\bf Case 3: $\alpha +1 \leq s \leq t_2$.}
Let $\tilde{\ell} \geq \ell$ be the smallest integer so that $\mathcal{Y}_{\tilde{\ell}}$ determines the load vector $x^{(\alpha)}$.
By the law of total expectation, it holds for any $\{u,v\} \in \M^{(s)}$ that
\begin{align*}
   \Ex{ e_{u,v}^{(s)} \, \mid \, \mathcal{Y}_{\ell} }
&= \Ex{  \Ex{ e_{u,v}^{(s)} \, \mid \, \mathcal{Y}_{\tilde{\ell}}   }  \, \mid \,  \mathcal{Y}_{\ell}    } = \Ex{ 0 \, \mid \, \mathcal{Y}_{\ell} } = 0,
\end{align*}
and the same also holds if we replace $\ell$ by $\ell-1$.
Hence, by the linearity of expectation,
\begin{align*}
 \left| \Ex{  \sum_{s=\alpha+1}^{t_2} \sum_{[u:v] \in \M^{(s)}} g^{(s)}_{u,v}\cdot e_{u,v}^{(s)} \, \Big| \, \mathcal{Y}_{\ell} }
 - \Ex{  \sum_{s=\alpha+1}^{t_2} \sum_{[u:v] \in \M^{(s)}} g^{(s)}_{u,v} \cdot e_{u,v}^{(s)} \, \Big| \, \mathcal{Y}_{\ell-1}  } \right|
 &= 0.
\end{align*}
Combining the contribution of all three cases establishes \eq{azuma}. Applying \thmref{martingale} to the martingale $X_{\ell}:= \Ex{ Z \, \mid \, {\cal Y}_{\ell}}$, where $(t_1+1) \cdot (n/2)+1 \leq \ell \leq (t_2+1) \cdot (n/2)$, finishes the proof.
\end{proof}

We now derive the following Chernoff-type bounds from \lemref{martingalebound}. Although similar bounds have been established in previous works \cite{FS09,BCFFS11},
we obtain a much better concentration (which is independent of the graph's expansion) due to our new approach of bounding the local $2$-divergence. Specifically, the second statement provides Gaussian tail bounds on the deviation from the average load.

\begin{lemma}\label{lem:bounddifference}
Consider an arbitrary sequence of matchings $\mathcal{M}=\langle \M^{(1)},\M^{(2)},\ldots \rangle$. Fix two rounds $t_1 \leq t_2$ and the initial load vector $x^{(0)}$. If the time-interval $[0,t_1]$ is $(K,1/(2n))$--smoothing, then for any node $w\in V$ and $\delta>1/(2n)$, it holds that
\begin{align*}
\Pro{ \left|\sum_{k\in V} x_k^{(t_1)}\cdot \M^{[t_1+1,t_2]}_{k,w} - \overline{x} \right|\geq\delta }\leq2\cdot \exp\left( -\frac{(\delta - 1/(2n) )^2}{4 \sum_{k\in V} \left( \M^{[t_1+1,t_2]}_{k,w} -1/n\right)^2}\right).
\end{align*}
In particular, for any node $w\in V$ and $\delta>1/(2n)$, it holds that
\[
\Pro{ \left|x_w^{(t_1)} - \overline{x} \right|\geq\delta }\leq2\cdot \exp\left( -\left(\delta - \frac{1}{2n} \right)^2\Big\slash 4 \right).
\]
\end{lemma}
\begin{proof}
 By \eq{std},  for any node $k$ and round $t$, it holds that
\begin{align*}
  x_k^{(t)} =  \xi_k^{(t)} +
    \sum_{s=1}^{t} \sum_{ [u:v] \in \M^{(s)}}
     \left( \M^{[s+1,t]}_{u,k} - \M^{[s+1,t]}_{v,k} \right)\cdot e_{u,v}^{(s)},
\end{align*}
where  $\xi^{(0)}=x^{(0)}$. Therefore,
\begin{align*}
\lefteqn{\sum_{k\in V} x_k^{(t_1)} \cdot \M^{[t_1+1,t_2]}_{k,w}}\\
&=\sum_{k\in V} \left(\xi_k^{(t_1)}+\sum_{s=1}^{t_1} \sum_{[u:v]\in \M^{(s)}}\left( \M^{[s+1,t_1]}_{u,k}-\M^{[s+1,t_1]}_{v,k}\right)\cdot e_{u,v}^{(s)}\right)\cdot \M^{[t_1+1,t_2]}_{k,w}\\
&=\sum_{k\in V} \xi_k^{(t_1)}\cdot \M^{[t_1+1,t_2]}_{k,w} + \sum_{k\in V} \sum_{s=1}^{t_1}\sum_{[u:v]\in \M^{(s)}} \M^{[t_1+1,t_2]}_{k,w}\cdot \left( \M^{[s+1,t_1]}_{u,k}-\M^{[s+1,t_1]}_{v,k} \right)\cdot e_{u,v}^{(s)}~.\\
 \end{align*}
By the first statement of Lemma~\ref{lem:boundbymixinglemma} and $\overline{\xi}=\overline{x}$,  we have
$ \sum_{k\in V} \xi_k^{(t_1)}\cdot \M^{[t_1+1,t_2]}_{k,w}=  \overline{x} + \Theta$ after $t_1$ rounds, where
$|\Theta| \leq 1/(2n)$.
  Therefore,
\[
\sum_{k\in V} x_k^{(t_1)} \cdot \M^{[t_1+1,t_2]}_{k,w}
=\overline{x} + \Theta +\sum_{k\in V} \sum_{s=1}^{t_1}\sum_{[u:v]\in \M^{(s)}}\M^{[t_1+1,t_2]}_{k,w}\cdot \left( \M^{[s+1,t_1]}_{u,k}-\M^{[s+1,t_1]}_{v,k} \right)\cdot e_{u,v}^{(s)},
\]
and
\begin{align*}
\lefteqn{\Pro{ \left|\sum_{k\in V} x_k^{(t_1)}\cdot\M^{[t_1+1,t_2]}_{k,w}- \overline{x} \right|\geq\delta }} \\
&= \Pro{ \left|\left(\sum_{k\in V} \sum_{s=1}^{t_1}\sum_{[u:v]\in \M^{(s)}} \M^{[t_1+1,t_2]}_{k,w} \cdot \left( \M^{[s+1,t_1]}_{u,k}-\M^{[s+1,t_1]}_{v,k} \right)\cdot e_{u,v}^{(s)}\right) + \Theta \right|\geq\delta }\\
&\leq \Pro{ \left| \sum_{s=1}^{t_1}\sum_{[u:v]\in \M^{(s)}} \left( \sum_{k\in V}  \M^{[t_1+1,t_2]}_{k,w}\cdot \left( \M^{[s+1,t_1]}_{u,k}-\M^{[s+1,t_1]}_{v,k} \right) \right) \cdot e_{u,v}^{(s)} \right|\geq\delta - |\Theta| }\\
&\leq  2\cdot \exp\left( -\frac{(\delta-1/(2n))^2}{
2\sum_{s=1}^{t_1}\sum_{[u:v]\in \M^{(s)}}\left(
\sum_{k\in V} \M_{k,w}^{[t_1+1,t_2]}\cdot \left(
\M^{[s+1,t_1]}_{u,k}- \M^{[s+1,t_1]}_{v,k}\right)
\right)^2}
\right),
\end{align*}
where the last inequality follows from Lemma~\ref{lem:martingalebound}. Further,
 \begin{align*}
 \lefteqn{\sum_{s=1}^{t_1} \sum_{[u: v]\in \M^{(s)}}
 \left( \sum_{k\in V} \M_{k,w}^{[t_1+1,t_2]}\cdot\left( \M_{u,k}^{[s+1,t_1]}-\M_{v,k}^{[s+1,t_1]} \right) \right)^2} \\
& = \sum_{s=1}^{t_1} \sum_{[u:v]\in \M^{(s)}}\left( \M_{u,w}^{[s+1,t_2]}-\M_{v,w}^{[s+1,t_2]} \right)^2  \\
&\leq  2\cdot \sum_{k \in V} \left( \M_{k,w}^{[t_1+1,t_2]} - \frac{1}{n} \right)^2,
\end{align*}
where the last  inequality follows from the first statement of \thmref{divergencebalancing}. Therefore,
\begin{align*}
\Pro{ \left|\sum_{k\in V} x_k^{(t_1)}\cdot \M_{k,w}^{[t_1+1,t_2]}  - \overline{x} \right|\geq\delta }\leq2\cdot\exp\left( -\frac{(\delta - 1/(2n))^2}{
4 \sum_{k \in V} \left( \M^{[t_1+1,t_2]}_{k,w} -1/n \right)^2
} \right),
\end{align*}
which finishes the proof of the first statement. The second statement follows directly by using the first statement with $t_1=t_2$.
\end{proof}

Based on the upper bound on $\Psi_2(\mathcal{M})$, we obtain the following theorem:
\begin{theorem}\label{thm:deviationmatching}
Let $G$ be any graph. Then the following three statements hold:
\begin{itemize}\itemsep 0pt
\item Let $\mathcal{M}=\langle\M^{(1)},\M^{(2)},\ldots\rangle$ be any sequence of matchings, and let $x^{(0)} = \xi^{(0)}$. Then for any round $t$ and any $\kappa \geq 1$, it holds that
\[
  \Pro{ \max_{w \in V} \left| x_{w}^{(t)} - \xi_{w}^{(t)} \right| \leq \sqrt{4 \kappa \cdot \log n}}  \geq 1 - 2 n^{-\kappa+1}.
  \]
\item
In the balancing circuit model, for any initial load vector $x^{(0)}$ with discrepancy at most $K$, we reach a discrepancy of $ \sqrt{48 \log n}+1$ after $ \tau_{\cont}(K,1) = \Oh\left(d \cdot \frac{\log (Kn)}{1-\lambda(\M)}\right)$ rounds with probability at least $1-2 n^{-2}$. In the random matching model, we reach a discrepancy of $ \sqrt{48 \log n}+1$ after $\tau_{\cont}(K,1)=\Oh\left(\frac{\log (Kn)}{1-\lambda(\P)}\right)$ rounds with probability at least $1-2 n^{-1}$.
\item Consider the random matching model with $x^{(0)}=\xi^{(0)}$. If the initial load vector $x^{(0)}$ has discrepancy at most $K$, where $K\geq \mathrm{e}^{\mathrm{e}}$, then
\[
   \Pro{ \sup_{t \in \mathbb{N}} \max_{w \in V} \left| x_w^{(t)} - \xi_w^{(t)} \right| \leq \sqrt{4 \cdot \left(7 \log n + \log \log K \right)} + 1 } \geq 1 - 2 n^{-1}.
\]

\end{itemize}
\end{theorem}

The first statement of \thmref{deviationmatching} states that even if an adversary specifies the matchings for all rounds, it is not possible to achieve a deviation of more than $\Oh(\sqrt{\log n})$ between the discrete and the continuous cases (for a fixed round).
As shown in \cite[Theorem~6.3]{FGS12}, for the $d$-dimensional torus, the deviation is at least $\Omega((\log n)^{1/(4d)})$ with probability $1-o(1)$, which means our bound of $\Oh(\sqrt{\log n})$ is almost tight.
We continue with the proof of \thmref{deviationmatching}.

\begin{proof}
Recall by \eq{std} that
\begin{align*}
  x_{w}^{(t)} - \xi_{w}^{(t)} &= \sum_{s=1}^{t} \sum_{[u:v] \in \M^{(s)}} e_{u,v}^{(s)} \left(	 \M_{u,w}^{[s+1,t]} - \M_{v,w}^{[s+1,t]}		 \right).
\end{align*}
Applying \lemref{martingalebound} with $g_{u,v}^{(s)} = \M_{u,w}^{[s+1,t]} - \M_{v,w}^{[s+1,t]}$ yields, for any
$w\in V$, that
\begin{align*}
  \lefteqn{\Pro{   \left| x_{w}^{(t)} - \xi_{w}^{(t)} \right| \geq \sqrt{2 \kappa \cdot \log n} \cdot \Psi_2(\mathcal{M}) }} \\
    &\leq 2\exp\left(-\frac{2\kappa\cdot \log n\cdot (\Psi_2(\mathcal{M}))^2}{2\cdot \sum_{s=1}^t\sum_{[u:v]\in\M^{(s)}} \left(\M_{u,w}^{[s+1,t]} - \M_{v,w}^{[s+1,t]} \right)^2 }\right)\\
    &\leq 2 n^{-\kappa}.
\end{align*}
Taking the union bound over all $n$ nodes and recalling the bound $\Psi_2(\mathcal{M}) \leq \sqrt{2}$ from \thmref{divergencebalancing} completes the proof of the first statement. The second statement follows directly from the first statement ($\kappa=3$) and the definition of $\tau_{\cont}(K,1)$.

For the proof of the third statement, we require a general estimate on $\lambda_2(\bP)$. By Cheeger's inequality (cf.~\cite{Si92}), we have
$
 \lambda_2(\bP) \leq 1 - \Phi(\bP)^2/2$,
where $\Phi(\bP)$ is the conductance defined by
\[
  \Phi(\bP) = \min_{ \substack{S \subseteq V \colon \\ 0 < \pi(S) \leq 1/2 }} \frac{ \sum_{u \in S, v \not\in S} \pi(u) \cdot \bP_{u,v}  }{ \pi(S)}.
\]
Since $G$ is connected and $\pi$ is the uniform distribution, we have $\Phi(\bP) \geq \frac{1/n \cdot 1/(2n)}{1/2} =\frac{1}{n^2}$, and Cheeger's inequality implies that $\lambda_2(\bP) \leq 1 - \frac{1}{2 n^4}$. Using \thmref{contmatching}, it follows that in the continuous case, the discrepancy is at most $1$ after $\tau:=\Oh\left(\frac{ \log (Kn)}{ 1-\lambda_2(\bP)}\right)=\Oh( \log (Kn) \cdot n^4)$ rounds with probability at least $1-n^{-1}$. Using the first statement with $\kappa=7 + \log \log K/\log n$, it follows by the union bound over the time-interval $[1,\tau]$ that
\[
   \Pro{ \max_{t \in [0,\tau]} \max_{w \in V} \left| x_{w}^{(t)} - \xi_{w}^{(t)} \right| \leq \sqrt{4 \kappa \cdot \log n}  } \geq 1 - \Oh\left( \log (Kn) \cdot n^4\right) \cdot 2 n^{-\kappa+1} \geq 1 - n^{-1}.
\]

Now, consider the rounds $t > \tau$. Assuming that the above event occurs, we have for any $w \in V$ that
\begin{align*}
  \left| x_w^{(t)} - \xi_w^{(t)} \right|
  &\leq \max \left\{ x_{\max}^{(t)} - \xi_{\min}^{(t)}, \xi_{\max}^{(t)} - x_{\min}^{(t)}  \right\} \\
  &\leq \max \left\{ x_{\max}^{(\tau)} - \xi_{\min}^{(\tau)}, \xi_{\max}^{(\tau)} - x_{\min}^{(\tau)}  \right\} \\
  &\leq  \max \left\{ \max_{u \in V} \left( x_{u}^{(\tau)} - \xi_{u}^{(\tau)} \right), \max_{u \in V} \left( \xi_{u}^{(\tau)} - x_{u}^{(\tau)} \right) \right\} + \disc(\xi^{(\tau)}) \\
  &\leq \sqrt{4 \kappa \cdot \log n} + 1,
\end{align*}
which completes the proof.
\end{proof}




\section{Token-Based Analysis via Random Walks}\label{sec:randomwalk}

In this section, we relate the movement of tokens to independent random walks. In \secref{boundingload}, we formalize this relation and derive strong concentration results for the load distribution on a subset of nodes. In \secref{randomwalkresults}, we use these new concentration results to analyze the discrepancy on arbitrary (possibly non-regular) graphs. All our results in this section will hold for the balancing circuit and random matching model.

\subsection{Bounding the Load via Random Walks}\label{sec:boundingload}

We now present our  new approach that
allows us to upper bound the load of a node by assuming that the tokens
perform independent random walks in every round. Throughout \secref{boundingload}, we assume that the load vector is non-negative.

Let ${\cal T}=\{1,\ldots,\|x^{(0)} \|_{1} \}$ be the set of all tokens, which are assumed to be distinguishable for the sake of the analysis. The tokens may change their location via matching edges according to the following rule: If two nodes $u$ and $v$ are matched in round $t$, then  the $x_u^{(t-1)}+x_v^{(t-1)}$
tokens located at $u$ and $v$ at the end of round $t-1$ are placed in a single urn.
After that,  if $\Phi_{u,v}^{(t)}=1$, then node $u$ draws $\Big\lceil \frac{x_u^{(t-1)} + x_v^{(t-1)}}{2} \Big\rceil$ tokens from the urn uniformly at random without replacement, and node $v$ receives the remaining tokens. Otherwise,
$\Phi_{u,v}^{(t)}=-1$, and node
$u$ draws $\Big\lfloor \frac{x_u^{(t-1)} + x_v^{(t-1)}}{2}\Big\rfloor$ tokens from the urn and  node $v$ receives the remaining tokens again. We observe that each individual token  located at node $u$ or $v$ at the end of round $t-1$ is assigned to either $u$ or $v$ with probability $1/2$. Note that this token-based process performs  in exactly the same way as the original protocol introduced in \secref{matching}.

We now prove that every token viewed individually performs a random walk with respect to the matching matrices. Henceforth, we use $w_i^{(t)}$ to represent the location (the node) of token $i\in\mathcal{T}$ at the end of round $t$. We also use the the notation that for any $n$ by $n$ matrix $\M$, any node $u \in V$ and subset $D \subseteq V$, $\M_{u,D} := \sum_{v \in D} \M_{u,v}$.

\begin{lemma}\label{lem:rightprobability}
Consider an arbitrary sequence of matchings $\mathcal{M}=\langle \M^{(1)},\M^{(2)},\ldots \rangle$.
Fix any non-negative load vector at the end of round $t_1$, and consider a token $i \in {\cal T}$ located at node $u=w_i^{(t_1)}$ at the end of round $t_1$.
Then, for any $t_2 \geq t_1$,
\[
  \Pro{ w_i^{(t_2)} = v } = \M_{u,v}^{[t_1+1,t_2]},
 \]
and more generally, for any set $D \subseteq V$,
\[
  \Pro{ w_i^{(t_2)} \in D } = \M_{u,D}^{[t_1+1,t_2]}.
\]
\end{lemma}
\begin{proof}
We prove by induction on $t$ that for an arbitrary pair of nodes $u,v \in V$ and round $t \geq t_1$, the probability for a token which is at node $u$ at the end of round $t_1$ to be at node $v$ at the end of round $t$ equals $\M_{u,v}^{[t_1+1,t]}$.
Since $\M^{[t_1+1,t_1]}=\mathbf{I}$, the claim is trivially true for $t=t_1$.
Consider now any  round $t \geq t_1$ for which the induction hypothesis holds.
 If node $v$ is not part of the matching in round $t+1$, then the induction step holds trivially since
 $\M_{u,v}^{[t_1+1,t]} = \M_{u,v}^{[t_1+1,t+1]}$.
  So suppose that node $v$ is matched with a node $k$
in round $t+1$. Since tokens are distributed uniformly, it follows that any token at node $v$ or $k$ (if there are any) in round $t$ will be assigned to node $v$ with probability exactly $1/2$, regardless of whether the sum of tokens is even or not. Notice that tokens  located at other nodes at the end of round $t$ cannot be located at node $v$ at the end of round $t+1$. Therefore,
\begin{align*}
 \Pro{ w_{i}^{(t+1)} = v } &= \frac{1}{2} \cdot   \Pro{ w_{i}^{(t)} = v } + \frac{1}{2} \cdot  \Pro{ w_{i}^{(t)} =k }.
\end{align*}
Using the induction hypothesis, it follows that
\begin{align*}
 \Pro{ w_{i}^{(t+1)} = v } &= \frac{1}{2} \cdot   \M_{u,v}^{[t_1+1,t]} + \frac{1}{2} \cdot   \M_{u,k}^{[t_1+1,t]} \\
 &= \M_{v,v}^{(t+1)} \cdot   \M_{u,v}^{[t_1+1,t]} + \M_{k,v}^{(t+1)} \cdot   \M_{u,k}^{[t_1+1,t]} \\
 &= \M_{u,v}^{[t_1+1,t+1]},
\end{align*}
which completes the induction. The second statement of the lemma follows immediately by summing over all nodes in $D$.
\end{proof}

The next lemma is the crux of our token-based analysis. It shows that the probability of a certain set of tokens being located on a set of nodes $D$ at the end of round $t_2$ is at most the product of the individual probabilities. This negative correlation will enable us to derive a strong version of the Chernoff bound (see~\lemref{chernofftoken}).

\begin{lemma}\label{lem:randomwalklemma}
Consider an arbitrary sequence of matchings $\mathcal{M}=\langle \M^{(1)},\M^{(2)},\ldots \rangle$.
Fix any non-negative load vector at the end of round $t_1\geq 0$, and let ${\cal B} \subseteq {\cal T}$ be an arbitrary subset of tokens. Then, for any subset of nodes $D \subseteq V$ and round $t_2 > t_1$, it holds that
\begin{align*}
 \Pro{  \bigcap_{i \in {\cal B}} \left\{ w_i^{(t_2)} \in D\right\}} &\leq \prod_{i\in\mathcal{B}}\M_{w_i^{(t_1)},D}^{[t_1+1,t_2]} = \prod_{i \in {\cal B}} \Pro{ w_i^{(t_2)} \in D }.
\end{align*}
\end{lemma}
\begin{proof}
We only have to prove the inequality, as the equality follows directly from \lemref{rightprobability}.
Assume for simplicity that all tokens in ${\cal B}$ are numbered from $1$ to $\beta:=|{\cal B}|$. To simplify the notation, we
define for any token $i \in {\cal B}$ and round $t \in [t_1,t_2]$,
\[
z_i^{(t)}:=\M^{[t+1,t_2]}_{w^{(t)}_i,D}
\]
 and
\[
 Z^{(t)}: = \prod_{i=1}^\beta z_{i}^{(t)}.
\]
Our goal is to prove that the sequence $Z^{(t)}$, $t \geq t_1$, forms a supermartingale with respect to the sequence of load vectors $x^{(t)}$, $t \geq t_1$, i.e., it holds for any $t > t_1$ that
\begin{align}
  \Ex{ Z^{(t)} \, \mid \, x^{(t-1)},\ldots,x^{(t_1)} }
  &\leq Z^{(t-1)}. \label{eq:corr}
\end{align}
Assuming that \eq{corr} holds, we can deduce the statement of the lemma as follows:
\begin{align*}
 \Ex{ Z^{(t_2)} } &\leq Z^{(t_1)} = \prod_{i=1}^{\beta} z_i^{(t_1)} = \prod_{i=1}^{\beta} \M_{w_i^{(t_1)},D}^{[t_1+1,t_2]}.
\end{align*}
By definition,
\begin{align*}
Z^{(t_2)} &= \prod_{i=1}^{\beta} \M_{w_i^{(t_2)},D}^{[t_2+1,t_2]},
\end{align*}
which is one if $w_i^{(t_2)} \in D$ for all $i\in\mathcal{B}$, and zero, otherwise. Therefore,
\[
\Ex{ Z^{(t_2)} } = \Pro{ Z^{(t_2)} = 1} = \Pro{ \bigcap_{i \in {\cal B}} \left\{ w_i^{(t_2)} \in D \right\} },
\]
and the proof is complete.

It remains to prove \eq{corr}.
To this end, fix the load vector $x^{(t-1)}$ and partition the set of tokens  $\mathcal{B} = \{1,\ldots,\beta\}$ into
 disjoint sets $S_1,S_2,\ldots,S_{\beta'}$ with $1 \leq \beta' \leq \beta$ so that
  every token in $S_j$ has the same set of possible assignments at the end of  round $t$.
Since tokens with different sets of possible assignments behave independently in round $t$, it follows that
\begin{align}
  \Ex{ Z^{(t)} \, \big| \, x^{(t-1)}, \ldots, x^{(t_1)} } &= \Ex{ \prod_{j=1}^{\beta'} \prod_{i \in S_{j}} z_i^{(t)} \, \bigg| \, x^{(t-1)}, \ldots, x^{(t_1)} } \notag \\
&= \prod_{j=1}^{\beta'} \Ex{ \prod_{i \in S_{j}} z_i^{(t)} \, \bigg| \, x^{(t-1)}, \ldots, x^{(t_1)} }.  \label{eq:correlation}
\end{align}

Hence, in order to prove \eq{corr}, it suffices to prove that for every $1 \leq j \leq \beta'$,
\begin{align}
  \Ex{ \prod_{i \in S_{j}} z_i^{(t)} \, \bigg| \, x^{(t-1)}, \ldots, x^{(t_1)} } &\leq \prod_{i \in S_j} z_i^{(t-1)}. \label{eq:corrtoprove}
\end{align}
Consider first those sets $S_j$ so that every token $i \in S_j$ has only one possible assignment, meaning that node $w_i^{(t-1)}$ is not incident to any matching edge in round $t$. In this case we have $w_i^{(t)} = w_i^{(t-1)}$, $\M_{w_i^{(t-1)},{\cal D}}^{[t,t_2]} = \M_{w_i^{(t-1)},{\cal D}}^{[t+1,t_2]}$,
 and hence, $z_{i}^{(t)} = z_{i}^{(t-1)}$, and, consequently, \eq{corrtoprove} holds.

The second and more involved case concerns those sets $S_j$ for which every token in $S_j$ has two possible assignments at the end of round $t$, denoted by $u=u(j)$ and $v=v(j)$, with $\{u,v\} \in \M^{(t)}$.
 Assume w.l.o.g. that $\M^{[t+1,t_2]}_{u,D}\geq \M^{[t+1,t_2]}_{v,D}$, and  tokens in $S_j$ are numbered from $1$ to $\gamma=|S_j|$. Then, for every token $i \in \{1,\ldots,\gamma\}$, define a random variable $X_i$ as follows:
\begin{equation*}
 X_i =
\begin{cases}
 1 & \mbox{if $w_i^{(t)}=u$}, \\
 0 & \mbox{if $w_i^{(t)}=v.$}
\end{cases}
\end{equation*}
Our claim is that the (random) vector $X=(X_1,\ldots,X_{\gamma}) \in \{0,1\}^{\gamma}$ satisfies the negative regression condition~(cf.~Definition~\ref{def:regression}), i.e.,
for every two disjoint subsets $\mathcal{L}$ and $\mathcal{R}$ of $S_j$ and every non-decreasing function $f \colon \{0,1\}^{|\mathcal{L}|} \to \mathbb{R}$, it holds that
\[
  \Ex{ f(X_q, q \in \mathcal{L}) \, \mid \, X_r = \sigma_r, r \in \mathcal{R} }
\]
is non-increasing in each $\sigma_r\in\{0,1\}, r \in \mathcal{R}$. To establish this, it suffices to show that
\begin{align}
 \Ex{ f(X_q, q \in \mathcal{L}) \, \mid \, X_r = \sigma_r, r \in \mathcal{R} } &\geq
 \Ex{ f(X_q, q \in \mathcal{L}) \, \mid \, X_r = \widetilde{\sigma}_r, r\in \mathcal{R} } \label{eq:toprove},
\end{align}
where $\widetilde{\sigma}_r=\sigma_r$ for every $r \in \mathcal{R}$ except for one $r' \in \mathcal{R}$, where $\widetilde{\sigma}_{r'} > \sigma_{r'}$.
To prove the above inequality, we use a coupling argument. We expose the locations of the tokens in $S_j$ one after another in an arbitrary order. In particular, we may expose the assignments of tokens in $S_j$ before considering the other tokens (the ones not in ${\cal B}$) that are located on $u$ and $v$ at the beginning of round $t$. Note that for every token $i \in S_j$, the probability of being assigned to node $u$ (or $v$) depends on the placement of the previous tokens.
In fact, the exact probability is not required here; instead, we shall only use the fact that the probability for a token to be assigned to node $u$ is non-increasing in the fraction of tokens that have been assigned to $u$ before.
More formally, for any $1 \leq i \leq \gamma+1$, let $\alpha(i)$ be the number of tokens in $\{1,\dots,i-1\}$ that are assigned to node $u$. Hence, if we associate to every token $i \in S_j$ a uniform random variable $U_i \in [0,1]$, then there exists a threshold function $T(i,\alpha(i)) \in [0,1]$ satisfying the following properties:
\begin{enumerate}\itemsep -0pt
 \item If $U_i \geq T(i,\alpha(i))$, then token $i$ is assigned to node $u$.
 \item If $U_i < T(i,\alpha(i))$, then token $i$ is assigned to node $v$.
 \item For any fixed $i$, $T(i,\alpha(i))$ is non-decreasing in $\alpha(i)$.
\end{enumerate}
Without loss of generality, assume that $S_j=\{1,\ldots,\gamma\}$,
$\mathcal{R} = \{ 1, \ldots, r \}$, $\mathcal{L} = \{ r+1,\ldots, r+ \ell \}$, $r+\ell \leq \gamma$. Recall that \eq{toprove} involves two conditional probability spaces, one for $X_r = \sigma_r, r\in\mathcal{R}$   and the other one for $X_r=\widetilde{\sigma}_{r},r\in\mathcal{R}$. We denote these probability spaces by $\Omega$ and $\widetilde{\Omega}$, respectively.

Since these probability spaces are only conditional on the placements of tokens in $\mathcal{R}$, we can couple both probability spaces by assuming that the random variables $U_i$ attain the same values for every $i \in \mathcal{L}$ in $\Omega$ and $\widetilde{\Omega}$. Further,
 let us denote by $\widetilde{\alpha}(i)$ the number of tokens in $\{1,\ldots, i-1\}$, which are placed on node $u$ in
 $\widetilde{\Omega}$. Then, the values $U_i~(i \in \mathcal{L})$, $\alpha(r+1)$, and $\widetilde{\alpha}(r+1)$ determine the placement of all tokens in $\mathcal{L}$ for the two probability spaces. By assumption on $\widetilde{\sigma}$, we have
 $ \alpha(r+1) \in \{ \widetilde{\alpha}(r+1)-1, \widetilde{\alpha}(r+1) \}$.

Now, by the three properties and the coupling described above, if for some $i \in \mathcal{L}$, $\alpha(i) = \widetilde{\alpha}(i)$, then $\alpha(i+1) = \widetilde{\alpha}(i+1)$, and all further tokens in $\mathcal{L}$ are placed in both probability spaces in the same way. Additionally,
 if for some $i \in \mathcal{L}$, $\alpha(i) = \widetilde{\alpha}(i) - 1$, then it follows that $\alpha(i+1) \in\{ \widetilde{\alpha}(i+1) - 1, \widetilde{\alpha}(i+1)\}$ by the third property of the threshold function. This means that
 every token $i \in \mathcal{L}$ that is placed on $u$ in $\widetilde{\Omega}$ will also  be placed on $u$ in $\Omega$.
Since $f$ is non-decreasing in each coordinate (by assumption), the coupling argument above establishes~\eq{toprove}.

By~\eq{toprove}, the vector $X=(X_1,\ldots,X_{\gamma})$ satisfies the negative regression property.
Then, for any $i\in\{1,\dots, \gamma\}$, define a random variable $h(X_i)$ as follows:
$h(X_i)=\M^{[t+1,t_2]}_{u(j),D}$ if $X_i=1$, and $h(X_i) = \M^{[t+1,t_2]}_{v(j),D}$ if $X_i=0$. By the choice of
$u$ and $v$, we know that the function $\prod_{i \in S_j} h(X_i)$ is non-decreasing in every coordinate $X_i$. Hence,
\begin{align*}
 \Ex{ \prod_{i \in S_{j}} z_i^{(t)} \, \bigg| \, x^{(t-1)}, \ldots, x^{(t_1)} } &= \Ex{ \prod_{i \in S_j} h(X_i) \, \bigg| \, x^{(t-1)}, \ldots, x^{(t_1)}} \\ &\hspace{-0.62cm}\stackrel{\text{\lemref{lemmaA8}}}{\leq} \prod_{i \in S_j} \Ex{ h(X_i) \, \big| \, x^{(t-1)}, \ldots, x^{(t_1)}} \\
&=
\prod_{i\in S_j} \left( \frac{1}{2} \M^{[t+1,t_2]}_{u(j),D} + \frac{1}{2} \M^{[t+1,t_2]}_{v(j),D}  \right).
\end{align*}
Since $\M_{u(j),D}^{[t,t_2]} = \frac{1}{2} \M_{u(j),D}^{[t+1,t_2]} + \frac{1}{2} \M_{v(j),D}^{[t+1,t_2]} = \M_{v(j),D}^{[t,t_2]}$, we arrive at
\[
\Ex{ \prod_{i\in S_j} z_i^{(t)} \, \bigg| \, x^{(t-1)}, \ldots, x^{(t_1)} }  \leq
\prod_{i\in S_j} \M^{[t,t_2]}_{w_i^{(t-1)},D}
=\prod_{i\in S_j} z_i^{(t-1)}.
\]
Applying this to \eq{correlation} for every $1 \leq j \leq \beta'$ implies that
\begin{align*}
 \Ex{ Z^{(t)} \, \big| \, x^{(t-1)}, \ldots, x^{(t_1)} } &= \prod_{j=1}^{\beta'} \Ex{ \prod_{i \in S_{j}} z_i^{(t)} \, \bigg| \, x^{(t-1)}, \ldots, x^{(t_1)} } \leq \prod_{j=1}^{\beta'} \prod_{i \in S_j}  z_i^{(t-1)}  = Z^{(t-1)},
\end{align*}
showing that $Z^{(t)}$ is indeed a supermartingale.
This establishes \eq{corr} and finishes the proof of the lemma.
\end{proof}

Combining \lemref{randomwalklemma} and \lemref{panconesi}, we  directly obtain the following Chernoff bound:
\begin{lemma}\label{lem:chernofftoken}
Consider an arbitrary sequence of matchings $\mathcal{M}=\langle \M^{(1)},\M^{(2)},\ldots \rangle$.
Fix any non-negative load vector $x^{(t_1)}$ at the end of round $t_1\geq 0$, and let ${\cal T}$ be the set of all tokens. Let $D$ be any subset of nodes and $t_2 > t_1$. Then, for the random variable
\[
Z:=  \sum_{i \in {\cal T}} \chi_{w_i^{(t_2)} \in D} = \sum_{u \in D} x_u^{(t_2)},
\]
it holds for any $\delta > 0$ that
\begin{align*}
  \Pro{	Z \geq (1+\delta) \Ex{Z} } &\leq \left(   \frac{\ce^{\delta}}{(1+\delta)^{1+\delta}}   \right)^{\Ex{Z}}.
\end{align*}
\end{lemma}
The strength of \lemref{chernofftoken} is that the sum of loads is analyzed by means of a sum of indicator random variables over all tokens instead of a sum of rounding errors as in \lemref{bounddifference}, for instance. For an illustration of the power of \lemref{chernofftoken}, we consider the following toy example.

\begin{corollary}\label{cor:verysparse}
Consider the random matching model or balancing circuit model.
Let $x^{(0)}$ be any non-negative load vector with $\| x^{(0)} \|_{1} \leq n^{1-\eps}$, where $0 < \eps < 1$ is a constant. Then,  the discrepancy after $\tau_{\cont}(1,n^{-1})$ rounds is at most $5/\eps$ with probability at least $ 1-2 n^{-1}$.
\end{corollary}

We can think of the allocation of the $\| x^{(0)} \|_{1}$ tokens in terms of the popular balls-and-bins model \cite{MU05}. If we run our randomized protocol for sufficiently many rounds, say $\tau_{\cont}(1,n^{-1})$ rounds, then every token (corresponding to a ball) is located at any node (corresponding to a bin) with almost the same probability. While in the standard balls-and-bins model, the allocation of different balls are mutually independent, \lemref{randomwalklemma} established  that these allocations are negatively correlated in our model. Therefore we obtain a constant maximum load if the number of tokens is bounded by $n^{1-\epsilon}$, which is a well-known fact for the balls-and-bins model.
We now give the proof of \corref{verysparse}.


\begin{proof}
Fix any node $u \in V$. By definition, it
holds for  $t := \tau_{\cont}(1,n^{-1})$ that the time-interval $[0,t]$
is $(1,n^{-1})$--smoothing with probability at least $1-n^{-1}$ (noticing that this probability is even $1$ for the balancing circuit model). Define $Z:=\sum_{i \in \cal{T}} \chi_{w_i^{(t)} = u}$ as the number of tokens located on node $u$ at the end of round $t$.
Assuming that $[0,t]$ is $(1,n^{-1})$--smoothing, the second statement of \lemref{boundbymixinglemma} implies that \[
\Pro{ w_i^{(t)} = u } = \M_{w_i^{(0)},u}^{[1,t]} \leq 2/n\]
 for any token $i$.
Hence,
\[
\Ex{Z} \leq \| x^{(0)} \|_{1} \cdot 2/n \leq 2 n^{-\epsilon} < 1.
\]
 By applying \lemref{chernofftoken} with $\delta = \frac{4}{\epsilon} \cdot \frac{1}{\Ex{Z}}$, we obtain that
\begin{align*}
 \Pro{ Z \geq \frac{5}{\epsilon} } &\leq \left( \frac{\mathrm{e}}{\delta} \right)^{\delta \cdot \Ex{Z}} \leq \left( \epsilon \Ex{Z}  \right)^{ 4/\epsilon   } \leq n^{-2}.
\end{align*}
Taking the union bound over all nodes yields the claim.
\end{proof}

The next lemma provides a concrete tail bound, which is not only exponentially small in the deviation from the mean, but also exponentially small in the ``sparseness'' of the load vector. By contrast, previous analyses expressing the load as a sum of rounding errors \cite{HT06,MS10,FS09,BCFFS11,RSW98} yield weaker tail bounds for sparse load vectors (cmp.~\lemref{bounddifference}).
Another advantage of \lemref{sparsiter} is that
it gives a tail bound for an {\em arbitrary convex combination} of the load vector.

\begin{lemma}\label{lem:sparsiter}
Consider an arbitrary sequence of matchings $\mathcal{M}=\langle \M^{(1)},\M^{(2)},\ldots \rangle$ so that $[t_1,t_2]$ is $(1,n^{-1})$--smoothing for a pair of rounds $t_2 \geq t_1$.
Fix any non-negative load vector $x^{(t_1)}$  with $\|x^{(t_1)}\|_{1} \leq n \cdot \ce^{-(\log n)^{\sigma}}$ for some constant $\sigma \in (0,1)$.
Moreover,  let
\[
Z:= \sum_{v \in V} y_v x_v^{(t_2)} ,
\]
 where $y$ is any non-negative vector with $\|y \|_{1} = 1$. Then, it holds for any $\delta > 0$ that
\begin{align*}
 \Pro{ Z \geq \ce^{-\frac{1}{5} (\log n)^{\sigma}} + 8 \|  y \|_{\infty} \cdot (\log n)^{\delta} } &\leq \ce^{-\frac{1}{3} (\log n)^{\delta+\sigma}}.
\end{align*}
\end{lemma}
\begin{proof}
Let $\alpha:= \left\| y \right\|_{\infty}$. Partition $V$ into at most $\lceil 2 \log_2 n \rceil + 1$ groups, defined as follows:
\begin{align*}
 S_j &:= \left\{ v \in V \colon 2^{-j-1} < y_v \leq 2^{-j}  \right\}, \quad \lfloor \log_2 (1/\alpha) \rfloor  \leq j < \lceil 2 \log_2 n \rceil, \\
 S_{\lceil 2 \log_2 n \rceil} &:= \left\{ v \in V \colon  y_v \leq 2^{-\lceil 2 \log_2 n \rceil}  \right\}.
\end{align*}
Clearly, $|S_j| \leq 2^{j+1}$ for every $j$, which also holds for $S_{\lceil 2 \log_2 n \rceil}$ since $|S_{\lceil 2 \log_2 n \rceil}| \leq |V|=n$. In order to obtain a bound on $Z$, we will upper bound the following approximation of $Z$:
\begin{align*}
  \tilde{Z} &:= \sum_{j=\lfloor \log_2 (1/\alpha) \rfloor}^{\lceil2 \log_2 n\rceil} \sum_{v \in S_j} x_{v}^{(t_2)} \cdot 2^{-j}.
\end{align*}
By definition, $Z \leq \tilde{Z}$.
We now apply our new relation between the movements of tokens and independent random walks in order to upper bound $\tilde{Z}$. We do this by considering the contribution from each $S_j$ individually. Since $[t_1,t_2]$ is $(1,n^{-1})$--smoothing, the second statement of \lemref{boundbymixinglemma} implies that every token is located at any node in round $t_2$ with probability at most $2/n$; thus, every token is located at a node in $S_j$ with probability at most $2|S_j|/n$.

Suppose first that $|S_j| \geq \ce^{ \frac{1}{2} (\log n)^{\sigma}}$.
By \lemref{randomwalklemma}, the probability that we have more than $4 |S_j| \cdot \ce^{-\frac{1}{4} (\log n)^{\sigma}}$ tokens on nodes in $S_j$ in round $t_2$ is upper bounded by
\[
  \binom{ n \cdot \ce^{-(\log n)^{\sigma}}}{ 4 |S_j| \cdot \ce^{-\frac{1}{4} (\log n)^{\sigma}}} \cdot \left(  \frac{2 |S_j|}{n}  \right)^{ 4 |S_j| \cdot \ce^{-\frac{1}{4} (\log n)^{\sigma}}} \leq \left( \frac{ \ce \cdot \ce^{- \frac{3}{4} (\log n)^{\sigma}}}{2}		 \right)^{4 |S_j| \cdot \ce^{-\frac{1}{4} (\log n)^{\sigma}}} = n^{-\omega(1)},
\]
since for any two integers $1 \leq y \leq x$, $\binom{x}{y} \leq \big( \frac{\ce x}{y} \big)^{y}$ (the same bound on the probability holds trivially when $ 4 |S_j| \cdot \ce^{-\frac{1}{4} (\log n)^{\sigma}}$ is larger than $n \cdot \ce^{-(\log n)^{\sigma}}$). Next, assume that $|S_j| \leq \ce^{\frac{1}{2} (\log n)^{\sigma}}$.
 For every token $i \in \{1,\ldots, \| x^{(t_1)} \|_{1} \}$,
 define $X_{i,j} = \chi_{w_i^{(t_2)} \in S_j}$, that is,
  $X_{i,j} = 1$, if token $i$ is located at a node in $S_j$ at the end of round $t_2$, and $X_{i,j} = 0$, otherwise. Let $X_j :=\sum_{i=1}^{\|x^{(t_1)} \|_1} X_{i,j}$. Then,
\begin{align*}
   \Ex{ X_j } &\leq \left\| x^{(t_1)} \right\|_{1} \cdot \frac{2|S_j|}{n} \leq  \ce^{-(\log n)^{\sigma}} \cdot 2 \ce^{\frac{1}{2} (\log n)^{\sigma}} = 2 \ce^{-\frac{1}{2} (\log n)^{\sigma}}.
\end{align*}
Using the Chernoff bound (\lemref{chernofftoken}),
\begin{align*}
 \Pro{ X_j \geq (1 + \beta) \Ex{X_j} } &\leq \left( \frac{\mathrm{e}}{\beta} \right)^{\beta \cdot \Ex{X_j}}.
\end{align*}
Choosing
$\beta := (\log n)^{ \delta} / \Ex{X_j}$, we conclude that
\begin{align*}
 \Pro{ X_j \geq 2\cdot (\log n)^{\delta} } &\leq \left( \Ex{X_j} \right)^{ (\log n)^{\delta}  } \leq \left(2 \ce^{-\frac{1}{2} (\log n)^{\sigma}} \right)^{ (\log n)^{\delta} }.
\end{align*}

By the union bound over at most $\lceil 2 \log_2 n \rceil + 1$ groups, we conclude that with probability at least
\[
1- (\lceil 2 \log_2 n \rceil + 1) \cdot \max\left\{n^{-\omega(1)}, \left(2 \ce^{-\frac{1}{2} (\log n)^{\sigma}} \right)^{ (\log n)^{\delta} } \right\} \geq 1 - \ce^{-\frac{1}{3}(\log n)^{\delta+\sigma}} ,
\]
\begin{align*}
 Z \leq \sum_{j=\lfloor \log_2 (1/\alpha) \rfloor}^{\lceil 2 \log_2 n \rceil}
 \sum_{v\in S_j} x_v^{(t_2)} \cdot 2^{-j} &= \sum_{j=\lfloor \log_2 (1/\alpha) \rfloor}^{\lceil 2 \log_2 n \rceil} X_j \cdot 2^{-j} \\ &\leq \sum_{j=\lfloor \log_2 (1/\alpha) \rfloor}^{\lceil 2 \log_2 n \rceil} 2^{-j} \cdot \left( \frac{4| S_j|}{ \ce^{\frac{1}{4} (\log n)^{\sigma}} } + 2\cdot (\log n)^{\delta} \right) \\ &\leq
 \sum_{j=\lfloor \log_2 (1/\alpha) \rfloor}^{\lceil 2 \log_2 n \rceil} 2^{-j} \cdot  \frac{4 \cdot 2^{j+1} }{  \ce^{\frac{1}{4} (\log n)^{\sigma}} } + \sum_{j=\lfloor \log_2 (1/\alpha) \rfloor}^{\lceil 2 \log_2 n \rceil} 2^{-j+1} \cdot (\log n)^{\delta} \\
  &\leq \frac{8\cdot \lceil 2 \log_2 n \rceil}{ \ce^{\frac{1}{4} (\log n)^{\sigma}} } + 8 \alpha \cdot(\log n)^{\delta}\\
  &\leq \ce^{-\frac{1}{5} (\log n)^{\sigma}}  + 8\alpha \cdot(\log n)^{\delta}.
  \end{align*}
\end{proof}

\subsection{Bounding the Discrepancy in Arbitrary Graphs}\label{sec:randomwalkresults}

Throughout this subsection,
we consider the random matching or
balancing circuit model. We assume, without loss of generality, that  $x^{(0)}\in\mathbb{Z}^{n}$ is any initial load vector with  $\overline{x}\in[0,1)$ (cf.~Observation~\ref{obs:scaling} for a justification).
%
%
%
Let us fix any value $\epsilon > 0$, not necessarily constant. Then, define the following set of vectors for any $\ell \geq 1$:
\[
  \mathcal{E}_{\ell} := \left\{ x \in \mathbb{Z}^n \colon  \sum_{u \in V} \max \left\{ x_u  - 8 \ell \cdot \lceil (\log n)^{\eps} \rceil - \ell, 0 \right\}    \leq 4 n \cdot \ce^{-\frac{1}{4} \cdot (\log n)^{\ell \epsilon} }       \right\}.
\]
Roughly speaking, $\mathcal{E}_{\ell}$ includes all load vectors
in which the number of tokens exceeding  $8 \ell \cdot \lceil (\log n)^{\eps} \rceil+\ell$ for every node is not too large. In particular,
for any load vector $x\in\mathcal{E}_{\ell}$, $\ell \geq \lceil 2/\eps \rceil$,  the maximum load of $x$ is at most $8 \ell \cdot \lceil (\log n)^{\eps} \rceil + \ell$.

The next lemma shows that if we start with a load vector in $\mathcal{E}_{\ell-1}$, then the load vector after $\tau_{\cont}(1,n^{-2})$
rounds will be in $\mathcal{E}_{\ell}$ with high probability.

\begin{lemma}\label{lem:reducinglemma}
For any integer $\ell \geq 2$, round $t \in \mathbb{N}$, $\epsilon \geq 16/(\log \log n)$ and any vector $x \in \mathcal{E}_{\ell-1}$, it holds that
\begin{align*}
  \Pro{  x^{(t + \kappa) } \in   \mathcal{E}_{\ell}    \, \mid \,  x^{(t)} = x }  &\geq  1 - \ce^{- \frac{1}{4} (\log n)^{\ell \epsilon}} - n^{-1},
\end{align*}
where $\kappa:=\tau_{\cont}(1,n^{-2})$. In particular, $\Pro{ x^{(\kappa)} \in \mathcal{E}_{1} } \geq 1 - \ce^{- \frac{1}{4} (\log n)^{\epsilon} } - 3 n^{-1}$, if $\kappa := \tau_{\cont}(K,1/(2n))$.
\end{lemma}
Let us briefly describe the key steps in the proof of \lemref{reducinglemma}.
The proof that $x^{(\kappa)}\in\mathcal{E}_1$ makes use of the concentration inequality for the sum of rounding errors (\lemref{bounddifference}).
Based on \lemref{chernofftoken}, we prove that, starting with a load vector in $\mathcal{E}_{\ell-1}$, we obtain a load vector which is in $\mathcal{E}_{\ell}$ after $\kappa$ additional rounds.
\begin{proof}
Recall that we assume here that $\overline{x} \in [0,1)$.
Let us first consider the event $x^{(\kappa)}\in\mathcal{E}_1$. Consider the following potential function in round $\kappa$:
\[
\Phi^{(\kappa)} := \sum_{u\in V} \exp\left( \left(x_u^{(\kappa)} -\overline{x} \right)^2 \Big\slash 16\right).
\]
Since $\kappa = \tau_{\cont}(K,1/(2n))$, it follows that with probability at least $1-n^{-1}$, the time-interval $[0,\kappa]$ is $(K,1/(2n))$--smoothing, which we will condition on in the remainder of the proof (Noticing  that this probability is $1$ in the balancing circuit model).
By the second statement of \lemref{bounddifference}, it holds for any node $u\in V$ and any $\delta > 1/n$ that
\begin{equation}
  \Pro{ \left|x_u^{(\kappa)} - \overline{x}\right| \geq \delta } \leq
  2 \cdot \exp\left( - \left(\delta -\frac{1}{2n} \right)^2 \bigg\slash 4 \right),\label{eq:generalgraphpr}
 \end{equation}
and therefore,
\begin{align*}
 \Ex{ \Phi^{(\kappa)} } &\leq n \cdot \max_{u\in V} \Ex{ \left\lceil \exp\left( \left(x_u^{(\kappa)} -\overline{x} \right)^2 \Big\slash 16\right) \right\rceil } \\
    &= n \cdot \max_{u \in V} \sum_{k=1}^{\infty} \Pro{ \left\lceil \exp\left( \left(x_u^{(\kappa)}  -\overline{x} \right)^2 \Big\slash 16 \right) \right\rceil  \geq k } \\
  &\leq n \cdot \max_{u \in V} \left(3+ \sum_{k=4}^{\infty} \Pro{ \left( x_u^{(\kappa)} - \overline{x} \right)^2 \Big\slash 16 \geq \log (k-1) } \right) \\
  &= n \cdot \max_{u \in V} \left( 3+ \sum_{k=3}^{\infty} \Pro{ \left| x_u^{(\kappa)} - \overline{x} \right| \geq 4 \sqrt{\log k} } \right) \ .
    \end{align*}
Combining this with \eq{generalgraphpr}, we get
\begin{align*}
  \Ex{ \Phi^{(\kappa)} } &\leq n \cdot \max_{u \in V} \left( 3+ \sum_{k=3}^{\infty} 2\cdot \exp\left( -
  \left(4\sqrt{\log k} - \frac{1}{2n}\right)^2\Big\slash 4 \right) \right) \\
  & \leq n\cdot \left( 3 + \sum_{k=3}^{\infty} 2\cdot \mathrm{e}^{-3\log k} \right) \leq 4n,
\end{align*}
where in the last inequality, we used the fact that $\sum_{k=3}^{\infty} 2 k^{-3} \leq 1$.
Hence, by Markov's inequality,
\begin{equation}\label{eq:markovphit1}
 \Pro{ \Phi^{(\kappa)} \geq 4n \cdot \mathrm{e}^{ \frac{1}{4} \cdot  (\log n)^{\eps}   } } \leq
 \Pro{ \Phi^{(\kappa)} \geq \mathrm{e}^{ \frac{1}{4} \cdot  (\log n)^{\eps}} \cdot \Ex{ \Phi^{(\kappa)} }   }
 \leq \mathrm{e}^{-\frac{1}{4} \cdot (\log n)^{\eps}  }\ .
\end{equation}
Furthermore, recall that by the first statement of \thmref{deviationmatching}, the maximum load at the end of round $\kappa$ is upper bounded by $\sqrt{8 \log n} + 2$ with probability at least $1-2 n^{-1}$.
If we assume that both $\Phi^{(\kappa)} \leq 4n \cdot \mathrm{e}^{\frac{1}{4} \cdot  (\log n)^{\eps}  }$ and $x^{(\kappa)}_{\max} \leq \sqrt{8 \log n} + 2$ hold,
then
\begin{equation}\label{eq:boundtoken}
  \sum_{u \in V} \max\left\{  x_u^{(\kappa)} - 8 \cdot \lceil (\log n)^{\eps} \rceil-1, 0 \right\} \leq   \frac{4 n \cdot \mathrm{e}^{\frac{1}{4} \cdot (\log n)^{\eps} } }{ \mathrm{e}^{4 \cdot (\log n)^{2\eps}}} \cdot \left( \sqrt{8 \log n} + 2 \right) \leq 4n \cdot \mathrm{e}^{- (\log n)^{\eps} },
\end{equation}
where the first inequality  is due to
the fact that every node with load more than $8\lceil (\log n)^{\varepsilon}\rceil+1$ contributes at least
\[
\exp\left(  \left( 8\cdot\lceil (\log n)^{\varepsilon}\rceil \right)^2 \big\slash 16 \right) \geq
\ce^{ 4 \cdot (\log n)^{2\eps}}
\]
to $\Phi^{(\kappa)}$.
Combining \eq{markovphit1} and \eq{boundtoken} yields
\begin{align*}
  \Pro{   \sum_{u \in V} \max \left\{ x_u^{(\kappa)} - 8 \cdot \lceil (\log n)^{\eps} \rceil - 1, 0 \right\}    \geq 4 n \cdot \ce^{- (\log n)^{\epsilon} }         } &\leq \ce^{-\frac{1}{4}\cdot (\log n)^{\epsilon}  } + 2n^{-1} + n^{-1},
\end{align*}
where the final term $n^{-1}$ corresponds to the event that $[0,\kappa]$ is $(K,1/(2n))$-smoothing. This implies that
\[
  \Pro{ x^{(\kappa)} \in \mathcal{E}_{1} }  \geq  1 - \ce^{- \frac{1}{4} \cdot (\log n)^{\epsilon }  } - 3 n^{-1},
\]
completing the proof of the base case.

For the statement with $\ell \geq 2$, we consider the probability space conditioned on $x^{(t)}= x\in\mathcal{E}_{\ell-1}$.
To simplify the notation, we will omit this condition in the following probabilities and expectations. In order to analyze the load vector $x^{(s)}$, $s \geq t$, we consider an auxiliary load vector $\widetilde{x}^{(s)}$, $s \geq t$, which is initialized in round $t$ by
\[
\widetilde{x}^{(t)}_u:=\max\left\{ x^{(t)}_u- 8 (\ell-1) \cdot \lceil (\log n)^{\varepsilon} \rceil - (\ell-1) ,0\right\} \mbox{~~ for every $u \in V$.}
\]
For any $s > t$, $\widetilde{x}^{(s)}$ is the load vector obtained by the execution of the discrete load balancing protocol starting with $\widetilde{x}^{(t)}$ in round $t$  which uses in every round $s > t$ the same matchings and orientations as for the load vector $x^{(s)}$.
We define another set of vectors $\tilde{\mathcal{E}_{\ell}}$ by
\[
 \tilde{\mathcal{E}_{\ell}} := \left\{ \widetilde{x} \in \mathbb{N}^{n} \colon \sum_{u\in V} \max\left\{ \widetilde{x}_u - 8 \cdot \lceil (\log n)^{\eps} \rceil -1, 0 \right\} \leq
  4 n \cdot \ce^{ - \frac{1}{4} (\log n)^{\ell \cdot \epsilon}} \right\}.
\]
We shall prove that
\begin{align}
  \Pro{ \widetilde{x}^{(t+\kappa)} \in \tilde{\mathcal{E}_{\ell}}  } \geq 1 - \mathrm{e}^{-\frac{3}{4}(\log n)^{\ell \cdot \epsilon}} - n^{-1} \label{eq:sparsecrucial}.
\end{align}
Before proving \eq{sparsecrucial}, we show
that $ \widetilde{x}^{(t+\kappa)} \in \tilde{\mathcal{E}_{\ell}} $
 implies $x^{(t+\kappa)} \in \mathcal{E}_{\ell}$. Define another load vector $\widehat{x}^{(s)}$, $s \geq t$, which is initialized in round $t$ by
\[
 \widehat{x}_u^{(t)} := \max \{ x_u^{(t)}, 8 (\ell-1) \cdot \lceil (\log n)^{\varepsilon} \rceil - (\ell-1)  \} \mbox{~~ for every $u \in V$.}
\]
Note that
\begin{align}
 \widetilde{x}_u^{(t)} &= \widehat{x}_u^{(t)} - 8 (\ell-1) \cdot \lceil (\log n)^{\varepsilon} \rceil - (\ell-1). \label{eq:simplee}
\end{align}
Hence,
\begin{align*}
 \sum_{u\in V} \max\left\{ x_u^{(t+\kappa)} - 8 \ell \cdot \lceil (\log n)^{\eps} \rceil - \ell, 0 \right\} &\leq
  \sum_{u\in V} \max\left\{ \widehat{x}_u^{(t+\kappa)} - 8 \ell \cdot \lceil (\log n)^{\eps} \rceil - \ell, 0 \right\} \\
&= \sum_{u\in V} \max\left\{ \widetilde{x}_u^{(t+\kappa)} - 8 \cdot \lceil (\log n)^{\eps} \rceil -1, 0 \right\} \\ &\leq
  4 n \cdot \ce^{ - \frac{1}{4} (\log n)^{\ell \cdot \epsilon}},
\end{align*}
where the first line follows from the second statement of \obsref{scaling}, and the second line follows from \eq{simplee} and the first statement of \obsref{scaling}. Hence, we conclude that
\begin{align*}
 \Pro{  x^{(t+\kappa)} \in   \mathcal{E}_{\ell}  } &\geq 1 -  \ce^{ - \frac{3}{4} (\log n)^{\ell \cdot \epsilon}} - n^{-1} \geq 1 -
 \ce^{ - \frac{1}{4} (\log n)^{\ell \cdot \epsilon}},
\end{align*}
which finishes the proof once \eq{sparsecrucial} is established.

It remains to prove  \eq{sparsecrucial}. We focus on the non-negative load vector $\widetilde{x}^{(s)}$, $s \geq t$, for the remainder of the proof. Fix an arbitrary node $u \in V$.
Let $Z_i = \chi_{w_i^{(t+\kappa)}=u}$ be the 0/1-indicator random variable for every token $i$, which is one if and only if token $i$ reaches node $u$ at the end of round $t+\kappa$.
Let $\beta:=\|\widetilde{x}^{(t)}\|_1$ be the number of tokens,  and  $Z:=\sum_{i=1}^{\beta} Z_i$. Clearly, $\widetilde{x}_{u}^{(t+\kappa)} = Z$. Further, for every token $i$, we have by \lemref{rightprobability} that
\[
  \Pro{ Z_i = 1} = \M_{w_i^{(t)},u}^{[t+1,t+\kappa]}.
\]
Since $\kappa = \tau_{\cont}(1,n^{-2})$, the time-interval $[t,t+\kappa]$ is $(1,n^{-2})$--smoothing with probability at least $1-n^{-1}$, which we will assume in the following. Hence, \lemref{boundbymixinglemma} yields for every pair of nodes $v,u \in V$,
\[
  \M_{v,u}^{[t+1,t+\kappa ]} \leq \frac{1}{n} + \frac{1}{n^2}.
\]

By definition of the set $\mathcal{E}_{\ell-1}$, we have that
\[
   \beta = \left\| \widetilde{x}^{(t)}  \right\|_{1} = \sum_{u \in V} \max \left\{ x_u^{(t)} - 8(\ell-1) \cdot \lceil (\log n)^{\eps} \rceil - (\ell-1), 0 \right\}    \leq 4 n \cdot \ce^{- \frac{1}{4}  (\log n)^{(\ell-1)\cdot\eps }}.
\]
Hence,
\begin{align}
 \Ex{Z}  \leq  \sum_{i=1}^{\beta} \left( \frac{1}{n} + \frac{1}{n^2} \right) &\leq 5\ce^{ - \frac{1}{4}  (\log n)^{(\ell-1)\cdot\eps} } \label{eq:sixteen}.
\end{align}
By \lemref{chernofftoken}, we can upper bound $Z$ as follows:
For any $\delta >0$,
\begin{align}
 \Pro{ Z \geq (1+\delta) \Ex{Z} } &\leq \left(\frac{\ce }{ \delta}\right)^{\delta\cdot \Ex{Z}}. \label{eq:boundZ}
\end{align}
Since $\Ex{Z} \leq 1$,
\begin{align}
  \Pro{ \widetilde{x}_{u}^{(t+\kappa)} \geq \delta \, \Ex{Z} + 1} & = \Pro{ Z \geq  \delta \, \Ex{Z} + 1} \leq \Pro{Z\geq (1+\delta) \Ex{Z} } \notag.
   \end{align}
Choosing $\delta = \delta(\alpha) =  \frac{1}{\Ex{Z}} \cdot ( 8 \cdot (\log n)^{\eps} + \alpha)$ (for any integer $\alpha \geq 0$) in \eq{boundZ}, we obtain
\begin{align}
     \Pro{ \widetilde{x}_{u}^{(t+\kappa)} \geq (8 \cdot (\log n)^{\eps} + \alpha ) + 1 }  &\leq \left(\frac{\ce}{\delta}\right)^{  8 \cdot (\log n)^{\eps} + \alpha         } \notag \\
      &\leq \left(  \frac{\Ex{Z}}{5}       \right)^{ 8 \cdot (\log n)^{\eps} + \alpha           }   \notag \\
       &\hspace{-1em}\stackrel{\text{by~(\ref{eq:sixteen})}}{\leq} \exp
      \left( - \frac{1}{4} \cdot (\log n)^{(\ell-1) \cdot \eps} \cdot \left( 8 \cdot (\log n)^{\eps} + \alpha  \right)\right), \label{eq:probupperbound}
\end{align}
where the second inequality follows by the assumption that $\epsilon\geq 16 /(\log\log n)$.  Our goal is now to bound the number of tokens in the load vector $\tilde{x}^{(t + \kappa )}$ that are
above the threshold $8 \cdot \lceil (\log n)^{\eps} \rceil +1$. To this end,
define the potential function $\Lambda^{(t+\kappa)}$ with respect to load vector $\widetilde{x}^{(t+k)}$ by
\[
\Lambda^{(t+\kappa)}:=\sum_{u\in V} \max\left\{ \widetilde{x}^{(t+\kappa)}_u - 8  \cdot \lceil (\log n)^{\eps} \rceil  -1 , 0\right\}.
\]
Then, we can upper bound the expectation of $\Lambda^{(t+\kappa)}$ as follows:
\begin{align*}
\Ex{ \Lambda^{(t+\kappa)} }
  &=\sum_{u \in V} \sum_{\alpha=1}^{\infty} \Pro{ \max\left\{ \widetilde{x}^{(t+\kappa)}_u - 8  \cdot \lceil (\log n)^{\eps} \rceil  -1 , 0\right\} \geq \alpha          }  \\
  &=\sum_{u \in V} \sum_{\alpha=1}^{\infty} \Pro{ \widetilde{x}_u^{(t+\kappa)} \geq 8 \cdot \lceil (\log n)^{\eps} \rceil +1  + \alpha}  \\
 &\hspace{-1em}\stackrel{\text{by~(\ref{eq:probupperbound})}}{\leq} \sum_{u \in V} \sum_{\alpha = 1}^{\infty} \exp \left(- \frac{1}{4} \cdot (\log n)^{(\ell - 1) \cdot \epsilon } \cdot ( 8 \cdot (\log n)^{\epsilon} + \alpha )  \right)\enspace.
 \end{align*}
 Since $\frac{1}{4}\cdot (\log n)^{(\ell-1)\cdot \epsilon}\geq 1$ by assumption on $\epsilon$, we have
 \begin{align*}
 \Ex{ \Lambda^{(t+\kappa)} } &\leq n \cdot \exp \left( -  (\log n)^{(\ell-1) \cdot \epsilon  + \epsilon} \right) \cdot \sum_{\alpha=1}^{\infty}
 \ce^{ - \alpha}  \\
  &\leq n \cdot \ce^{-  (\log n)^{ \ell \cdot \epsilon} } \cdot \frac{1}{1 - \ce^{-1}} \leq 4n \cdot \ce^{ - (\log n)^{\ell \cdot \epsilon}}.
\end{align*}
By Markov's inequality, we have that
\begin{align*}
 \Pro{ \Lambda^{(t+\kappa)} \geq 4 n \cdot \ce^{ - \frac{1}{4} (\log n)^{\ell \cdot \epsilon} } } &\leq\ce^{- \frac{3}{4} (\log n)^{\ell \cdot \epsilon} }.
\end{align*}
Hence, by definition of $\tilde{\mathcal{E}}_{\ell}$, it holds that
\begin{align*}
 \Pro{ \widetilde{x}^{(t+\kappa)} \in \tilde{\mathcal{E}_{\ell}}  } \geq 1 - \mathrm{e}^{-\frac{3}{4}(\log n)^{\ell \cdot \epsilon}},
\end{align*}
under the assumption that $[t,t+\kappa]$ is $(1,n^{-2})$-smoothing, which holds with probability at least $1-n^{-1}$. This establishes \eq{sparsecrucial} and completes the proof of the theorem.
\end{proof}

Iterating \lemref{reducinglemma} reveals an interesting tradeoff.
First, we obtain a discrepancy of $\Oh( (\log n)^{\eps})$ for an arbitrarily small constant $\epsilon > 0$ by increasing the runtime $\tau_{\cont}(K,n^{-2})$ only by a constant factor. Furthermore, choosing $\epsilon = \Theta(1/(\log \log n))$ and $\ell$ appropriately, we obtain a discrepancy of $\Oh( \log \log n)$ by increasing the runtime by a factor of $\Oh(\log \log n)$.
\begin{theorem}\label{thm:logg}
Let $G$ be any graph, and consider the random matching or balancing circuit model. The following statements hold:
\begin{itemize}
 \item Let $\epsilon > 0$ be an arbitrarily small constant.
Then,
after $\Oh(\tau_{\cont}(K,n^{-2}))$ rounds, the discrepancy is  $\Oh( (\log n)^{\eps} )$ with probability at least $1-\ce^{-\frac{1}{5} (\log n)^{\eps}}$.
\item After $\Oh( \tau_{\cont}(K,n^{-2}) \cdot \log \log n)$ rounds, the discrepancy is  $\Oh(\log \log n)$ with probability at least $1-\frac{1}{\log n}$.
\end{itemize}
\end{theorem}

\begin{proof}
We start with the proof of the first statement.
By~\lemref{reducinglemma} with \[\kappa:=\tau_{\cont}(K,n^{-2})\geq
\max\{\tau_{\cont}(K,1/(2n)),\tau_{\cont}(1,n^{-2}) \},\]
 for any vector $x \in \mathcal{E}_{\ell-1}, \ell \geq 2$ and any round $t \in \N$, it holds that
\begin{align}
 \Pro{ x^{(t+\kappa)} \in \mathcal{E}_{\ell} \, \mid \, x^{(t)}=x } &\geq 1 - \ce^{- \frac{1}{4} (\log n)^{\ell \cdot \epsilon }} - 3 n^{-1} \notag \\ &\geq 1 - \ce^{ - \frac{1}{4} (\log n)^{\epsilon }} - 3 n^{-1} =: p, \label{eq:probsmall}
\end{align}
and the same lower bound also holds for $\Pro{ x^{(\kappa)} \in \mathcal{E}_{1} }$.
Our goal is to show that $x^{(\ell \cdot \kappa)}$ is in $\mathcal{E}_{\ell}$, where $\ell := \lceil \frac{2}{\epsilon} \rceil$. Applying \eq{probsmall} $\ell$-times and the union bound,
\begin{align*}
  \Pro{ x^{(\ell \cdot \kappa)} \in \mathcal{E}_{\ell} } \geq 1 - \ell \cdot \left( \ce^{ - \frac{1}{4} (\log n)^{\epsilon} } + 3 n^{-1} \right) \geq
   1 - \frac{1}{2} \cdot \ce^{ - \frac{1}{5} (\log n)^{\epsilon}},
\end{align*}
where the second inequality holds since $\eps$ and $\ell$ are constants.
If the load vector $x^{(\ell \cdot \kappa)}$ is in $\mathcal{E}_{\ell}$, then
\begin{align*}
  \sum_{u \in V} \max \left\{ x_u^{(\ell \cdot \kappa)} - 8 \left\lceil \frac{2}{\epsilon} \right\rceil \cdot \lceil (\log n)^{\eps} \rceil - \left\lceil \frac{2}{\epsilon} \right\rceil , 0  \right\} &\leq 4 n \cdot \exp\left(-\frac{1}{4} \cdot (\log n)^{\lceil \frac{2}{\epsilon} \rceil\cdot \epsilon} \right) < 1,
\end{align*}
which implies that the maximum load in round $\ell \cdot \kappa$ is $\Oh( (\log n)^{\eps} )$. The corresponding lower bound on the minimum load follows by symmetry (\lemref{maxminrelation}). Hence with probability $1-2 \cdot \frac{1}{2} \ce^{-\frac{1}{5}(\log n)^{\eps}}= 1-\ce^{-\frac{1}{5}(\log n)^{\eps}}$,
the discrepancy in round $\ell \cdot \kappa$ is $\Oh( (\log n)^{\eps} )$, completing the proof of the first statement.

Let us now prove the second statement. First, observe that if
$x^{(t)} \in \mathcal{E}_{\ell}$ for some round $t$, then also $x^{(t+1)} \in \mathcal{E}_{\ell}$.
We now choose $\epsilon := \frac{16}{\log \log n}$, $\ell := \lceil \log \log n \rceil$ and bound the number of rounds required to reach a load vector  in $\mathcal{E}_{\ell}$.
We divide this time into phases, each of which being of length $\kappa :=\tau_{\cont}(K,n^{-2}) $. As the success probability $p$ in \eq{probsmall} is only a positive constant for our choice of $\eps$, we may have to repeat some of the phases. However, the number of repetitions $R$ before we reach a load vector in $\mathcal{E}_{\ell}$ is stochastically smaller than the sum of $\ell$ independent geometric random variables, each of which having success probability $p$. Hence, by \lemref{chernoffgeo}, with probability at least $1 - \frac{1}{2 \log n}$, it holds that $R=\Oh(\ell)$, i.e., the
 load vector is in $\mathcal{E}_{\ell}$ after at most $\Oh( \ell)$ repetitions of $\kappa$ many rounds. If this indeed happens, then
\begin{align*}
 \sum_{u \in V} \max \left\{ x_{u}^{(R \cdot \kappa)} - 8 \ell \cdot \lceil (\log n)^{\epsilon} \rceil - \ell, 0   \right\} \leq 4 n \cdot \ce^{-\frac{1}{4} (\log n)^{\ell \epsilon}}.
\end{align*}
Plugging in the values of $\epsilon$ and $\ell$ yields
\begin{align*}
  \lefteqn{\hspace{-1cm} \sum_{u \in V} \max \left\{ x_{u}^{(R \cdot \kappa)} - 8 \left\lceil \log \log n \right\rceil \cdot \left\lceil (\log n)^{\frac{16}{\log \log n}} \right\rceil - \lceil \log \log n \rceil, 0   \right\} } \\
  &\leq 4 n \cdot \exp\left(-\frac{1}{4} \cdot (\log n)^{\lceil \log \log n\rceil \cdot \frac{16}{\log \log n}} \right) < 1.
\end{align*}
From the last inequality, it follows directly that every nodes $u \in V$ satisfies
\[
   x_{u}^{(R \cdot \kappa)} \leq 8 \lceil \log \log n \rceil \cdot \lceil \ce^{16} \rceil + \lceil \log \log n \rceil.
\]
To get the corresponding lower bound for the minimum load, we use \lemref{maxminrelation}. Thus with probability at least $1-\frac{1}{\log n}$, the discrepancy in round $R \cdot \kappa$ is upper bounded by $\Oh( \log \log n)$. This finishes the proof of the second statement and the proof of the theorem.
\end{proof}

\section{Proof of the Main Theorem (Theorem 1.1)}\label{sec:main}

The first part of this section gives a sketch of the proof of \thmref{main}. For the ease of the analysis, we remove the same number of tokens from every node such that the resulting load vector $x$ satisfies $\overline{x}\in[0,1)$ (cf.~\obsref{scaling}).
As illustrated in \figref{proof}, our proof consists of the following three main steps:
\begin{enumerate}\itemsep 0pt
\item \textbf{Reducing the Discrepancy to $(\log n)^{\epsilon_d}$.} We first use \thmref{logg} from \secref{randomwalk} to show that in round $t_1:=\Oh\left(\tau_{\cont}(K,n^{-2})\right)=\Oh\big(\frac{\log (Kn)}{1-\lambda}\big)$, the discrepancy is at most $(\log n)^{\epsilon_d}$, where $\epsilon_d > 0$ is a sufficiently small constant.
\item \textbf{Sparsification of the Load Vector.}
Since our goal is to achieve a constant discrepancy, we fix a constant $C>0$ and only consider nodes with more than $C$ tokens. We prove in \thmref{firstsparsification} that the number of tokens above the threshold $C$ on these nodes is at most $n \cdot \ce^{-(\log n)^{1-\eps}}$ in round $t_2:=t_1+\Oh\big(\frac{\log n}{1-\lambda}\big)$. The proof of this step is based on a polynomial potential function and the small discrepancy of the load vector in round $t_1$.

\item \textbf{Reducing the Discrepancy to a Constant.} Now we only need to analyze the $n \cdot \ce^{-(\log n)^{1-\eps}}$ tokens above the threshold $C$. This problem can be reduced to the analysis of a non-negative load vector with at most $n \cdot \ce^{-(\log n)^{1-\eps}}$ tokens (Observation~\ref{obs:scaling}).
 We prove in \thmref{final} that in round $t_3:=t_2+\Oh\big(\frac{\log n}{1-\lambda}\big)$, there is no token above the threshold $C+1$. The proof employs the token-based analysis via random walks from~\secref{randomwalk}. The lower bound on the minimum load follows by symmetry (\lemref{maxminrelation}). These two bounds on the maximum and minimum load together imply that the discrepancy in round $t_3$ is at most $2C+2$.
\end{enumerate}

\begin{figure}[htb]
\centering
\begin{tikzpicture}[xscale=0.55,yscale=0.7,thick, >=stealth, knoten/.style={circle, scale=0.5, draw=black}, rknoten/.style={circle, scale=0.5, draw=red}]
\draw[fill=black!10!white, thin] (0,9) -- (10,5) -- (17,5) -- (24,1.25) -- (0,1.25) -- (0,9);

  \draw[->] (-1,0) -- (25,0);
  \draw[->] (0,-1.5) -- (0,10);
  \draw (-0.4,9) -- (0.4,9);
  \draw (-0.4,5) -- (0.4,5);
  \draw (-0.4,1.25) -- (0.4,1.25);
  \node at (-0.4,9) [left] {$K$};
  \node at (-0.4,5) [left] {$(\log n)^{\epsilon_d}$};
  \node at (-0.4,1.25) [left] {$C+1$};
  \node at (0,10) [left] {$x_{\max}^{(t)}$};
  \draw (10,-0.4) -- (10,+0.4);
  \draw (17,-0.4) -- (17,+0.4);
  \draw (24,-0.4) -- (24,+0.4);
  \node at (25,0) [right] {$t$};
  \draw [gray, dashed, very thin] (10,-1.5) -- (10,10);
  \draw [gray, dashed, very thin] (17,-1.5) -- (17,10);
  \draw [gray, dashed, very thin] (24,-1.5) -- (24,10);
  \draw [gray, dashed, very thin] (0,1.25) -- (25,1.25);
  \draw [gray, dashed, very thin] (0,5) -- (25,5);
  \draw [gray, dashed, very thin] (0,9) -- (25,9);

  \node at (10,-0.5) [right] {$t_1$};
  \node at (17,-0.5) [right] {$t_2$};
  \node at (24,-0.5) [right] {$t_3$};

  \draw[<-] (0.25,-1.2) -- (1.25,-1.2);
  \node at (5,-1.2) [] {$\mathcal{O}\big( \frac{\log (Kn)}{1-\lambda} \big)$};
  \draw[->] (8.75,-1.2) -- (9.75,-1.2);

  \draw[<-] (10.25,-1.2) -- (11.25,-1.2);
  \node at (13.5,-1.2) [] {$\mathcal{O}\big( \frac{\log n}{1-\lambda} \big)$};
  \draw[->] (15.75,-1.2) -- (16.75,-1.2);

  \draw[<-] (17.25,-1.2) -- (18.25,-1.2);
  \node at (20.5,-1.2) [] {$\mathcal{O}\big( \frac{\log n}{1-\lambda} \big)$};
  \draw[->] (22.75,-1.2) -- (23.75,-1.2);

  \node at (4,6.4) [] {Theorem~\ref{thm:logg}};
  \node at (4,5.6) [] {\small{$x^{(0)}_{\max} \leq K$}};
  \node at (4,4.92) [] {$\Downarrow$};
  \node at (4,4.1) [] {\small{$x^{(t_1)}_{\max} \leq (\log n)^{\epsilon_d}$}};

  \node at (13.5,4.5) [] {Theorem~\ref{thm:firstsparsification}};
  \node at (13.5,3.7) [] {\small{$x^{(t_1)}_{\max} \leq (\log n)^{\epsilon_d}$}};
  \node at (13.5,3) [] {$\Downarrow$};
  \node at (13.5,2.2) [] {\footnotesize{$\sum\limits_{u \in V} \max\{ x_u^{(t_2)} - C, 0 \}$}};
  \node at (13.5,1.6) [] {\footnotesize{$\leq n \ce^{-(\log n)^{1-\epsilon}}$}};

  \node at (20.5,3.45) [] {Theorem~\ref{thm:final}};

  \node at (20.35,2.55) [] {\footnotesize{$\sum\limits_{u \in V} \max\{ x_u^{(t_2)} - C, 0 \}$}};
  \node at (20.35,1.95) [] {\footnotesize{$\leq n \ce^{-(\log n)^{1-\epsilon}}$}};
  \node at (20.35,1.2) [] {$\Downarrow$};
  \node at (20.35,0.55) [] {\small{$x^{(t_3)}_{\max} \leq C+1$}};
  \end{tikzpicture}

\centering
\caption{The above diagram illustrates how \thmref{logg}, \thmref{firstsparsification} and \thmref{final} are combined to prove \thmref{main}. We assume w.l.o.g.\ that $\overline{x} \in [0,1)$ and consider the decrease of the maximum load.
}\label{fig:proof}
\end{figure}

\begin{rem}
All results and arguments in this section will hold for the balancing circuit model (with constant $d$) and the random matching model for any regular graph $G$  as described in \secref{matching}, unless mentioned otherwise.
In the analysis, one round in the random matching model corresponds to $d$ consecutive rounds in the balancing circuit model, which ensures  smooth convergence as we periodically apply the same sequence of $d$ matchings. In fact, many of the complications in the proof come from the random matching model, as some nodes may not be part of any matching for up to $O(\log n)$ rounds.
\end{rem}


\subsection{Proof of \thmref{main}}

In this subsection, we state \thmref{firstsparsification} and \thmref{final} with their proofs deferred to \secref{firstsparsification} and \secref{final}, respectively. By assuming the correctness of these two theorems, we prove our main result, \thmref{main}, at the end of this subsection.

\begin{theorem}\label{thm:firstsparsification}
Let $G$ be any  graph,  and let $0 < \epsilon < 1$ be any constant.
Then, there are constants $\epsilon_d = \epsilon_d(\epsilon) > 0$, $C=C(\epsilon) > 0$ and $\nu=\nu(\epsilon) \in (0,1)$
such that the following holds. For any load vector $x^{(0)}$ with discrepancy at most $(\log n)^{\eps_d}$ and $\overline{x} \in [0,1)$,  it holds after  $\tau:=\Oh\big( \frac{\log n}{1-\lambda}\big)$ rounds with probability at least $1-\ce^{-(\log n)^{\nu}}$ that
\begin{align*}
  \sum_{u \in V} \max \left\{ x_u^{(\tau)} - C, 0 \right\} \leq n \cdot \ce^{-(\log n)^{1-\eps}}.
\end{align*}
\end{theorem}

The complete proof of \thmref{firstsparsification} is given in \secref{firstsparsification}.

\begin{theorem}\label{thm:final}
Let $G$ be any regular graph and $0 < \epsilon \leq \frac{1}{192}$ be any constant. Assume that $x^{(0)}$ is a non-negative load vector with $\| x^{(0)} \|_{1} \leq n \cdot \ce^{-(\log n)^{1-\epsilon}  }$. Then with probability  at least $1 - 3 \ce^{-(\log n)^{1-3 \eps}}$, it holds after $\kappa := \Oh\big(\frac{\log n}{1  - \lambda}\big)$ rounds that $ x_{\max}^{(\kappa)} \leq 1$.
\end{theorem}


We defer the proof of \thmref{final} to \secref{final} and first prove \thmref{main}, assuming the correctness of \thmref{firstsparsification} and \thmref{final}.

\begin{proof}
Set the value of $\epsilon$ in \thmref{final} to $1/192$, which in turn gives us a constant $\epsilon_d = \epsilon_d(\epsilon) > 0$ for the precondition in \thmref{firstsparsification}. By~\thmref{logg}, the discrepancy is at most $\mathcal{O}\left((\log n)^{\epsilon_d}\right)$ with probability at least $1-\ce^{-\frac{1}{5}(\log n)^{\epsilon_d}}$ in round $t_1:=\Oh\big(\frac{\log (Kn)}{1-\lambda}\big)$.
Next, we apply \thmref{firstsparsification} to prove that with probability at least $1-\ce^{-(\log n)^{\nu}}$, the load vector $x^{(t_2)}$ in round $t_2:=t_1+\Oh\big(\frac{\log n}{1-\lambda}\big)$ satisfies
\[
  \sum_{u \in V} \max \big\{ x_u^{(t_2)} - C, 0 \big\} \leq n \cdot \ce^{-(\log n)^{1-\eps}}.
\]
Consider now an auxiliary load vector $\widetilde{x}^{(s)}$, $s \geq t_2$, which is initialized in round $t_2$ by
$
 \widetilde{x}_u^{(t_2)} := \max \bigl\{ x_u^{(t_2)} - C, 0 \bigr\}
$ for any $u \in V$. For any $s > t_2$, $\widetilde{x}^{(s)}$ is obtained by the execution of the discrete load balancing protocol, starting with $\widetilde{x}^{(t_2)}$ in round $t_2$, which uses  the same matchings and orientations as for the load vector $x^{(s)}$ in every round $s > t_2$.
By Observation~\ref{obs:scaling}, $x_u^{(s)} \leq \widetilde{x}_u^{(s)} + C$ for every $s \geq t_2$ and thus  suffices to bound the maximum load of the non-negative load vector $\widetilde{x}^{(s)}$.
As $\| \widetilde{x}^{(t_2)} \|_{1} \leq n \cdot \ce^{-(\log n)^{1-\eps}}$,
we may apply \thmref{final} to conclude that in round   $t_3 := t_2 + \Oh\bigl(\frac{\log n}{1-\lambda}\bigr)$, $ \widetilde{x}_{\max}^{(t_3)} \leq 1$ holds with probability at least $1-3\mathrm{e}^{-(\log n)^{1-3\varepsilon}}$. Hence by the union bound and the relation between $\widetilde{x}^{(t_3)}$ and $x^{(t_3)}$, the maximum load of $x^{(t_3)}$ is at most $C+1$ with probability at least $1-\frac{1}{2} \ce^{-(\log n)^{c_1}}$, where $c_1 > 0$ is some constant. The corresponding lower bound on the minimum load is derived by symmetry (\lemref{maxminrelation}). Hence, with probability at least $1-\ce^{-(\log n)^{c_1}}$, the discrepancy is at most $c_2= 2C+2$. This finishes the proof of the first result in \thmref{main}.

The second result on the expected discrepancy can be derived as follows. First, consider the case where the initial discrepancy $K$ is at most $n$. Then, with probability $1-n^{-1}$ after $\tau:=\tau_{\cont}(n,1)=\Oh \big( \frac{\log n}{1-\lambda} \big)$ rounds, the discrepancy is at most $1$ in the continuous case. Using the first statement of \thmref{deviationmatching} with $\kappa=5$, with probability at least $1-2n^{-4}$, we have $\max_{w \in V}  \left| x_w^{(\tau)} - \xi_w^{(\tau)} \right| = \Oh( \sqrt{\log n})$. Hence, with probability at least $1-2n^{-4}-n^{-1}$ the discrepancy is at most $\Oh(\sqrt{\log n})$. If the discrepancy is at most $\Oh(\sqrt{\log n})$ after $\tau$ rounds, then by the first result,  the discrepancy is at most $c_2$ after $\ell=\Oh \big( \frac{\log n}{1-\lambda} \big)$ additional rounds with probability at least $1-\mathrm{e}^{-(\log n)^{c_1}}$. Combining this, we conclude that the expected discrepancy after $\tau+\ell$ rounds is at most
\[
  \Ex{ \disc\big(x^{(\tau+\ell)}\big)} \leq K \cdot \left( 2n^{-4} + n^{-1} \right) + \Oh\big(\sqrt{\log n}\big) \cdot  \ce^{-(\log n)^{c_1}} + c_2 \leq c_3,
\]
where $c_3 > 0$ is a constant.

Next, consider the case where the initial discrepancy $K$ is larger than $n$. Then, we first consider $4 \cdot \log_{n}(K) = 4 \cdot \frac{\log K}{\log n}$ consecutive intervals of length $ \tau_{\cont}(n,1) = \Oh \big( \frac{\log n}{1-\lambda} \big)$ each. Let $\kappa:=4 \cdot \log_{n}(K) \cdot \tau_{\cont}(n,1)$ be the final round at the end of these intervals. If at the beginning of an interval the discrepancy is larger than $n$, then by using the same arguments as above, it follows that with probability at least $1-2n^{-1}$ the discrepancy is reduced by a factor of $n/\Oh(\sqrt{\log n}) \geq \sqrt{n}$. If $\disc(x^{(\kappa)}) \leq n$, then from the first case we know that after additional $\ell$ rounds, the expected discrepancy is at most $c_2$.
Conditioning on the value of $\disc(x^{(\kappa)})$, we conclude that
\begin{align*}
 \Ex{ \disc\big(x^{(\kappa+\ell)}\big) }
 &\leq \Pro{ \disc\big(x^{(\kappa)}\big) \leq n } \cdot \Ex{ \disc\big(x^{(\kappa+\ell)}\big)               \, \mid \,  \disc\big(x^{(\kappa)}\big) \leq n  } \\ & ~~~~ +
 \Pro{ \disc\big(x^{(\kappa)}\big) > n }
 \cdot \Ex{  \disc\big(x^{(\kappa+\ell)}\big)               \, \mid \, \disc\big(x^{(\kappa)}\big) > n  } \\
 &\leq 1 \cdot c_3 +
 \Pro{ \disc\big(x^{(\kappa)}\big) > n } \cdot K.
\end{align*}
If $ \disc\big(x^{(\kappa)}\big) > n$ occurs, then in
less than $\log_{\sqrt{n}}(K) = 2 \cdot \log_{n}(K)$ intervals, the discrepancy is reduced by a factor of at most $\sqrt{n}$. Hence,
\begin{align*}
  \Pro{  \disc\big(x^{(\kappa)}\big) > n  } &\leq
  \binom{ 4 \cdot \log_{n}(K)}{ 2 \cdot \log_{n}(K)} \cdot
  \left( \frac{2}{n}\right)^{2 \cdot \log_{n}(K)}
  \leq \left( \frac{4\cdot\mathrm{e}}{n} \right)^{2 \cdot \log_{n}(K)}  
  \leq \frac{1}{K},
\end{align*}
and therefore, $\Ex{ \disc\big(x^{(\kappa+\ell)}\big) }=\Oh(1)$.
\end{proof}

\subsection{Proof of \thmref{firstsparsification}}\label{sec:firstsparsification}

Throughout  the proof of \thmref{firstsparsification}, we use the following potential function:
 \begin{align}\label{eq:def_potential}
 \Phi^{(t)} := \sum_{u \in V \colon x_u^{(t)} \geq 11} \big( x_u^{(t)} \big)^8 .
  \end{align}
  Occasionally, we will also apply this potential function to a different sequence of load vectors $\widetilde{x}^{(t)}, t \geq 0,$ and denote this by $\Phi^{(t)}(\widetilde{x})$. Our next observation is that $\Phi^{(t)}$ is non-increasing in $t$. Indeed, since our protocol only transfers tokens from nodes with larger load to ones with smaller load, it suffices to show that $x \mapsto x^ 8 \cdot \mathbf{1}_{x \geq 11}$ is convex, which follows from the convexity of $x \mapsto x^8$ and $11^8 - 0 \leq 12^8 - 11^8$.

The key step in proving \thmref{firstsparsification} is
  to analyze the drop of the potential function $\Phi$, which is formalized in the following lemma:

\begin{lemma}\label{lem:firstsparsificationlemma}
Let $\tau:=\Oh\big( \frac{\log n}{1-\lambda}\big)$. Then, the following two statements hold:
\begin{itemize}\itemsep 0pt
\item
For any load vector $x^{(0)}$ with discrepancy at most $ (\log n)^{13/900}$, it holds with probability at least $1-\ce^{-(\log n)^{\nu}}$, where $\nu \in (0,1)$ is a constant, that
\[
 \Phi^{(\tau)} \leq n \cdot \ce^{-(\log n)^{\frac{1}{100}}}.
 \]
\item
For any non-negative load vector $x^{(0)}$ with discrepancy at most $(\log n)^{\epsilon_d}$, where $0 < \epsilon_d \leq 13/900$, and $\| x^{(0)} \|_{1} \leq n \cdot \ce^{-(\log n)^{\sigma}}$ for some constant $\sigma \in (0,1)$, where $1-27 \epsilon_d - \sigma > 0$, it holds with probability at least $1-\ce^{-(\log n)^{\nu}}$, where $\nu \in (0,1)$ is a constant, that
\[ \Phi^{(\tau)} \leq n \cdot \exp\left(-(\log n)^{1 - 27 \eps_d - \frac{105}{106} (1 - 27 \eps_d - \sigma)   }\right).
\]
\end{itemize}
\end{lemma}
 Since $\sum_{u \in V} \max\big\{ x_u^{(\tau)} - 10, 0 \big\} \leq \Phi^{(\tau)}$, the second statement of \lemref{firstsparsificationlemma} states that the number of tokens above the threshold $10$ at the end of round $\tau$ is much smaller than $n \cdot \ce^{-(\log n)^{\sigma}}$. This argument can be iterated a constant number of times so that the number of tokens above the threshold $10 \cdot k$ at the end of round $k \cdot \tau$ is at most $n \cdot \ce^{-(\log n)^{1-\eps}}$ for a sufficiently large constant $k$, yielding~\thmref{firstsparsification}.

 \subsubsection{Proof of \thmref{firstsparsification}  using \lemref{firstsparsificationlemma}} We now proceed with the formal proof of~\thmref{firstsparsification} assuming the correctness of \lemref{firstsparsificationlemma}, whose proof is given in
 \secref{later}.

\begin{proof}
We choose $\epsilon_d \leq \min\{ \epsilon/54, 13/900 \}$. Assume that the discrepancy of the initial load vector $x^{(0)}$ is at most $(\log n)^{\epsilon_d}$. Then, the first statement of \lemref{firstsparsificationlemma} implies that with probability $1-\ce^{-(\log n)^{\nu}}$,
 \begin{align}
  \sum_{u \in V} \max \big\{ x_u^{(\tau)} - 10, 0 \big\} \leq \Phi^{(\tau)} \leq   n\cdot \ce^{-(\log n)^{\frac{1}{100}} }, \label{eq:refff}
\end{align}
where $\tau:= \Oh \big( \frac{\log n}{1-\lambda} \big)$ is as defined in \lemref{firstsparsificationlemma}.
Let us now define an auxiliary load vector $\widetilde{x}^{(s)}$ for any $s \geq \tau$. This load vector is initialized by $\widetilde{x}^{(\tau)}:= \max \big\{ x_u^{(\tau)} - 10, 0 \big\}$ and uses for any round $s > \tau$ the same matchings and orientations as the load vector $x^{(s)}$.
By Observation~\ref{obs:scaling}, it holds for any $s \geq \tau$ and node $u \in V$ that
\[
   x_u^{(s)} \leq \widetilde{x}_u^{(s)} + 10,
\]
which allows us to work with the non-negative load vector $\widetilde{x}$ in the following. By definition of $\widetilde{x}$ and \eq{refff},
$
  \left\| \widetilde{x}^{(\tau)}  \right\|_{1} \leq  n\cdot \ce^{-(\log n)^{\frac{1}{100}} }.
$
 Applying the second statement of \lemref{firstsparsificationlemma} to $\widetilde{x}^{(\tau)}$, it follows with probability $1-\ce^{-(\log n)^{\nu}}$ that
\begin{align*}
  \sum_{u \in V} \max \left\{ \widetilde{x}_u^{(2 \tau)} - 10, 0 \right\} \leq \Phi^{(2 \tau)}(\tilde{x})& \leq n \cdot \exp\left(-(\log n)^{1-27 \epsilon_d - \frac{105}{106}(1 - 27 \epsilon_d - a(1))} \right),
\end{align*}
where $a(1):=\frac{1}{100}$. Consequently,
\begin{align*}
  \sum_{u \in V} \max \left\{ x_u^{(2 \tau)} - 2 \cdot 10, 0 \right\} &\leq  \sum_{u \in V} \max \left\{ \widetilde{x}_u^{(2 \tau)} - 10, 0 \right\} \\
  &\leq n\cdot \exp\left(-(\log n)^{1- 27 \epsilon_d - \frac{105}{106}(1 - 27 \epsilon_d - a(1))} \right).
\end{align*}
Since the sequence $a(i), i \geq 2$, defined by the recursion
\begin{align*}
a(i) &:= (1- 27 \epsilon_d) - \frac{105}{106} \Big( (1 - 27 \epsilon_d) - a(i-1) \Big) = \frac{1}{106} \cdot (1 - 27 \epsilon_d) + \frac{105}{106} a(i-1)
\end{align*}
and $a(1) = \frac{1}{100}$, is non-decreasing in $i$
(as $1 - 27 \epsilon_d \geq \frac{1}{100}$),
and converges to $1- 27 \epsilon_d$, it follows for any integer $k \in \N$ by the union bound that with probability at least $1- k \cdot \ce^{-(\log n)^{\nu}}$, \[
  \sum_{u \in V} \max \left\{ x_u^{(k \cdot \tau)} - k \cdot 10, 0 \right\} \leq \Phi^{(k \cdot \tau)} \leq n\cdot \ce^{-(\log n)^{a(k)} }.
\]
Further, for any $\eps > 0$, there exists a constant $\iota=\iota(\eps) \in \mathbb{N}$ so that for any $k \geq \iota$, $a(k) \geq 1 - 27 \epsilon_d - \epsilon/2$. Since $\epsilon_d \leq \epsilon/54$, \thmref{firstsparsification} follows.
\end{proof}

\subsubsection{Proof of \lemref{firstsparsificationlemma}}\label{sec:later}

This part is devoted to the proof of \lemref{firstsparsificationlemma}. First, we define canonical paths.
\begin{defi}[{\cite{FS09}}]
    \label{def:canonicalpathitself}
 The sequence $\mathcal{P}_v^{(t_1)}=(\mathcal{P}_v^{(t_1)}(t_1)=v, \mathcal{P}_v^{(t_1)}(t_1+1), \ldots)$ is called the \emph{canonical path}
    of $v$ from round $t_1$ if for all rounds $t$ with $t > t_1$ the
    following holds:
    If $v_{t}:=\mathcal{P}_v^{(t_1)}(t)$  is unmatched in  $\M^{(t+1)}$, then $\mathcal{P}_v^{(t_1)}(t+1) := v_{t+1}$, where $v_{t+1} := v_t$.  Otherwise,
    let $u\in V$ be the node such that $\{v_t,u\}\in\M^{(t+1)}$.  Then,
    \begin{itemize}
    \setlength{\itemsep}{0pt}
    \setlength{\parskip}{0pt}
    \item
    if $x_{v_t}^{(t)}\geq x_{u}^{(t)}$ and $\Phi_{v_t,u}^{(t+1)}=1$,
    then $v_{t+1} = v_t$,
    \item
    if $x_{v_t}^{(t)}\geq x_{u}^{(t)}$ and $\Phi_{v_t,u}^{(t+1)}=-1$,
    then $v_{t+1} = u$,
    \item
    if $x_{v_t}^{(t)}< x_{u}^{(t)}$ and $\Phi_{v_t,u}^{(t+1)}=1$,
    then $v_{t+1} = u$,
    \item
    if $x_{v_t}^{(t)}< x_{u}^{(t)}$ and $\Phi_{v_t,u}^{(t+1)}=-1$,
    then $v_{t+1} = v_t$.
    \end{itemize}
 \end{defi}

An illustration of this definition is given in \figref{canonical}.

\begin{figure}[h]

\begin{tikzpicture}[xscale=0.65,yscale=0.65,thick, >=stealth, knoten/.style={circle, scale=0.4, draw=black, fill=black}, rknoten/.style={circle, scale=0.5, draw=red}]

\draw (2,4) -- (6,4);
\draw (2,2) -- (6,2);
\draw[color=black!50, line width=3pt, <-] (4,4) -- (4,2);
\draw[line width=3 pt, join=round, cap=round, color=red!40, ->] (2,4) -- (6,4);

\node[] () at (2,4) [label=left:$v_{t}$] {};
\node[] () at (2,2) [label=left:$u$] {};
\node[] () at (2.2,4) [label=above:$5$] {};
\node[] () at (2.2,2) [label=above:$2$] {};
\node[] () at (5.8,4) [label=above:$4$] {};
\node[] () at (5.8,2) [label=above:$3$] {};
\node[] () at (4,1)
[] {$\Phi_{v_t,u}^{(t+1)}=1$};
\end{tikzpicture}
\hspace{0.5cm}
\begin{tikzpicture}[xscale=0.65,yscale=0.65,thick, >=stealth, knoten/.style={circle, scale=0.4, draw=black, fill=black}, rknoten/.style={circle, scale=0.5, draw=red}]

\draw (2,4) -- (6,4);
\draw (2,2) -- (6,2);
\draw[color=black!50,  line width=3pt, ->] (4,4) -- (4,2);
\draw[line width=3 pt,  join=round, cap=round, color=red!40, ->] (2,4)  to  (3.5,4) to [bend left=30] (4,3.5) to (4,2.5) to [bend right=30] (4.5,2) to (6,2);

\node[] () at (2,4) [label=left:$v_{t}$] {};
\node[] () at (2,2) [label=left:$u$] {};
\node[] () at (2.2,4) [label=above:$5$] {};
\node[] () at (2.2,2) [label=above:$2$] {};
\node[] () at (5.8,4) [label=above:$3$] {};
\node[] () at (5.8,2) [label=above:$4$] {};
\node[] () at (4,1)
[] {$\Phi_{v_t,u}^{(t+1)}=-1$};
\end{tikzpicture}
\hspace{0.5cm}
\begin{tikzpicture}[xscale=0.65,yscale=0.65,thick, >=stealth, knoten/.style={circle, scale=0.4, draw=black, fill=black}, rknoten/.style={circle, scale=0.5, draw=red}]

\draw (2,4) -- (6,4);
\draw (2,2) -- (6,2);
\draw[color=black!50, line width=3pt, <-] (4,4) -- (4,2);
\draw[line width=3 pt, join=round, cap=round, color=red!40, ->] (2,4)  to  (3.5,4) to [bend left=30] (4,3.5) to (4,2.5) to [bend right=30] (4.5,2) to (6,2);

\node[] () at (2,4) [label=left:$v_{t}$] {};
\node[] () at (2,2) [label=left:$u$] {};
\node[] () at (2.2,4) [label=above:$2$] {};
\node[] () at (2.2,2) [label=above:$5$] {};
\node[] () at (5.8,4) [label=above:$4$] {};
\node[] () at (5.8,2) [label=above:$3$] {};
\node[] () at (4,1)
[] {$\Phi_{v_t,u}^{(t+1)}=1$};
\end{tikzpicture}
\hspace{0.5cm}
\begin{tikzpicture}[xscale=0.65,yscale=0.65,thick, >=stealth, knoten/.style={circle, scale=0.4, draw=black, fill=black}, rknoten/.style={circle, scale=0.5, draw=red}]

\draw (2,4) -- (6,4);
\draw (2,2) -- (6,2);
\draw[color=black!50, line width=3pt, ->] (4,4) -- (4,2);
\draw[line width=3 pt, join=round, cap=round, color=red!40, ->] (2,4) -- (6,4);

\node[] () at (2,4) [label=left:$v_{t}$] {};
\node[] () at (2,2) [label=left:$u$] {};
\node[] () at (2.2,4) [label=above:$2$] {};
\node[] () at (2.2,2) [label=above:$5$] {};
\node[] () at (5.8,4) [label=above:$3$] {};
\node[] () at (5.8,2) [label=above:$4$] {};
\node[] () at (4,1)
[] {$\Phi_{v_t,u}^{(t+1)}=-1$};
\end{tikzpicture}

\caption{Illustration of the four cases in the definition of canonical path.
}\label{fig:canonical}
\end{figure}
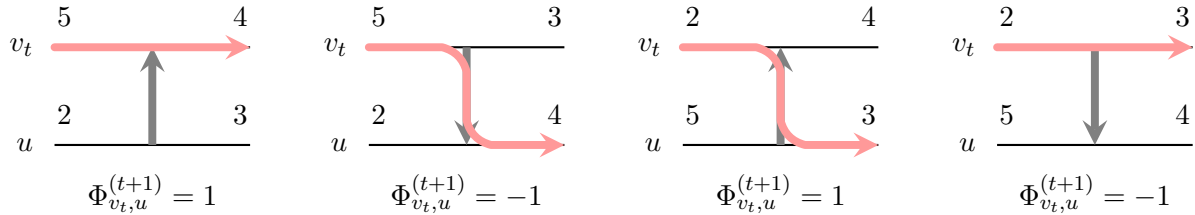

Note that there are always exactly $n$ canonical paths, which are all vertex-disjoint. We define canonical paths so that if
two of them are connected by a matching edge,  they continue in a
way so that the changes of the load (in absolute values) on each of the two paths is minimized. Hence, any increase or decrease of the load in a canonical path implies that the load vector becomes more balanced.

The next lemma relates the evolution of a canonical path to random walks and provides a lower bound on the probability that canonical paths meet, i.e., are connected by a matching edge in at least one round.

\begin{lemma}\label{lem:canonical}
Let $t_1 < t_2$ be two rounds. Fix a sequence of matchings $\langle \M^{(t_1+1)},\M^{(t_1+2)}, \ldots, \M^{(t_2)} \rangle$ and the load vector $x^{(t_1)}$.
Then, the following statements hold:
\begin{itemize}
\item
For any canonical path $\pat_v^{(t_1)}$ and node $w \in V$, it holds that
\begin{align*}
  \Pro{ \pat_v^{(t_1)}(t_2) = w } &= \M_{v,w}^{[t_1+1,t_2]}.
\end{align*}
\item For any pair of different nodes $u,v \in V$ and two canonical paths $\pat^{(t_1)}_u$ and $\pat^{(t_1)}_v$, it holds that
\begin{align*}
 \Pro{ \exists t\in[t_1, t_2-1]: \left\{ \pat^{(t_1)}_u(t), \pat^{(t_1)}_v(t) \right\} \in \M^{(t+1)} } \geq \sum_{w \in V} \M_{u,w}^{[t_1+1,t_2]} \cdot \M_{v,w}^{[t_1+1,t_2]}.
\end{align*}
\end{itemize}
\end{lemma}
\begin{proof}
Consider two independent random walks that start at nodes $u, v$ at the end of round $t_ 1$ and terminate at the end of round $t_2$. If a random walk is located at a node $w \in V$ at the end of a round $t \in [t_1,t_2-1]$, it can make the following transition. If node $w \in V$ is unmatched in round $t+1$, then the random walk stays at $w$. If node $w \in V$ is matched with a node $k \in V$, then the random walk stays at $w$ with probability $1/2$ and, otherwise, switches to $k$. Hence, the probability distribution for a random walk starting from $u$ in round $t_2$ is given by $\M_{u,\cdot}^{[t_1+1,t_2]}$ (cf.~proof of \lemref{rightprobability}). By Definition~\ref{def:canonicalpathitself}, a canonical path has the same transition probabilities as the aforementioned random walk. Therefore, the first statement of the lemma follows.

We now prove the second statement. The first statement implies that the probability that both independent random walks are located at the same node at the end of round $t_2$ is given by
\[
\sum_{w \in V} \M_{u,w}^{[t_1+1,t_2]} \cdot \M_{v,w}^{[t_1+1,t_2]}.
\]
\begin{figure}
\begin{center}
\begin{tikzpicture}[xscale=0.65,yscale=0.65,thick, >=stealth, knoten/.style={circle, scale=0.4, draw=black, fill=black}, rknoten/.style={circle, scale=0.5, draw=red}]

\draw[->] (0,10) -- (0,-1.7);
\draw (-0.25,9.75) -- (0.25,9.75);
\draw (-0.25,2.75) -- (0.25,2.75);
\draw (-0.25,-1.25) -- (0.25,-1.25);
\node[] () at (-0.25,9.75) [label=left:$t_1$] {};
\node[] () at (-0.25,-1.25) [label=left:$t_2$] {};
\node[] () at (-0.25,2.75) [label=left:$t$] {};

\node[knoten] (1) at (4,9.75) {};
\node[] () at (4,9.75) [label=above:$u$] {};


\draw[color=black!50, line width=3pt, <-] (4,9.75) -- (5,9.75);
\draw[color=black!50, line width=3pt, ->] (4,7.75) -- (5,7.75);
\draw[color=black!50, line width=3pt, ->] (4,6.75) -- (5,6.75);
\draw[color=black!50, line width=3pt, ->] (5,5.75) -- (6,5.75);
\draw[color=black!50, line width=3pt, ->] (6,3.75) -- (7,3.75);
\draw[color=black!50, line width=3pt, <-] (7,2.75) -- (8,2.75);
\draw[color=black!50, line width=3pt, <-] (7,0.75) -- (8,0.75);
\draw[color=black!50, line width=3pt, <-] (6,-0.25) -- (7,-0.25);

(4,9.75) -- (4,8.75) -- (4,8) [bend right=30] to (4.25,7.75) -- (4.75,7.75) [bend left=30] to (5,7.5) -- (5,6.75) --
(5,6) [bend right=30] to (5.25,5.75) -- (5.75,5.75)
[bend left=30] to (6,5.5) -- (6,4.75) -- (6,4) [bend right=30] to (6.25,3.75) -- (6.75,3.75) [bend left=30] to (7,3.5) -- (7,2.75) -- (7,1.75) -- (7,0.75) -- (7,0) [bend left=30] to (6.75,-0.25) -- (6.25,-0.25) [bend right=30] to (6,-0.5) -- (6,-0.75) -- (6,-1.75)
;
\draw[line width=4 pt, join=round, cap=round, color=red!40, ->]
(4,9.75) -- (4,8.75) -- (4,8) [bend right=30] to (4.25,7.75) -- (4.75,7.75) [bend left=30] to (5,7.5) -- (5,6.75) --
(5,6) [bend right=30] to (5.25,5.75) -- (5.75,5.75)
[bend left=30] to (6,5.5) -- (6,4.75) -- (6,4) [bend right=30] to (6.25,3.75) -- (6.75,3.75) [bend left=30] to (7,3.5) -- (7,2.75) -- (7,1.75) -- (7,0.75) -- (7,0) [bend left=30] to (6.75,-0.25) -- (6.25,-0.25) [bend right=30] to (6,-0.5) -- (6,-0.75) -- (6,-1.75)
;


\draw[color=black!50, line width=3pt, <-] (9,8.75) -- (10,8.75);
\draw[color=black!50, line width=3pt, ->] (10,7.75) -- (11,7.75);
\draw[color=black!50, line width=3pt, ->] (9,6.75) -- (10,6.75);
\draw[color=black!50, line width=3pt, ->] (9,5.75) -- (10,5.75);
\draw[color=black!50, line width=3pt, ->] (8,4.75) -- (9,4.75);
\draw[color=black!50, line width=3pt, <-] (8,1.75) -- (9,1.75);
\draw[color=black!50, line width=3pt, ->] (9,-0.25) -- (10,-0.25);
\draw[color=black!50, line width=3pt, <-] (9,-1.25) -- (10,-1.25);

(10,9.75) -- (10,8.75) --
(10,7.75) --
(10,7) [bend left=30] to (9.75,6.75) -- (9.25,6.75) [bend right=30] to (9,6.5) -- (9,5.75) -- (9,5)
[bend left=30] to (8.75,4.75) -- (8.25,4.75) [bend right=30] to (8,4.5) -- (8,3.75) -- (8,2.75) -- (8,2) [bend right=30] to (8.25,1.75) -- (8.75,1.75) [bend left=30] to (9,1.5) -- (9,0.75) -- (9,-0.25) -- (9,-1) [bend right=30] to (9.25,-1.25) -- (9.75,-1.25) [bend left=30] to (10,-1.5) -- (10,-1.75);

\draw[line width=4 pt, join=round, cap=round, color=blue!40, ->]
(10,9.75) -- (10,8.75) --
(10,7.75) --
(10,7) [bend left=30] to (9.75,6.75) -- (9.25,6.75) [bend right=30] to (9,6.5) -- (9,5.75) -- (9,5)
[bend left=30] to (8.75,4.75) -- (8.25,4.75) [bend right=30] to (8,4.5) -- (8,3.75) -- (8,2.75) -- (8,2) [bend right=30] to (8.25,1.75) -- (8.75,1.75) [bend left=30] to (9,1.5) -- (9,0.75) -- (9,-0.25) -- (9,-1) [bend right=30] to (9.25,-1.25) -- (9.75,-1.25) [bend left=30] to (10,-1.5) -- (10,-1.75);

\node[knoten] (1) at (4,8.75) {};
\node[knoten] (1) at (4,7.75) {};
\node[knoten] (1) at (5,7.75) {};
\node[knoten] (1) at (5,6.75) {};
\node[knoten] (1) at (5,5.75) {};
\node[knoten] (1) at (6,5.75) {};
\node[knoten] (1) at (6,4.75) {};
\node[knoten] (1) at (6,3.75) {};
\node[knoten] (1) at (7,3.75) {};
\node[knoten] (1) at (7,2.75) {};
\node[knoten] (1) at (8,2.75) {};
\node[knoten] (1) at (7,1.75) {};
\node[knoten] (1) at (7,0.75) {};
\node[knoten] (1) at (7,-0.25) {};
\node[knoten] (1) at (6,-0.25) {};
\node[knoten] (1) at (6,-1.25) {};

\node[knoten] (1) at (10,9.75) {};
\node[] () at (10,9.75) [label=above:$v$] {};
\node[] () at (6,-1.5) [label=below:$w$] {};
\node[] () at (10,-1.5) [label=below:$z$] {};

\node[knoten] (1) at (10,8.75) {};
\node[knoten] (1) at (10,7.75) {};
\node[knoten] (1) at (10,6.75) {};
\node[knoten] (1) at (9,6.75) {};
\node[knoten] (1) at (9,5.75) {};
\node[knoten] (1) at (9,4.75) {};
\node[knoten] (1) at (8,4.75) {};
\node[knoten] (1) at (8,3.75) {};
\node[knoten] (1) at (8,1.75) {};
\node[knoten] (1) at (9,1.75) {};
\node[knoten] (1) at (9,0.75) {};
\node[knoten] (1) at (9,-0.25) {};
\node[knoten] (1) at (9,-1.25) {};
\node[knoten] (1) at (10,-1.25) {};

\end{tikzpicture}
\end{center}

\caption{An illustration of two canonical paths $\pat_{u}^{(t_1)}$ and $\pat_{v}^{(t_1)}$ reaching nodes $w$ and $z$ at the end of round $t_2$, respectively. The two paths evolve independently except for round $t$ in which they are connected by a matching edge. In this illustration, the load of the canonical path $\pat_{u}^{(t_1)}$ is always at least the load at the other endpoint of the incident matching edges, while the opposite holds for $\pat_{v}^{(t_1)}$. Hence $\pat_{u}^{(t_1)}$ always follows the direction of the orientation, whereas $\pat_{v}^{(t_1)}$ does the opposite.
\label{fig:twocanonical}}
\end{figure}

Consider now the two canonical paths $\pat_{u}^{(t_1)}$ and $\pat_{v}^{(t_1)}$. Further, consider any round $t > t_1$, and fix the load vector $x^{(t-1)}$. By definition, both canonical paths make the same transition as a random walk, and furthermore, the two transitions are independent unless $\{\pat_{u}^{(t_1)}(t), \pat_{v}^{(t_1)}(t) \} \in \M^{(t+1)}$ (see \figref{twocanonical} for an illustration). Hence,
we can couple the transitions of the two canonical paths with the transitions of the two random walks until the first time when the two canonical paths are connected by a matching edge. If the two random walks are located at the same node $w$ at the end of round $t_2$, then, as two canonical paths cannot be located on the same vertex,
 the coupling above implies that there is one round $t \in [t_1,t_2-1]$ with $\left\{\pat_{u}^{(t_1)}(t), \pat_{v}^{(t_1)}(t) \right\} \in \M^{(t+1)}$, and thus,
\begin{align*}
 \Pro{  \exists t\in [t_1, t_2-1]: \left\{ \pat^{(t_1)}_u(t), \pat^{(t_1)}_v(t) \right\} \in \M^{(t+1)} } \geq \sum_{w \in V} \M_{u,w}^{[t_1+1,t_2]} \cdot \M_{v,w}^{[t_1+1,t_2]}.
\end{align*}
\end{proof}

Now we sketch the key ideas of the proof of \lemref{firstsparsificationlemma}. We use the polynomial potential function $\Phi^{(t)}$ that only involves nodes with load at least $11$ (see \eq{def_potential}). Using the condition that the discrepancy of the load vector is at most $(\log n)^{\epsilon_d}$, it follows directly that the initial value of the potential $\Phi$ is upper bounded by an almost linear function in $n$ (see Observation~\ref{obs:potentialinitial} below). To prove that $\Phi$ decreases, we consider phases of length $\beta:=(\log n)^{\epsilon_t}$, where $\epsilon_t \in (0,1)$ is a small constant. In each such phase, we consider canonical paths starting from nodes with load at least $11$ together with canonical paths starting from nodes with load at most $9$.

To lower bound the probability that two canonical paths meet, we use the relation between canonical paths and random walks, i.e., as long as both canonical paths have not been  connected by a matching edge, they evolve like  independent
random walks. Then, we focus on the nodes with load at least $11$ from which a canonical path has a large probability to meet with a canonical path starting from a node with load at most $9$ (Definition~\ref{def:good} and \lemref{goodnodess}). Using this, we establish in \lemref{inner} that there are indeed sufficiently many encounters and the potential drops.

We continue with the formal proof.
The next observation provides a simple bound on the initial value of the potential function $\Phi^{(t)}$ defined by \eq{def_potential}, exploiting the small discrepancy of the load vector.

\begin{obs}\label{obs:potentialinitial}
Consider any load vector $x^{(t)}$ with $\overline{x} \in [0,1)$ and discrepancy at most $(\log n)^{\epsilon_d}$. Then,
$
  \Phi^{(t)} \leq n \cdot ( (\log n)^{\epsilon_d} + 1)^{8}.
$
\end{obs}%

Our goal is to consider phases of length $\beta=(\log n)^{\epsilon_t}$, where $0 < \epsilon_t <1$ is a constant,  and prove that after each such phase, the expected value of the potential drops. Before analyzing the potential change, we introduce three conditions for a pair of node $u$ and round $t$.



\begin{defi}\label{defi:conditions}
Let $\beta=(\log n)^{\eps_t}$, where $0 < \eps_t < 1$ is any constant.
For any node $u \in V$ and round $t \in \mathbb{N}$, define the following three conditions:
\begin{itemize}\itemsep -0pt
\item $\coni(u,t) :\quad x_u^{(t)}\geq 11 $.
\item  $\conii(u,t) :\quad \max \left\{ \left\| \M_{u,\cdot}^{[t+1,t+\beta]} \right\|_2^{2}  ,  \left\| \M_{\cdot,u}^{[t+1,t+\beta]} \right\|_2^{2} \right\} \leq (\log n)^{-\eps_t/9}$.
 \item $\coniii(u,t) :\quad \sum_{w \in V}\left( \M_{u,w}^{[t+1,t+\beta]} \sum_{v \in V} x_v^{(t)} \M_{v,w}^{[t+1,t+\beta]}\right) < 4 $.
\end{itemize}
Moreover let $\mathcal{S}_1^{(t)}$ be the set of nodes $u$ satisfying $\coni(u,t)$. The sets $\mathcal{S}_2^{(t)}$ and $\mathcal{S}_3^{(t)}$ are defined in the same way.
\end{defi}

Note that $\conii(u,t)$ ensures that the local neighborhood around the node $u$, with respect to the graph induced by the matchings within the time-interval $[t+1,t+\beta]$, expands sufficiently.
Regarding $\coniii(u,t)$, recall that the probability distribution of $\pat_u^{(t)}(t+\beta)$ is given by $\M_{u,\cdot}^{[t+1,t+\beta]}$. Moreover, for any $w \in V$ and fixed load vector $x^{(t)}$, $\sum_{v \in V} x_v^{(t)} \M_{v,w}^{[t+1,t+\beta]}$ is the expected load on node $w$ in round $t+\beta$. Hence, if $\coniii(u,t)$ holds, then at the (random) location $\pat_u^{(t)}(t+\beta)$, the expected load is less than $4$. Thus, if $\coni(u,t)$ and $\coniii(u,t)$ both hold for a node $u \in V$, we would expect that $\pat_u^{(t)}$ has  a good chance to contribute to a decrease of the potential. This intuition motivates the following definition of set $V_1^{(t)}$,
which can be considered as the set of ``good'' nodes.

\begin{defi}\label{def:good}
Define
\[
  V_1^{(t)} := {\cal S}_1^{(t)} \cap {\cal S}_3^{(t)}.
\]
\end{defi}



The  next lemma provides two lower bounds on the size of $V_1^{(t)}$, where the tail bound of the second lower bound is exponentially small in the ``sparseness'' $\sigma$.
\begin{lemma}\label{lem:goodnodess}
Let $\epsilon_t \in (0,1), \epsilon_d \in (0,1)$ be two arbitrary constants. Fix an arbitrary load vector $x^{(0)}$ with discrepancy at most $(\log n)^{\epsilon_d}$ and $\overline{x} \in [0,1)$. Then, the following statements hold:
\begin{itemize}\itemsep -0pt
\item For any round $t \geq \tau_{\cont}(n,n^{-3})$,
\begin{align*}
  \Pro{ \left| V_1^{(t)} \right| \geq  \left|{\cal S}_1^{(t)} \right| - n \cdot \ce^{-(\log n)^{\epsilon_t/11}}} &\geq 1 - \ce^{ - (\log n)^{\epsilon_t/11} } - n^{-1}.
\end{align*}
\item If  $x^{(0)}$ is non-negative and satisfies $\| x^{(0)} \|_{1} \leq n \cdot \ce^{-(\log n)^{\sigma}}$ for some constant $\sigma\in(0,1)$, then for any round $t \geq \tau_{\cont}(n,n^{-3})$,
\begin{align*}
  \Pro{ \left| V_1^{(t)} \right| \geq \frac{1}{2} \left|{\cal S}_1^{(t)} \right| - n \cdot \ce^{-(\log n)^{\epsilon_t/20 + \sigma}}  } &\geq 1 -\ce^{-(\log n)^{\epsilon_t/20 + \sigma}} - \ce^{-\frac{1}{2}(\log n)^{\epsilon_t/2}} - n^{-1}.
\end{align*}
\end{itemize}
\end{lemma}
\begin{proof}
We begin by proving the first statement, which will be established by an upper bound on ${\cal S}_3^{(t)}$. Fix any node $u \in V$ in round $t$, and consider $\left\| \M_{u,\cdot}^{[t+1,t+\beta]} \right\|_{2}^2$ (recall $\beta=(\log n)^{\epsilon_t}$). While this is a deterministic value for the balancing circuit model, it is a random variable in the random matching model. By \lemref{matchinglocalexpansion}, it holds for the random matching model that
\begin{align*}
  \Pro{\max \left\{ \left\| \M_{u,\cdot}^{[t+1,t+\beta ]} \right\|_{2}^2, \left\| \M_{\cdot,u}^{[t+1,t+\beta ]} \right\|_{2}^2 \right\} \geq (\log n)^{-\epsilon_t/7} } \leq \ce^{-(
 \log n)^{ \epsilon_t/2} }.
\end{align*}
For the balancing circuit model, we replace the $\beta$ rounds by $d \cdot\beta$ rounds, which corresponds to applying the round matrix $\M=\prod_{i=1}^{d} \M^{(i)}$ $\beta$ times. By \lemref{circuitlocalexpansion},
 we have for any pair of nodes $u,v \in V$ that
\begin{align*}
    \M_{u,v}^{(\beta)} \leq (\log n)^{-\epsilon_t/9},
\end{align*}
 and hence, for any $t$ being a multiple of $d$,
\begin{align*}
  \max\left\{ \left\| \M_{u,\cdot}^{[t+1,t+d \beta]} \right\|_2^2,  \left\| \M_{\cdot,u}^{[t+1,t+d \beta]} \right\|_2^2  \right\} \leq (\log n)^{-\epsilon_t/9}.
\end{align*}
Let us now return to the random matching model. By Markov's inequality, we get
\begin{align*}
 \Pro{ \left| \left\{ u \in \mathcal{S}^{(t)}_1 \colon  u \mbox{ does not satisfy } \conii(u,t)  \right\} \right| \geq \left| \mathcal{S}_1^{(t)} \right| \cdot \ce^{-\frac{1}{2} (\log n)^{\epsilon_t/2} }} \leq  \ce^{-\frac{1}{2} (\log n)^{ \epsilon_t/2} },
\end{align*}
that is,
\begin{align}\label{eq:boundS2}
 \Pro{ \left| \mathcal{S}_1^{(t)} \setminus \mathcal{S}^{(t)}_2\right| \leq \left| \mathcal{S}_1^{(t)} \right| \cdot \ce^{-\frac{1}{2} (\log n)^{\epsilon_t/2} } } \geq  1- \ce^{-\frac{1}{2} (\log n)^{ \epsilon_t/2} }.
\end{align}
Similarly,
\begin{align}\label{eq:boundS2add}
 \Pro{ \left| V \setminus \mathcal{S}^{(t)}_2\right| \leq n \cdot \ce^{-\frac{1}{2} (\log n)^{\epsilon_t/2} } } \geq  1- \ce^{-\frac{1}{2}  (\log n)^{ \epsilon_t/2} }.
\end{align}

Note that $\mathcal{S}^{(t)}_2$ depends only on the random choices
for the matchings within the time-interval $[t+1,t+\beta]$, which is independent
of the load vector $x^{(t)}$ and the matchings in the time-interval $[0,t]$. Since $t \geq \tau_{\cont}(n,n^{-3})$, it follows that
\begin{align}
  \Pro{ \mbox{ $[0,t]$ is $(n,n^{-3})$--smoothing } } \geq 1 - n^{-1} \label{eq:smoothing},
\end{align}
and this probability is $1$ for the balancing circuit model.
For the remainder of the proof, we tacitly assume that $[0,t]$ is $(n,n^{-3})$--smoothing and focus on the random matching model. (The same arguments apply to the balancing circuit model by replacing a single round by $d$ rounds corresponding to the $d$ matchings, which are applied periodically.) Hence, by \lemref{bounddifference}, it holds for any sequence of matchings $\langle \M^{(t+1)}, \M^{(t+2)},\ldots,\M^{(t+\beta)} \rangle$, any node $u \in V$, and $\delta > 1/(2n)$ that
\[
 \Pro{ \left| \sum_{v \in V} x_v^{(t)} \M_{v,u}^{[t+1,t+\beta]} - \overline{x} \right| \geq \delta } \leq 2 \exp \left(- \frac{ (\delta - 1/(2n) )^2}{ 4\cdot \left\| \M_{.,u}^{[t+1,t+\beta]} \right\|_{2}^{2}     } \right).
\]
Choosing $\delta=1$ and recalling $\overline{x} \leq 1$ yields
\[
 \Pro{  \sum_{v \in V} x_v^{(t)} \M_{v,u}^{[t+1,t+\beta]}  \geq 2 \,\, \bigg| \,\, u \in \mathcal{S}_2^{(t)} } \leq 2 \exp \left(- \frac{(1- 1/(2n))^2}{ 4\cdot (\log n)^{-\epsilon_t/9}   } \right) \leq \ce^{ - (\log n)^{\epsilon_t/10}}.
\]
Hence, by Markov's inequality,
\[
 \Pro{ \left| \left\{ u \in \mathcal{S}_2^{(t)} \colon \sum_{v \in V} x_v^{(t)} \M_{v,u}^{[t+1,t+\beta]}  \geq 2  \right\} \right| \geq \left| \mathcal{S}_2^{(t)} \right| \cdot \ce^{- \frac{1}{2} (\log n)^{\epsilon_t/10}}  } \leq
  \ce^{- \frac{1}{2}(\log n)^{\epsilon_t/10} }.
\]
Let $\mathcal{R}^{(t)} := \left\{
 u \in V \colon  \sum_{v \in V} x_v^{(t)} \M_{v,u}^{[t+1,t+\beta]}   \geq 2 \right\}
$. Then, we infer from the above inequality that
\begin{align} \label{eq:boundR}
 \Pro{ \left|\mathcal{R}^{(t)}  \cap \mathcal{S}_2^{(t)} \right| \leq \left| \mathcal{S}_2^{(t)} \right| \cdot \ce^{- \frac{1}{2}  (\log n)^{\epsilon_t/10}} } \geq 1- \ce^{- \frac{1}{2}  (\log n)^{\epsilon_t/10} }.
\end{align}
Combining \eq{boundS2add}, \eq{boundR}, and $\mathcal{R}^{(t)} \subseteq \left( \mathcal{R}^{(t)}  \cap \mathcal{S}_2^{(t)} \right) \bigcup \left(  V \setminus \mathcal{S}_2^{(t)} \right)$ gives
\begin{align}
 \Pro{ \left| \mathcal{R}^{(t)} \right| \leq n \cdot \ce^{- \frac{1}{2} (\log n)^{\epsilon_t/10}}
 + n \cdot \ce^{ -\frac{1}{2} (\log n)^{\epsilon_t/2}   } }
 \geq 1 - 2 \ce^{- \frac{1}{2} (\log n)^{\epsilon_t/10} }. \label{eq:finall}
\end{align}
Our next goal is to upper bound the size of $\mathcal{S}_3^{(t)}$ in terms of $ | \mathcal{R}^{(t)} |  $. By definition,  we have
\begin{align}
\sum_{u\in V \setminus \mathcal{S}_3^{(t)}}\sum_{w \in V} \left( \M_{u,w}^{[t+1,t+\beta]} \sum_{v \in V} x_v^{(t)} \M_{v,w}^{[t+1,t+\beta]} \right) \geq 4\cdot \left| V \setminus \mathcal{S}_3^{(t)}\right|.
\label{eq:bound1}
\end{align}
On the other hand, since $x_{\max}^{(t)} \leq x_{\max}^{(0)} \leq (\log n)^{\epsilon_d} + 1 \leq 2 (\log n)^{\epsilon_d}$, we have
\begin{align}
\lefteqn{\sum_{u\in V \setminus \mathcal{S}_3^{(t)}}\sum_{w \in V} \left( \M_{u,w}^{[t+1,t+\beta]} \sum_{v \in V} x_v^{(t)} \M_{v,w}^{[t+1,t+\beta]} \right) }\nonumber\\
&=  \sum_{u\in V \setminus \mathcal{S}_3^{(t)}} \sum_{w \in \mathcal{R}^{(t)}} \left( \M_{u,w}^{[t+1,t+\beta]} \sum_{v \in V} x_v^{(t)} \M_{v,w}^{[t+1,t+\beta]} \right) \nonumber \\ &\qquad +
\sum_{u\in V \setminus \mathcal{S}_3^{(t)}} \sum_{w \in V\setminus\mathcal{R}^{(t)}} \left(  \M_{u,w}^{[t+1,t+\beta]} \sum_{v \in V} x_v^{(t)} \M_{v,w}^{[t+1,t+\beta]} \right) \nonumber\\
&\leq 2(\log n)^{\epsilon_d}\cdot \sum_{u\in V \setminus \mathcal{S}_3^{(t)}} \sum_{w \in \mathcal{R}^{(t)}} \M_{u,w}^{[t+1,t+\beta]} \sum_{v \in V} \M_{v,w}^{[t+1,t+\beta]} + \sum_{u \in V \setminus \mathcal{S}_3^{(t)}} 2 \cdot \sum_{w \in V} \M_{u,w}^{[t+1,t+\beta]}   \nonumber\\
& \leq 2(\log n)^{\epsilon_d}\cdot\sum_{w \in \mathcal{R}^{(t)}} \sum_{u\in V} \M_{u,w}^{[t+1,t+\beta]}   + 2\left|V \setminus \mathcal{S}_3^{(t)}\right|\nonumber\\
&\leq 2(\log n)^{\epsilon_d}\cdot  \left| \mathcal{R}^{(t)}\right| +2\left|V \setminus \mathcal{S}_3^{(t)}\right|.\label{eq:bound2}
\end{align}
Combining \eq{bound1} and \eq{bound2} yields
\begin{align*}
  \left| V \setminus \mathcal{S}_3^{(t)} \right| &\leq (\log n)^{\eps_d} \cdot
  \left|\mathcal{R}^{(t)}\right| .
\end{align*}
Using this, \eq{finall} and
\[
V_1^{(t)} = \mathcal{S}_1^{(t)} \cap \mathcal{S}_3^{(t)} = \mathcal{S}_1^{(t)} \setminus \left(V \setminus \mathcal{S}_3^{(t)}\right)\]
 complete the proof of the first statement.

Let us now turn to the second statement of the lemma.
As in the proof of the first statement, we tacitly assume that the time-interval $[0,t]$ is $(n,n^{-3})$--smoothing, which holds with probability at least $1-n^{-1}$, and we focus on the random matching model. Fix any node  $u \in \mathcal{S}_2^{(t)}$.
Note that
\begin{align*}
\sum_{w \in V} \left( \M_{u,w}^{[t+1,t+\beta]}  \sum_{v \in V} x_v^{(t)} \M_{v,w}^{[t+1,t+\beta]} \right) &= \sum_{v \in V} \left( \sum_{w \in V} \M_{u,w}^{[t+1,t+\beta]} \cdot \M_{v,w}^{[t+1,t+\beta]} \right) \cdot x_v^{(t)} .
\end{align*}
Now, we apply \lemref{sparsiter} with $y_v = \sum_{w \in V} \M_{u,w}^{[t+1,t+\beta]} \cdot \M_{v,w}^{[t+1,t+\beta]}$. Since $u \in \mathcal{S}_2^{(t)}$, the vector $y$ satisfies
\[
  \| y \|_{\infty} \leq \left\| \mathbf{M}_{u,\cdot}^{[t+1,t+\beta]} \right\|_{\infty} \leq \left\| \mathbf{M}_{u,\cdot}^{[t+1,t+\beta]} \right\|_{2} \leq (\log n)^{-\epsilon_t/18},
\]
and choosing $\delta=\epsilon_t/19$ in \lemref{sparsiter} gives
\begin{align*}
\lefteqn{ \Pro{u\not\in \mathcal{S}_3^{(t)} } } \\
  &=
  \Pro{  \sum_{v \in V} \left( \sum_{w \in V} \M_{u,w}^{[t+1,t+\beta]} \cdot \M_{v,w}^{[t+1,t+\beta]} \right) \cdot x_v^{(t)}  \geq 4 }  \\
  &\leq  \Pro{  \sum_{v \in V} \left( \sum_{w \in V} \M_{u,w}^{[t+1,t+\beta]} \cdot \M_{v,w}^{[t+1,t+\beta]} \right) \cdot x_v^{(t)}  \geq \ce^{-\frac{1}{5} (\log n)^{\sigma}} + 8 (\log n)^{-\epsilon_t/18} \cdot (\log n)^{\epsilon_t/19} } \\
   & \leq  \exp\left(- \frac{1}{3} (\log n)^{\epsilon_t / 19 + \sigma }  \right) \leq \exp\left(- 2\cdot (\log n)^{\epsilon_t / 20 + \sigma }   \right).
\end{align*}
Therefore, by Markov's inequality,
\begin{align*}
 \Pro{ \left| \mathcal{S}_2^{(t)} \setminus \mathcal{S}_3^{(t)}  \right| \geq n \cdot 	 \ce^{- (\log n)^{\epsilon_t / 20 + \sigma }  }} &\leq \ce^{- (\log n)^{\epsilon_t / 20 + \sigma }  }\enspace.
\end{align*}
Moreover, it follows from \eq{boundS2} that
\begin{align*}
 \Pro{ \left| \mathcal{S}_1^{(t)} \cap \mathcal{S}_2^{(t)} \right| \geq \frac{1}{2} \cdot \left| \mathcal{S}_1^{(t)} \right| } \geq 1 - \ce^{-\frac{1}{2} (\log n)^{\epsilon_t/2}}.
\end{align*}
Using this and the fact that $V_1^{(t)}= \mathcal{S}^{(t)}_1 \cap \mathcal{S}^{(t)}_3 \supseteq \left(\mathcal{S}^{(t)}_1 \cap \mathcal{S}^{(t)}_2\right) \setminus \left(\mathcal{S}^{(t)}_2 \setminus \mathcal{S}^{(t)}_3\right)$ implies that, under the condition that $[0,t]$ is $(n,n^{-3})$-smoothing,
\begin{align*}
 \Pro{ \Big| V_1^{(t)} \Big| \geq \frac{1}{2} \left|{\cal S}_1^{(t)}\right|  - n \cdot \ce^{-(\log n)^{\epsilon_t/20 + \sigma}}  } &\geq 1 -\ce^{-(\log n)^{\epsilon_t/20 + \sigma}} - \ce^{-\frac{1}{2}(\log n)^{\epsilon_t/2}},
\end{align*}
which finishes the proof.
\end{proof}


\begin{lemma}\label{lem:inner}
Let $\epsilon_b, \epsilon_d$ and $\epsilon_t$ be three arbitrary constants in the interval $(0,1)$. Fix a load vector $x^{(t)}$ with $\overline{x} \in [0,1)$ and discrepancy at most $(\log n)^{\epsilon_d}$.
Assume that the load vector $x^{(t)}=x$ and the matchings in the time-interval $[t+1,t+\beta]$ satisfy
\begin{align*}
 \left|V_{1}^{(t)} \right| &\geq \frac{1}{2} \left| \mathcal{S}_1^{(t)} \right| - n\cdot\ce^{-(\log n)^{\epsilon_b} }.
\end{align*}
Then, it holds in round
$t+\beta, \beta:=(\log n)^{\eps_t}$ that
\begin{align*}
 \Ex{ \Phi^{(t+\beta)} \, \middle| \, x^{(t)} = x } &\leq \max \left\{ \left(1 -  \frac{1}{9 (\log n)^{\epsilon_t + 9 \epsilon_d}} \right) \cdot \Phi^{(t)}, 4n \cdot \ce^{-(\log n)^{\epsilon_b}  } \cdot \left( (\log n)^{ \epsilon_d} + 1 \right)^{8} \right \}.
\end{align*}
\end{lemma}

\begin{proof}
We remark that the precondition on $V_{1}^{(t)}$ depends only on the random choices for the matching edges in the interval $[t+1,t+\beta]$, but not on the orientation of these edges. Therefore, the orientations of the matchings in the time-interval $[t+1,t+\beta]$ are still chosen independently and uniformly at random. In the proof we focus on the random matching model, but exactly the same arguments apply to the balancing circuit model, where one round of the random matching model corresponds to $d$ matchings in the balancing circuit model.

We first observe that the claim of the lemma holds trivially if $\left|\mathcal{S}_1^{(t)}\right| \leq 4 n \cdot \mathrm{e}^{-(\log n)^{\epsilon_b}}$, since, in this case, we have that
\begin{align*}
 \Phi^{(t+\beta)} \leq \Phi^{(t)} &=  \sum_{\substack{u\in V\\ x^{(t)}_u\geq 11 }} \big( x_u^{(t)} \big)^8 \leq \left| \mathcal{S}_1^{(t)} \right| \cdot \left( (\log n)^{ \epsilon_d} + 1 \right)^{8} \leq 4 n \cdot \ce^{-(\log n)^{\epsilon_b}  } \cdot \left( (\log n)^{ \epsilon_d} + 1 \right)^{8}.
\end{align*}
Hence  we may assume that
\begin{align}
 \left| \mathcal{S}_1^{(t)} \right| &\geq 4 n \cdot \mathrm{e}^{-(\log n)^{\epsilon_b}}. \label{eq:assumption}
\end{align}
In the following, we fix an arbitrary node $u \in V_1^{(t)}$.
By definition,  node $u$ satisfies
\[
  \sum_{w \in V} \left( \M_{u,w}^{[t+1,t+\beta]} \sum_{v \in V} x_v^{(t)} \M_{v,w}^{[t+1,t+\beta]} \right) < 4.
\]
Let
\[
  \mathcal{G}^{(t+\beta)} := \left\{ w \in V \colon \sum_{v \in V} x_v^{(t)} \M_{v,w}^{[t+1,t+\beta]} \leq 8 \right\}.
\]
Since $\mathcal{G}^{(t+\beta)} \subseteq V$, we have that
\[
 \sum_{w \not\in \mathcal{G}^{(t+\beta)}} \left( \M_{u,w}^{[t+1,t+\beta]} \sum_{v \in V} x_v^{(t)} \M_{v,w}^{[t+1,t+\beta]} \right) < 4.
\]
Combining this with the fact that
\[
 \sum_{w \not\in \mathcal{G}^{(t+\beta)}} \left( \M_{u,w}^{[t+1,t+\beta]} \sum_{v \in V} x_v^{(t)} \M_{v,w}^{[t+1,t+\beta]} \right) \geq  8\cdot \sum_{w \not\in \mathcal{G}^{(t+\beta)}}\M_{u,w}^{[t+1,t+\beta]}
\]
gives $\sum_{ w\not\in \mathcal{G}^{(t+\beta)}}\M_{u,w}^{[t+1,t+\beta]} \leq 1/2$.
Since the probability distribution of $\mathcal{P}_{u}^{(t)}(t+\beta)$, the location of the canonical path of $u$ from round $t$ in round $t+\beta$, is given by $\M_{u,\cdot}^{[t+1,t+\beta]}$, we have that
\begin{align}
 \Pro{ \pat_{u}^{(t)}(t+\beta) \in \mathcal{G}^{(t+\beta)} } = \sum_{w \in \mathcal{G}^{(t+\beta)}} \M_{u,w}^{[t+1,t+\beta]} \geq \frac{1}{2}. \label{eq:starstar}
\end{align}
Fix now a node $w \in \mathcal{G}^{(t+\beta)}$.
Our next aim is to derive a lower bound on the probability that there is a canonical path starting from a node $v$ with $x_v^{(t)} \leq 9$ that reaches $w$ in round $t+\beta$.
To this end, let us define $
  V_2^{(t)} := \left\{ v \in V \colon x_{v}^{(t)} \leq 9 \right\}
$ and $\alpha_w := \sum_{v \in V_2^{(t)}} \M_{v,w}^{[t+1,t+\beta]}$ for node $w$.
  Then,
\begin{align*}
  8 &\geq \sum_{v \in V} x_v^{(t)} \cdot \M_{v,w}^{[t+1,t+\beta]} \\
   &=
  \sum_{v \in V_2^{(t)}} x_{v}^{(t)} \cdot \M_{v,w}^{[t+1,t+\beta]} + \sum_{v \in V \setminus V_2^{(t)}} x_{v}^{(t)} \cdot \M_{v,w}^{[t+1,t+\beta]} \\
  &\geq-(\log n)^{\epsilon_d} \cdot \alpha_w + 10 \cdot (1-\alpha_w),
\end{align*}
and rearranging yields for any $w \in \mathcal{G}^{(t+\beta)}$ that
\begin{align}\label{eq:boundalphaw}
\alpha_w &\geq \frac{2}{(\log n)^{\epsilon_d} + 10} \geq \frac{1}{(\log n)^{\epsilon_d}}.
\end{align}

Let us now lower bound the probability for the event that the canonical path starting from $u\in V_1^{(t)}$ meets with a canonical path starting from a node $v \in V_2^{(t)}$ within the time-interval $[t+1,t+\beta]$.
By the second statement of \lemref{canonical}, it holds that
%
\begin{align*}
  \Pro{ \exists s \in [t,t+\beta-1] \colon \left\{ \pat_{u}^{(t)}(s), \pat_{v}^{(t)}(s) \right\} \in \M^{(s+1)}} \geq  \sum_{w \in V} \M_{u,w}^{[t+1,t+\beta]} \cdot \M_{v,w}^{[t+1,t+\beta]}.
\end{align*}

To aggregate the contribution from all nodes in $V_1^{(t)}$,
we now define an auxiliary bipartite graph $H=\big(V_1^{(t)} \cup V_2^{(t)}, E(H)\big)$, which depends on the load balancing process within the time-interval $[t+1,t+\beta]$. We place an edge $\{u,v\} \in E(H)$ if the canonical path of $u \in V_1^{(t)}$ meets with the canonical path of $v \in V_2^{(t)}$ within the time-interval $[t+1,t+\beta]$. We first lower bound the expected number of edges in $H$:
\begin{align*}
\Ex{ |E(H)| }
 &= \sum_{u \in V_1^{(t)}} \sum_{v \in V_2^{(t)}}
 \Pro{  \exists s \in [t,t+\beta-1] \colon  \left\{ \pat_{u}^{(t)}(s), \pat_{v}^{(t)}(s) \right\} \in \M^{(s+1)}              } \\
 &\geq \sum_{u \in V_1^{(t_1)}} \sum_{v \in V_2^{(t)}} \sum_{w \in V}  \M_{u,w}^{[t+1,t+\beta]} \cdot \M_{v,w}^{[t+1,t+\beta]}        \\
  &\geq \sum_{u \in V_1^{(t)}} \sum_{w \in \mathcal{G}^{(t+\beta)}} \M_{u,w}^{[t+1,t+\beta]} \cdot \sum_{v \in V_2^{(t)}}  \M_{v,w}^{[t+1,t+\beta]}.
  \end{align*}
By \eq{boundalphaw}, it holds for $w \in \mathcal{G}^{(t+\beta)}$ that $\sum_{v \in V_2^{(t)}}  \M_{v,w}^{[t+1,t+\beta]}\geq \frac{1}{(\log n)^{\epsilon_d}}$. Hence,
  \begin{align*}
  \Ex{ |E(H)| }&\geq
 \frac{1}{(\log n)^{\epsilon_d}} \cdot \sum_{u \in V_1^{(t)}} \sum_{w \in \mathcal{G}^{(t+\beta)}}  \M_{u,w}^{[t+1,t+\beta]} \geq \frac{1}{(\log n)^{\epsilon_d}} \cdot \left|V_1^{(t)} \right| \cdot \frac{1}{2}~,
 \end{align*}
 where the second inequality follows from \eq{starstar}.
 By the precondition of this lemma,
 \begin{align*}
  \left|V_1^{(t)} \right| &\geq \frac{1}{2}  \left|\mathcal{S}_1^{(t)} \right| - n \cdot \ce^{-(\log n)^{\epsilon_b}}.
 \end{align*}
 Since $\left|\mathcal{S}_1^{(t)}\right| \geq 4 n\cdot \ce^{-(\log n)^{\epsilon_b}}$ by \eq{assumption}, we conclude that $\left| V_1^{(t)} \right| \geq \frac{1}{4} \left|\mathcal{S}_1^{(t)}\right|$. This allows us to lower bound $\Ex{|E(H)|}$ as follows:
 \begin{align}
\Ex{|E(H)|} &\geq \frac{1}{8  (\log n)^{\epsilon_d} } \cdot \left| \mathcal{S}_1^{(t)} \right| \notag \\
 &\geq \frac{1}{8  (\log n)^{\epsilon_d} } \cdot \Phi^{(t)} \cdot \frac{1}{\left( (\log n)^{ \epsilon_d} + 1 \right)^{8} } \notag \\
 &\geq \frac{1}{9 (\log n)^{9 \epsilon_d} } \cdot \Phi^{(t)}. \label{eq:dro}
\end{align}

Consider now the auxiliary graph $H$ again. By definition, $H$ contains all those edges $\{u,v\}$ with the property that two canonical paths of $u \in V_1^{(t)}$ and $v \in V_2^{(t)}$ meet within the time-interval $[t+1,t+\beta]$. We consider all the edges of this bipartite graph $H=\left(V_1^{(t_1)} \cup V_2^{(t_1)},E(H)\right)$
in a round-by-round fashion as the load balancing protocol proceeds. We remove an edge once the load at one of the two endpoints changes (in which case the potential drops by at least $1$), and we also remove an edge
when the two corresponding canonical paths meet (noticing that, also in this case, the potential drops by
at least $1$ unless the edge was removed earlier).

Since by the end of round $t+\beta$, all edges in $H$ are removed, we conclude that the potential drops by at least
\[
  \Phi^{(t)} - \Phi^{(t+\beta)} \geq 1 \cdot \frac{|E(H)|}{\Delta(H)} \geq  \frac{|E(H)|}{(\log n)^{\epsilon_t}}.
\]
Taking expectations on both sides and using our lower bound from \eq{dro} on $\Ex{|E(H)|}$ finally gives
\begin{align*}
 \Ex{\Phi^{(t)} - \Phi^{(t+\beta)} \, \middle| \, x^{(t)} = x } &\geq
 \frac{\Ex{|E(H)|}}{(\log n)^{\epsilon_t}} \geq \frac{1}{9  (\log n)^{\epsilon_t+9 \epsilon_d} } \cdot \Phi^{(t)},
\end{align*}
which completes the proof.
\end{proof}

Finally, we are now ready to prove \lemref{firstsparsificationlemma} by combining \obsref{potentialinitial}, \lemref{goodnodess}, and \lemref{inner}.

\begin{proof}
In this proof, we focus on the random matching model, but exactly the same arguments apply to the balancing circuit model where one round of the random matching model corresponds to $d$ matchings in the balancing circuit model.

For the proof of the first statement, we choose the following constants:
$\epsilon_t := 3/25 , \epsilon_d := 13/900 $, and $\epsilon_b := \epsilon_t / 11 =  3/275 $. Fix a round $t_0 :=\tau_{\cont}(n,n^{-3})$. Define for any $t \geq t_0$ the event
\[
\mathcal{E}^{(t)} := \left\{  \left|   V_1^{(t)} \right|  \geq \frac{1}{2} \left| \mathcal{S}_1^{(t)} \right|  - n \cdot \ce^{-(\log n)^{\epsilon_b}} \right\},
\]
and $\mathcal{E} := \bigcup_{t=t_0}^{t_0+\log n-1} \mathcal{E}^{(t)}$. By the first statement of \lemref{goodnodess},
\begin{align*}
 \Pro{ \mathcal{E}^{(t)} } &\geq 1 - \ce^{-(\log n)^{\epsilon_t/11} } - n^{-1},
\end{align*}
and taking the union bound over the time-interval $[t_0,t_0+\log n-1]$,
\begin{align*}
 \Pro{ \mathcal{E} } &\geq 1 - \ce^{-(\log n)^{\epsilon_t/12} }.
\end{align*}
By the first statement of \lemref{inner} with $\epsilon_b = \epsilon_t/11$, for any round $t \in [t_0,t_0+\log n]$ such that the load vector $x^{(t)}=x$  and the matchings in the time-interval $[t,t+\beta]$ satisfy $\mathcal{E}^{(t)}$ and $\Phi^{(t)} \geq n \cdot \ce^{-(\log n)^{\epsilon_t/12}}$, it holds that
\[
  \left(1 -  \frac{1}{9 (\log n)^{\epsilon_t + 9 \epsilon_d}} \right) \cdot \Phi^{(t)}\geq  4n \cdot \ce^{-(\log n)^{\epsilon_b}  } \cdot \left( (\log n)^{ \epsilon_d} + 1 \right)^{8},
\]
and thus,
\begin{align}
 \Ex{ \Phi^{(t+\beta)} \, \middle| \, x^{(t)} = x } &\leq  \left(1 -  \frac{1}{9 (\log n)^{\epsilon_t + 9 \epsilon_d}} \right) \cdot \Phi^{(t)}. \label{eq:equalityf}
\end{align}
By Markov's inequality, we have that
\begin{align*}
 \lefteqn{\Pro{  \Phi^{(t+\beta)} \geq \left(1-\frac{1}{81 (\log n)^{2\epsilon_t + 18 \epsilon_d}}\right) \Phi^{(t)} \, \middle| \,  x^{(t)}=x }} \\
 &\leq
 \frac{  1-\frac{1}{9 (\log n)^{\epsilon_t + 9 \epsilon_d}}    }{1-\frac{1}{81 (\log n)^{2\epsilon_t + 18 \epsilon_d}}}
 \leq 1 - \frac{1}{18 \cdot (\log n)^{\epsilon_t + 9 \epsilon_d}},
\end{align*}
where the last inequality holds since $\frac{1-\frac{1}{x}}{1-\frac{1}{x^2}} \leq 1 - \frac{1}{2x}$ for any $x \geq 2$.
Equivalently,
\begin{align}
\Pro{  \Phi^{(t+\beta)} \leq \left(1-\frac{1}{81  (\log n)^{2\epsilon_t + 18 \epsilon_d}}\right) \Phi^{(t)} \, \middle| \,  x^{(t)}=x } \geq \frac{1}{18 \cdot (\log n)^{\epsilon_t + 9 \epsilon_d}} =: p. \label{eq:crc}
\end{align}

Now divide the time-interval $[t_0,t_0+\log n]$ into $(\log n)^{1-\epsilon_t}$ consecutive sections of length $\beta:=(\log n)^{\epsilon_t}$ each. In the following, we call a section $k \in \{1,\ldots,(\log n)^{1-\epsilon_t}\}$ good if (i) $\mathcal{E}^{(t_0+(k-1) \cdot \beta)}$ does not occur, or (ii) $
\Phi^{(t_0+k\cdot\beta)} \leq   \left(1-\frac{1}{81  (\log n)^{2\epsilon_t + 18 \epsilon_d}}\right) \Phi^{(t_0+(k-1)\cdot\beta)}
$, or (iii) $\Phi^{(t_0+(k-1)\cdot\beta)} \leq n \cdot \ce^{-(\log n)^{\epsilon_t/12}}$. Then by \eq{crc}, every section $k \in \{1,\ldots,(\log n)^{1-\epsilon_t}\}$ is good with probability at least $p$, independently of all previous sections.
Consider a sequence of $(\log n)^{1- \epsilon_t}$ independent Bernoulli random variables $X_1,X_2,\ldots,X_{(\log n)^{1- \epsilon_t}}$ with success probability $p$ each and let $X:=\sum_{k=1}^{(\log n)^{1- \epsilon_t}} X_i$. Then,
\[
 \Ex{X} = (\log n)^{1- \epsilon_t} \cdot p = (\log n)^{1-\epsilon_t} \cdot \frac{1}{18 \cdot (\log n)^{\epsilon_t + 9 \epsilon_d}} = \frac{\log n}{18 \cdot (\log n)^{2 \epsilon_t + 9 \epsilon_d} }.
\]
Hence by a Chernoff bound we have that
\begin{align*}
 \Pro{ X \geq \frac{1}{2} \cdot \Ex{X} } &\geq 1 - \mathrm{exp}\left(-\frac{\Ex{X}}{8}\right) \geq
 1 - \mathrm{exp}\left(-(\log n)^{1-3 \epsilon_t-9 \epsilon_d}\right),
\end{align*}
and by the union bound it holds that
\begin{align}
 \Pro{ \left( X \geq \frac{1}{2} \cdot \Ex{X}\right) \cap \mathcal{E} } \geq 1 - \mathrm{exp}\left({-(\log n)^{1-3 \epsilon_t - 9 \epsilon_d}}\right) - \ce^{-(\log n)^{\epsilon_t/12}}. \label{eq:uboundone}
\end{align}
If $( X \geq \frac{1}{2} \cdot \Ex{X}) \cap \mathcal{E}$ occurs, then in at least $\frac{\log n}{36 \cdot (\log n)^{2 \epsilon_t + 9 \epsilon_d}}$ sections, the potential $\Phi$ decreases by a factor of $(1 - \frac{1}{81 (\log n)^{2 \epsilon_t+ 18 \epsilon_d}})$, unless the potential is already smaller than $n \cdot \ce^{-(\log n)^{\epsilon_t/12}}$.
Hence, if  the event $( X \geq \frac{1}{2} \cdot \Ex{X}) \cap \mathcal{E}$ occurs, then
\[
 \Phi^{(t_0+\log n)} \leq \max\left\{ \Phi^{(t_0)} \cdot \left(1 - \frac{1}{81 (\log n)^{2 \epsilon_t+ 18 \epsilon_d}} \right)^{\frac{\log n}{36 \cdot (\log n)^{2 \epsilon_t + 9 \epsilon_d}} }, n \cdot \ce^{-(\log n)^{\epsilon_t/12}}
 \right\}.
\]
Since the initial potential satisfies $\Phi^{(t_0)} \leq n \cdot ( (\log n)^{\epsilon_d} + 1)^{8}$ (cf.~\obsref{potentialinitial}), we have
\begin{align*}
  \Phi^{(t_0+\log n)} &\leq  \max \left\{
 n \cdot ( (\log n)^{\epsilon_d} + 1)^{8} \cdot \mathrm{exp}\left({- \frac{\log n}{2916 \cdot (\log n)^{4 \epsilon_t+27\epsilon_d}   }}\right),  n \cdot \ce^{-(\log n)^{\epsilon_t/12}}
  \right\} \\
  &\leq  \max \left\{
   n  \cdot \mathrm{exp}\left(- \frac{\log n}{(\log n)^{5 \epsilon_t+27\epsilon_d}   }\right), n \cdot \ce^{-(\log n)^{\epsilon_t/12}}
  \right\}.
\end{align*}
Recall that we choose $\epsilon_d = \frac{13}{900}$ and $\epsilon_t = \frac{3}{25}$, which gives that \[1 - 5 \epsilon_t - 27 \epsilon_d = \epsilon_t/12 = \frac{1}{100} .\]
 Thus, the  inequality above and \eq{uboundone} imply
\begin{align*}
\Pro{ \Phi^{(t_0+\log n)} \leq n \cdot \mathrm{e}^{-(\log n)^{ \frac{1}{100}   }} } &\geq 1 - 2 \cdot \ce^{-(\log n)^{\epsilon_t/12}},
\end{align*}
which completes the proof of the first statement.

The proof of the second statement is very similar. Here, $0 < \epsilon_d \leq 13/900$ and $0 < \sigma < 1$ are given, and we choose $\epsilon_t := (1 - 27 \epsilon_d - \sigma) \cdot 21/106 > 0$ and  $\epsilon_{b} := \epsilon_t/20 + \sigma$. By the second statement of \lemref{goodnodess} and the union bound over the time-interval $[t_0,t_0+\log n-1]$,
\begin{align*}
 \Pro{ \mathcal{E} } \geq 1 - \ce^{-(\log n)^{\epsilon_t/20}}.
\end{align*}
Then, by \lemref{inner}, for any load vector $x^{(t)}=x$ and sequence of matchings in the time-interval $[t+1,t+\beta]$  that satisfy $\mathcal{E}^{(t)}$ and $\Phi^{(t)} \geq n \cdot \ce^{-(\log n)^{\epsilon_t/21 + \sigma}}$, it holds that
\begin{align*}
 \Ex{ \Phi^{(t+\beta)} \, \middle| \, x^{(t)} = x } &\leq  \left(1 -  \frac{1}{9 (\log n)^{\epsilon_t + 9 \epsilon_d}} \right) \cdot \Phi^{(t)},
\end{align*}
and as before, we obtain
\begin{align}
\Pro{  \Phi^{(t+\beta)} \leq \left(1-\frac{1}{81\cdot  (\log n)^{2\epsilon_t + 18 \epsilon_d}}\right) \Phi^{(t)} \, \middle| \,  x^{(t)}=x } \geq \frac{1}{18 \cdot (\log n)^{\epsilon_t + 9 \epsilon_d}} =: p. \label{eq:crctwo}
\end{align}
In the following, we call a section $k \in \left\{1,\ldots,(\log n)^{1-\epsilon_t}\right\}$ good if (i) $\mathcal{E}^{(t_0+(k-1) \cdot \beta)}$ does not occur, or (ii) $
\Phi^{(t_0+k \cdot \beta)} \leq   \left(1-\frac{1}{81\cdot  (\log n)^{2\epsilon_t + 18 \epsilon_d}}\right) \Phi^{(t_0+(k-1) \cdot \beta)}
$, or (iii) $\Phi^{(t_0+(k-1) \cdot \beta)} \leq n \cdot \ce^{-(\log n)^{\epsilon_t/21 + \sigma}}$. With $X$ denoting the number of good sections, the Chernoff bound and the union bound yield
\begin{align*}
 \Pro{ \left(X \geq \frac{1}{2} \cdot \Ex{X}\right) \cap \mathcal{E} } \geq 1 - \ce^{-(\log n)^{1- 3 \epsilon_t - 9 \epsilon_d}} - \ce^{-(\log n)^{\epsilon_t/20}}.
\end{align*}
If $(X \geq \frac{1}{2} \cdot \Ex{X}) \cap \mathcal{E}$ occurs, then the potential at the end of the last section is upper bounded by
\begin{align*}
 \Phi^{(t_0+\log n)} &\leq \max \left\{
  n \cdot \mathrm{exp}\left({- \frac{\log n}{ (\log n)^{5 \epsilon_t+27 \epsilon_d}}}\right), n \cdot \ce^{-(\log n)^{\epsilon_t/21+\sigma}}
  \right\}.
\end{align*}
Due to our choice of $\epsilon_t$, we have \[1-5 \epsilon_t - 27 \epsilon_d = \epsilon_t/21 + \sigma,
\] and thus,
\begin{align*}
\Pro{\Phi^{(t_0+\log n)} \leq n \cdot \mathrm{exp}\left({-(\log n)^{ 1 - 27 \epsilon_d - \frac{105}{106} \cdot (1 - 27 \epsilon_d - \sigma)   }}\right)} &\geq 1 - \ce^{-(\log n)^{\epsilon_t/21}},
\end{align*}
and the proof of the second statement is complete.
\end{proof}

\subsection{Proof of \thmref{final}}\label{sec:final}

We start with an outline of the proof of \thmref{final} (confining ourselves to the case where the degree of $G$ is small).
 Some techniques used here are similar to the ones used in proving \thmref{firstsparsification}. For instance, we also  use a potential function and short phases of polylogarithmic length to show that the value of the potential function decreases. However, to prove \thmref{final}, we have to eliminate {\em all} tokens above a constant threshold, while in \thmref{firstsparsification} it is sufficient to make this number smaller than $n \cdot \ce^{-(\log n)^{1-\eps}}$, which is still polynomial in $n$.
 To achieve this improvement, (i) we employ a super-exponential potential function, (ii) switch to a token-based viewpoint, and (iii) exploit the sparseness of the load vector by using the relation between random walks and the movement of tokens from \secref{randomwalk}.

First, fix a node $u$ with $x_u^{(t)} \geq 2$ and consider the next $\beta=\polylog(n)$  rounds. Clearly, the number of nodes from which a token could meet with a token in $B_{\beta}(u)$, i.e., the set of nodes with distance at most $\beta$ to $u$, is at most $d^{2 \beta}$. Now if $\beta$ is small enough, then a relatively straightforward calculation shows that the total number of tokens located at nodes in $B_{\beta}(u)$ is at most $16 (\log n)^{\eps}$ (\lemref{smallh}). Moreover, if the nodes holding these $16 (\log n)^{\eps}$ tokens expand in the graph induced by the random matchings, then we expect the probability for {\em any} pair of these tokens meeting in round $t+\beta$ to be small. This in turn implies that the load along the canonical path starting from $u$ in round $t$ is decreased by at least one within the time-interval $[t+1,t+\beta]$.

Formally, we use a super-exponential potential function $\Lambda^{(t)}$
defined by \begin{align*}
 \Lambda^{(t)} := \sum_{u \in V} \Lambda_{u}^{(t)},
\end{align*}
and
\begin{equation*}
\Lambda_u^{(t)} :=
\begin{cases}
 \ce^{ \frac{1}{8} (\log n)^{1-\eps} \cdot x_{u}^{(t)}}& \mbox{if $x_u^{(t)} \geq 2$,} \\
 0 & \mbox{otherwise,}
\end{cases}
\end{equation*}
where here and in the remainder of the proof, $0 < \epsilon \leq 1/192$ is the constant from the statement of \thmref{final}.


We first derive some basic properties of this potential function. After that, we turn to the more involved task of establishing an expected drop of the potential.
\begin{lemma}\label{lem:potential}
Let $x^{(0)}$ be an arbitrary, non-negative load vector. Then, the following statements  hold:
\begin{itemize}\itemsep 0pt

\item Let $\|x^{(0)}\|_{1} \leq n \cdot \ce^{- (\log n)^{1-\eps}}$. Then, in $t := \tau_{\cont}(n,n^{-2})$ rounds, it holds with probability at least $1-2 n^{-1}$ that $\Lambda^{(t)} \leq 9n^2$.
 \item If two nodes $u$ and $v$ are matched in round $t$ and $x_u^{(t-1)} - x_v^{(t-1)} \geq 2$, then
     \begin{align*}
      \Lambda_u^{(t-1)} + \Lambda_v^{(t-1)} - \Lambda_u^{(t)} - \Lambda_v^{(t)}   \geq  \Lambda_u^{(t-1)} \cdot \left(1 - \frac{2}{\ce^{\frac{1}{8} (\log n)^{1-\eps} }} \right).
\end{align*}
\item For any $u \in V$, the function $\Lambda_u^{(t)}$ is convex in $x_u^{(t)}$, and hence, $\Lambda^{(t)}$ is non-increasing in $t$.
\end{itemize}
\end{lemma}
\begin{proof}
We start with the proof of the first statement.
Since $t\geq \tau_{\cont}(n,n^{-2})$,
 the time-interval $[0,t]$ is $(n,n^{-2})$--smoothing with probability at least $1-n^{-1}$. Fix now all the matchings in $[1,t]$ and consider the orientations of the matching edges  in the time-interval $[1,t]$, which, together with $x^{(0)}$, determine
the load vector $x^{(t)}$. We fix any node $u \in V$ and consider $x_u^{(t)}$. By \lemref{rightprobability} and \lemref{boundbymixinglemma},
\begin{align*}
\Ex{x_u^{(t)}} &\leq  \| x^{(0)} \|_{1} \cdot \max_{v,w \in V} \M_{v,w}^{(t)} \leq 2 \cdot \ce^{-(\log n)^{1 - \eps}} < 1.
\end{align*}
By \lemref{chernofftoken}, we have for any $\delta \geq 1$,
\begin{align*}
  \Pro{ x_u^{(t)} \geq (1+\delta) \Ex{x_u^{(t)}} } \leq \left( \frac{\mathrm{e}}{\delta} \right)^{\delta \cdot \Ex{x_u^{(t)}} }.
\end{align*}
Choosing $\delta = \rho / \Ex{x_u^{(t)} }$ for any real number $\rho \geq 1$ yields
\begin{align*}
  \Pro{ x_u^{(t)} \geq 2\rho } \leq \left( \frac{\mathrm{e} \cdot \Ex{x_u^{(t)}}}{\rho} \right)^{\rho} \leq
    \ce^{-\frac{1}{2} (\log n)^{1 - \eps} \cdot \rho}.
\end{align*}
Hence,
 \begin{align*}
   \Ex{ \Lambda_{u}^{(t)} }
&\leq \sum_{k=2}^{\infty} \Pro{ x_{u}^{(t)} \geq k }
  \cdot \ce^{ \frac{1}{8} (\log n)^{1-\eps} \cdot k} \\
&\leq \sum_{k=2}^{\infty} \ce^{-\frac{1}{4} (\log n)^{1 - \eps} \cdot k} \cdot
\ce^{ \frac{1}{8} (\log n)^{1-\eps} \cdot k} \leq \sum_{k=2}^{\infty} \ce^{-k/8} \leq \frac{1}{1-\ce^{-1/8}} < 9.
\end{align*}
Therefore, we have
 \[\Ex{ \Lambda^{(t)} } = \sum_{u \in V} \Ex{ \Lambda_u^{(t)} } \leq 9n, \]  and by Markov's inequality, it holds that
\[
 \Pro{ \Lambda^{(t)} \geq 9 n^2 } \leq n^{-1}.
\]
This completes the proof of the first statement.

For the second statement, notice that
\begin{align*}
  \Lambda_u^{(t-1)} + \Lambda_v^{(t-1)} - \Lambda_u^{(t)} - \Lambda_v^{(t)}
  &\geq \Lambda_u^{(t-1)} \cdot \left(1 - \frac{\Lambda_u^{(t)} + \Lambda_v^{(t)}}{\Lambda_u^{(t-1)}} \right) \\
  &\geq \Lambda_u^{(t-1)} \cdot
  \left(1 - \frac{2 \cdot \ce^{\frac{1}{8}  (\log n)^{1-\eps}
  \cdot (x_u^{(t-1)}-1)}}{\ce^{\frac{1}{8}  (\log n)^{1-\eps}
  \cdot x_u^{(t-1)}}} \right) \\ &= \Lambda_u^{(t-1)} \cdot \left(1 - \frac{2}{\ce^{\frac{1}{8}  (\log n)^{1-\eps}
  }} \right),
\end{align*}
where, in the second inequality, we used the fact that $x_u^{(t)} \leq x_u^{(t-1)} - 1$ and $x_v^{(t)} \leq x_u^{(t-1)} - 1$.

For the third statement, it suffices to note that $\Lambda_u^{(t)}$ is convex since $x \mapsto \ce^{\frac{1}{8} (\log n)^{1-\eps} \cdot x}$ is convex and $2 \cdot \ce^{\frac{1}{8} (\log n)^{1-\eps} \cdot 2} \leq
\ce^{\frac{1}{8} (\log n)^{1-\eps} \cdot 3}$, which holds for sufficiently large $n$. Therefore, tokens being transferred along a matching edge $\{u,v\} \in \M^{(t)}$ do not increase $\Lambda_u^{(t)}+\Lambda_v^{(t)}$, and hence, $\Lambda^{(t)}$ is non-increasing in $t$.
\end{proof}


\begin{lemma}\label{lem:smallh}
Consider the random matching model. Fix an arbitrary, non-negative load vector $x^{(0)}$ with $\|x^{(0)} \|_{1} \leq n \cdot \ce^{-(\log n)^{1- \eps}}$. Assume that the time-interval $[0,t]$ is $(n,n^{-2})$--smoothing. Then, for any subset of nodes $S \subseteq V$ with $|S| \leq 4 \cdot \ce^{\frac{1}{2} (\log n)^{1-\eps} }$, it holds that
\begin{align*}
  \Pro{ \sum_{u \in S} x_{u}^{(t)} \geq 16\cdot (\log n)^{\eps} } \leq n^{-7}.
\end{align*}
\end{lemma}

\begin{proof}
  Let us consider the total number of tokens located in $S$ at the end of round  $t$. Define
  \[
    Z:= \sum_{i=1}^{\| x^{(0)} \|_{1}   } \chi_{w_i^{(t)} \in S},
  \]
where $w_i^{(t)}$ is the location (node) of the token $i$ at the end of round $t$.
 Since the time-interval $[0,t]$ is $(n,n^{-2})$--smoothing, every token is located at a fixed node in $S$ with probability at most $2/n$ (cf.~\lemref{boundbymixinglemma}), which implies  \begin{align*}
   \Ex{Z} & \leq \left\| x^{(0)} \right\|_{1} \cdot \frac{2 |S|}{n} \leq \ce^{-(\log n)^{1-\eps}} \cdot 8 \cdot \ce^{\frac{1}{2} (\log n)^{1-\eps}} = 8 \cdot \ce^{-\frac{1}{2} (\log n)^{1-\eps}} < 1.
  \end{align*}
 By \lemref{chernofftoken}, we have that
\[
  \Pro{Z \geq (1+\delta)\cdot \Ex{Z}} \leq
  \left(\frac{\ce}{\delta}\right)^{\delta\cdot\Ex{Z}},
\]
 and choosing $\delta = 15 (\log n)^{\eps}/  \Ex{Z}$ gives
  \begin{align*}
   \Pro{ Z \geq 16 (\log n)^{\eps} }
  & \leq  \left(\frac{\ce\cdot \Ex{Z}}{15(\log n)^{\eps} }\right)^{15 (\log n)^{\eps} }  \leq \ce^{-\frac{1}{2} (\log n)^{1-\eps} \cdot 15(\log n)^{\eps} } \leq n^{-7}.
  \end{align*}
 This completes the proof of the lemma.
\end{proof}

After these preparations, we are now able to analyze the drop of the potential function $\Lambda$. First, we consider the case where the graph is sparse, i.e., the degree satisfies $d \leq \ce^{(\log n)^{1/2}}$, and after that, we consider the dense case, where $d > \ce^{(\log n)^{1/2}}$.

\paragraph{Analysis for Sparse Graphs (Random Matching Model)}
For any node $u\in V$ and integer $r$, we define the ball $B_r(u)$ as the set of nodes whose distance to $u$ is at most $r$, i.e., $B_{r}(u) := \{ v \in V \colon \dist(u,v) \leq r \}$. Let $t_0:=\tau_{\cont}(n,n^{-2})=\Oh \big( \frac{\log n}{1-\lambda} \big)$. For any round $t \geq t_0$, define the event $\mathcal{E}^{(t)}$ as
\[
  \mathcal{E}^{(t)} := \bigcup_{u \in V} \left\{ \sum_{v \in B_{r}(u)} x_v^{(t)} \leq 16 \cdot (\log n)^{\eps}  \right\},
\]
where $r := (\log n)^{1/3}$.


\begin{lemma}\label{lem:smallt}
Let $G$ be any $d$-regular graph with $d \leq \ce^{(\log n)^{1/2}}$, and consider the random matching model.
Let $x^{(0)}$ be a non-negative load vector  with
$\| x^{(0)} \|_{1} \leq n\cdot \ce^{-(\log n)^{1-\epsilon}}$. Then, it holds that
\[
 \Pro{ \bigcap_{t=t_0}^{t_0+n-1} \mathcal{E}^{(t)}   } \geq 1 - 2 n^{-1}.
\]
\end{lemma}
\begin{proof}
Since $t_0 = \tau_{\cont}(n,n^{-2})$ by definition, the time-interval
$[0,t_0]$ is $(n,n^{-2})$--smoothing with probability at least $1-n^{-1}$. We assume that this event happens. Clearly, for any round $t \geq t_0$, the time-interval $[0,t]$ is also $(n,n^{-2})$--smoothing. Since for every $u \in V$ and $ \epsilon \leq \frac{1}{192}$,
\[ \left| B_{(\log n)^{1/3}}(u) \right| \leq d^{(\log n)^{1/3}} \leq
\left( \ce^{(\log n)^{1/2} } \right)^{(\log n)^{1/3}} = \ce^{(\log n)^{5/6}} < 4 \cdot \ce^{\frac{1}{2}  (\log n)^{1-\epsilon}}, \]
and, therefore, by \lemref{smallh} and the union bound over $n$ nodes, it holds that $
 \Pro{ \neg \mathcal{E}^{(t)}  } \leq n^{-6}.
$
By the union bound over the time-interval $[t_0,t_0+n-1]$, we have
\[
 \Pro{ \bigcup_{t=t_0}^{t_0+n-1} \left(\neg \mathcal{E}^{(t)} \right)} \leq n^{-5},
\]
which yields the claim of the lemma.
\end{proof}

Next, we lower bound the potential drop of $\Lambda$ for load vectors $x^{(t)}$ satisfying $\mathcal{E}^{(t)}$.

\begin{figure}
\centering
\begin{tikzpicture}[xscale=0.7,yscale=0.7,thick, >=stealth, knoten/.style={circle, scale=0.4, draw=black, fill=black}, rknoten/.style={circle, scale=0.5, draw=red}]
 \draw (10,5) circle (6cm);
 \draw (10,5) circle (3cm);
  \node[knoten] (1) at (10,5) [label=below:$u$] {};
  \node[knoten] (2) at (7,2) [label=below:$w_j^{(t)}$] {};
  \node[knoten] (3) at (11.5,9) [label=below:$w_i^{(t)}$] {};
  \node[knoten] (4) at (7.2,8) [label=left:$w$] {};
  \node[knoten] (5) at (17.5,8.5) [label=below:$w_k^{(t)}$] {};
  \draw[gray, thin, ->] (10,5) -- (13,5);
  \draw[gray, thin, ->] (10,5) -- (15.5,2.7);

  \draw[->, dashed] (7,2.75) -- (6.75,3.75) -- (6.5, 4.5) -- (6.5,5.7) -- (7.25,6.4) -- (7.5,6.3) -- (8,5.8) -- (8.5,5.5) -- (8.75, 5.7) -- (9,6.5) -- (8.5,7) -- (7.5,7.25) -- (7.2,7.9);

  \draw[->, dashed] (11.5,9.75) -- (12.25,9.5) -- (12.5, 9.25) -- (12.5,8) -- (11.85,7.2) -- (11.35, 7) -- (11.2, 7) -- (10.9,7.5) -- (10,8.3) -- (9.5,8.5) -- (8.5,8) -- (7.75,7.5) -- (7.2,7.9);

  \draw[->, dashed] (17.2,9.15) -- (16.4,9) -- (16,8.9) -- (15.8,8.6) -- (15.2,8) -- (14.6,7.8) -- (14,7.6) -- (13.5,7.7) -- (12.8,7.5);

  \draw[->] (10,5) to (9.5,4.5) to (9.25,3.75) to (9,3) to (9.25,2.25);

  \draw[fill=white] (11,9.25) rectangle (12,9.85);
  \draw[fill=white] (6.5,2.25) rectangle (7.5,2.85);

  \draw[fill=white] (9.5,5.25) rectangle (10.5,5.85);
  \draw[fill=white] (9.5,5.85) rectangle (10.5,6.45);
  \draw[fill=white] (9.5,6.45) rectangle (10.5,7.05);
  \draw[fill=white] (9.5,7.05) rectangle (10.5,7.65);

  \draw[fill=white] (17,8.75) rectangle (18,9.35);

 \node at (11.85,5) [above] {$B_{\beta}(u)$};
 \node at (14.4,3.1) [above, rotate=-20] {$B_{2\beta}(u)$};

 \node at (11.5,9.55) [] {$i$};
 \node at (7,2.55) [] {$j$};

 \node at (17.5,9.05) [] {$k$};

  \end{tikzpicture}

\caption{The above diagram illustrates the proof of \lemref{lowdegree}, the key step in proving \thmref{final}. We show that the load along the canonical path starting from $u$ decreases by at least one within the time-interval $[t+1,t+\beta]$ by analyzing all pairs of tokens $i,j$ in $B_{2\beta}(u)$ and showing that none of them is located on the same node $w$ at the end of round $[t+\beta]$. Note that tokens $k$ with $w_{k}^{(t)} \notin B_{2\beta}(u)$ cannot intersect with the canonical path from $u$ within $[t+1,t+\beta]$.
}\label{fig:final}
\end{figure}
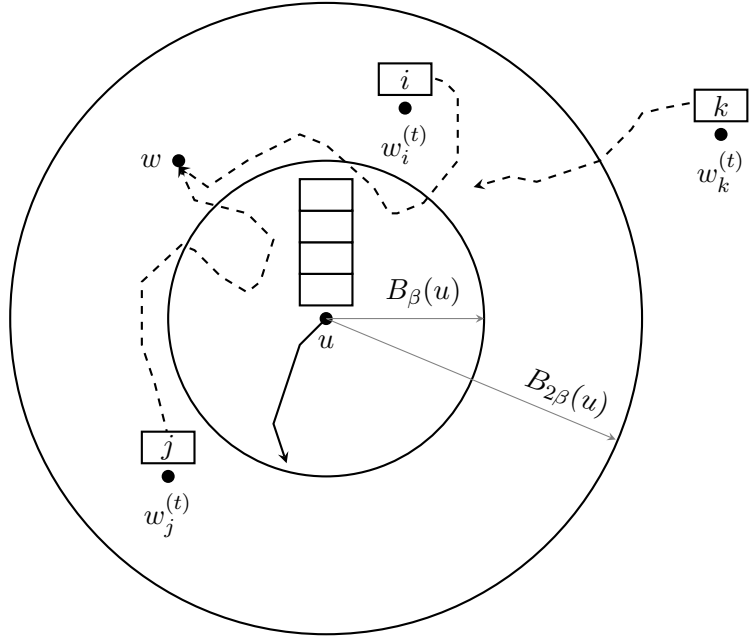

\begin{lemma}\label{lem:lowdegree}
Let $G$ be any $d$-regular graph with $d \leq \ce^{(\log n)^{1/2}}$, and consider the random matching model. Let $x^{(t)}=x$ be any fixed load vector that satisfies $\mathcal{E}^{(t)}$. Then, the following two statements hold:
\begin{itemize}\itemsep 0pt
\item If $\frac{1}{1-\lambda(\mathbf{P})} \leq (\log n)^{1/4}$, then
define $\beta:=\frac{(\log n)^{2 \eps}}{ 1-\lambda(\mathbf{P})}$. Then, there is a constant $c>0$ such that
\begin{align*}
 \Ex{ \Lambda^{(t+\beta)} \, \mid \, x^{(t)} = x } &\leq \ce^{-\frac{c}{4} (\log n)^{2 \eps}} \cdot \Lambda^{(t)},
\end{align*}
\item If $\frac{1}{1-\lambda(\mathbf{P})} > (\log n)^{1/4}$, then define
$\beta := (\log n)^{48 \eps}$. Then,
\begin{align*}
 \Ex{ \Lambda^{(t+\beta)}  \, \mid \, x^{(t)} = x  } &\leq (\log n)^{-\eps}\cdot \Lambda^{(t)}.
\end{align*}
\end{itemize}
\end{lemma}

\begin{proof}
Note that in both cases, $\eps \leq \frac{1}{192}$ is small enough so that $2 \cdot \beta \leq (\log n)^{1/3}= r$ holds.

\textbf{Case 1:} $\frac{1}{1-\lambda(\mathbf{P})} \leq (\log n)^{1/4}$.
Fix any node $u \in V$ with $x_u^{(t)} \geq 2$.
Our goal is to prove that the stack of $x_u^{(t)}$ tokens,
which is located at node $u \in V$, gets smoothed out in the next $\beta$ rounds. To this end, we consider all tokens in the set $B_{2\beta}(u)$ and bound the probability that any two tokens share the same location in round $t+\beta$ (see \figref{final} for an illustration).

By \corref{muthu}, it holds for any node $v$ that
\begin{align*}
  \Pro{ \bigcap_{w \in V} \left\{ \left| \bM_{v,w}^{[t+1,t+\beta]} - \frac{1}{n} \right| \leq \ce^{-c\cdot (\log n)^{2 \epsilon}} \right\}  } \geq 1 - \ce^{-c\cdot(\log n)^{2 \epsilon}},
\end{align*}
where $c>0$ is a constant. Hence, since $\epsilon < 1/4$,
\begin{align}
  \Pro{ \left\| \bM_{v,.}^{[t+1,t+\beta]}\right\|_{\infty} \leq \ce^{-\frac{c}{2}\cdot(\log n)^{2 \eps}}  } \geq 1 - \ce^{-c\cdot(\log n)^{2 \epsilon} }. \label{eq:goodexpansion}
\end{align}

Let us now define the following event:
\[
 \mathcal{A}_u := \bigcap_{v \in B_{2\beta}(u)} \bigg( \left\{ x_v^{(t)} = 0 \right\} \,\bigcup\, \left\{ \left\| \bM_{v,\cdot}^{[t+1,t+\beta]} \right\|_{\infty} \leq \ce^{-\frac{c}{2}\cdot (\log n)^{2 \eps}} \right\} \bigg).
\]
Intuitively, the event $\mathcal{A}_u$ ensures that,
 when looking at the graph induced by the matchings in the interval $[t+1,t+\beta]$, the neighborhood of all nodes in $B_{2\beta}(u)$ containing a token in round $t$ expands.
Since the load vector $x^{(t)}$ satisfies $\mathcal{E}^{(t)}$ by the precondition of the lemma, the total number of tokens in $B_{2\beta}(u)$ is upper bounded by $16\cdot (\log n)^{\eps}$. Consequently, the number of nodes with at least one token is at most $16\cdot (\log n)^{\eps}$, and thus,
\begin{align}
 \Pro{ \mathcal{A}_u }
 &\geq 1 - 16\cdot (\log n)^{\eps} \cdot \ce^{-c \cdot (\log n)^{2 \eps}}
  \geq 1 - \ce^{-\frac{c}{2} \cdot (\log n)^{2 \eps}}\label{eq:upperbounda}.
\end{align}
Let $\tilde{\mathcal{T}}$ denote the set of tokens located in $B_{2\beta}(u)$ at the end of round $t$. Consider any pair of tokens $i,j \in \tilde{\mathcal{T}}$. We upper bound the probability that $i$ and $j$ are located on the same node in round $t+\beta$, conditioned on $\mathcal{A}_u$ as follows:
\begin{align}
\Pro{ w_i^{(t+\beta)} = w_j^{(t+\beta)}  \, \mid \, \mathcal{A}_u  }
 &= \sum_{w \in V} \Pro{ \left\{ w_i^{(t+\beta)} = w \right\} \bigcap  \left\{ w_j^{(t+\beta)} = w \right\} \, \mid \, \mathcal{A}_u }\notag \\
&\leq \sum_{w \in V} \Pro{ w_i^{(t+\beta)} = w \, \mid \, \mathcal{A}_u } \cdot \Pro{  w_j^{(t+\beta)} = w \, \mid \, \mathcal{A}_u } \notag\\
&\leq \max_{k \in V}\left\{ \Pro{ w_i^{(t+\beta)} = k \, \mid \, \mathcal{A}_u } \right\}\cdot \sum_{w \in V} \Pro{  w_j^{(t+\beta)} = w \, \mid \, \mathcal{A}_u  }\notag \\
&\leq  \ce^{-\frac{c}{2}\cdot (\log n)^{2 \eps} } \cdot 1 = \ce^{-\frac{c}{2}\cdot (\log n)^{2 \eps} }\label{eq:boundcollision},
\end{align}
where the first inequality follows by \lemref{randomwalklemma}, and the last inequality follows from the definition of $\mathcal{A}_u$.
Since $x^{(t)}$ satisfies $\mathcal{E}^{(t)}$ and $B_{2\beta}(u) \subseteq B_{(\log n)^{1/3}}(u)$, there are at most
$16\cdot (\log n)^{\eps}$ tokens in $B_{2\beta}(u)$. Hence,
\begin{align*}
 \lefteqn{ \Pro{ \bigcup_{i,j \in \tilde{T}} \left\{ w_i^{(t+\beta)} = w_j^{(t+\beta)}     \right\} } } \\
 &\leq \Pro{ \bigcup_{i,j \in \tilde{T}} \left\{ w_i^{(t+\beta)} = w_j^{(t+\beta)}     \right\} ~~ \bigg| ~~ \mathcal{A}_u  } + \Pro{\neg \mathcal{A}_u  } \\
 &\leq 256 (\log n)^{2 \epsilon} \cdot \max_{i,j \in \tilde{T}}
 \Pro{ w_i^{(t+\beta)} =   w_j^{(t+\beta)}  \, \middle| \, \mathcal{A}_u } + \ce^{-\frac{c}{2}\cdot (\log n)^{2 \eps}} \\
 &\leq 256 (\log n)^{2 \epsilon} \cdot \ce^{- \frac{c}{2}\cdot (\log n)^{2 \eps}} + \ce^{-\frac{c}{2} \cdot (\log n)^{2 \eps}} \leq \ce^{-\frac{c}{3}\cdot (\log n)^{2 \eps}},
\end{align*}
where the second to last inequality follows from \eq{boundcollision}.
Note that in case there are no tokens in $B_{2 \beta}(u)$ that share the same node in round $t+\beta$, then following the canonical path $\mathcal{P}_{u}^{(t)}$, there is at least one round $t' \in [t+1,t+\beta]$ such that $x_{\mathcal{P}_{u}^{(t)}(t')}^{(t')} < x_{\mathcal{P}_{u}^{(t)}(t'-1)}^{(t'-1)}$, i.e., the load of the canonical path $\mathcal{P}_u^{(t)}$ decreases.

Let us now consider the decrease of $\Lambda^{(t)}$ by taking into account all nodes $u \in V$ with $x_{u}^{(t)} \geq 2$.
Since two canonical paths that meet in  round $t'$ cannot reduce their value simultaneously, we obtain by the second statement of \lemref{potential} that
\begin{align}
 \Ex{ \Lambda^{(t)} - \Lambda^{(t+\beta)} \, \mid \, x^{(t)} = x} &\geq \sum_{u \colon x_u^{(t)} \geq 2}  \Lambda_u^{(t)} \cdot \left( 1 - \frac{2}{\ce^{\frac{1}{8} (\log n)^{1-\eps}}} \right) \cdot \left( 1-\ce^{-\frac{c}{3}\cdot (\log n)^{2 \eps}} \right) \notag \\
 &\geq  \sum_{u \colon x_u^{(t)} \geq 2}  \Lambda_u^{(t)} \cdot \left( 1-\ce^{-\frac{c}{4}\cdot (\log n)^{2 \eps}} \right) = \left( 1-\ce^{-\frac{c}{4}\cdot (\log n)^{2 \eps}} \right) \cdot \Lambda^{(t)}, \label{eq:lambdadrop}
\end{align}
where the second inequality holds since $\epsilon < 1/3$. This completes the proof of the first case.

\textbf{Case 2:} $\frac{1}{1-\lambda(\mathbf{P})} > (\log n)^{1/4}$.
We proceed similarly as in the first case, but here we have $\beta= (\log n)^{48 \eps}$.
By \lemref{matchinglocalexpansion}, it holds that for any node $v \in V$
\begin{align}
 \Pro{ \left\| \bM_{v,\cdot}^{[t+1,t+\beta]}   \right\|_{2}^2 \leq (\log n)^{-8 \eps}  } \geq 1 - \ce^{- (\log n)^{32 \eps}.   } \label{eq:badexpansion}
\end{align}
We redefine $\mathcal{A}_u$ as
\[
 \mathcal{A}_u := \bigcap_{ v \in B_{2\beta}(u)} \bigg( \left\{ x_v^{(t)} = 0 \right\} \, \bigcup \, \left\{ \left\| \bM_{v,\cdot}^{[t+1,t+\beta]} \right\|_{\infty} \leq (\log n)^{-4 \eps} \right\} \bigg)  .
\]
As $x^{(t)}$ satisfies $\mathcal{E}^{(t)}$ by the precondition of the lemma, the set $B_{2\beta}(u)$ contains at most $16\cdot (\log n)^{\eps}$ tokens in round $t$. Similar to \eq{upperbounda}, we have
\begin{align*}
 \Pro{ \mathcal{A}_u  } \geq 1 - 16\cdot (\log n)^{\eps} \cdot \ce^{- (\log n)^{32 \eps}   } \geq 1 - \ce^{- \frac{1}{2} (\log n)^{32 \eps}  }.
\end{align*}
As in the first case, we conclude that
\begin{align*}
\Pro{ \bigcup_{i,j \in \tilde{T}} \left\{ w_i^{(t+\beta)} = w_j^{(t+\beta)}     \right\} } &\leq \Pro{ \bigcup_{i,j \in \tilde{T}} \left\{ w_i^{(t+\beta)} = w_j^{(t+\beta)}     \right\} \, \bigg| \, \mathcal{A}_u } + \Pro{\neg  \mathcal{A}_u } \\
 &\leq  256 \cdot (\log n)^{2 \eps} \cdot (\log n)^{-4 \eps} + \ce^{-\frac{1}{2} (\log n)^{32 \eps}   } \leq (\log n)^{-\frac{3}{2} \eps},
\end{align*}
and as in \eq{lambdadrop}, we conclude that
\[
 \Ex{ \Lambda^{(t + \beta)} \, \mid \, x^{(t)} = x } \leq (\log n)^{-\eps} \cdot \Lambda^{(t)},
\]
which finishes the proof.
\end{proof}

\paragraph{Analysis for Dense Graphs (Random Matching Model)}
We now consider the dense case where the degree of the graph satisfies $d > \ce^{(\log n)^{1/2}}$.
This case turns out to be easier than the sparse case, as the average load around every node is smaller than $1/2$ for all nodes and all rounds with high probability. Hence, most  neighbors of each node have zero load, which implies that, as long as the maximum load is larger than $1$, there is an expected exponential drop of the potential function within a single round.

To formalize this, let $t_0:=\tau_{\cont}(n,n^{-2})=\Oh \big( \frac{\log n}{1-\lambda} \big)$. Define for any round $t \geq t_0$ the following event:
\[
 \mathcal{F}^{(t)} := \bigcap_{u \in V} \left\{ \sum_{v \in N(u)} x_v^{(t)} \leq \frac{d}{2} \right\}.
\]
Similar to \lemref{smallt}, we now prove the following lemma:

\begin{lemma}\label{lem:bigt}
Let $G$ be any $d$-regular graph with $d > \ce^{(\log n)^{1/2}}$, and consider the random matching model. Let $x^{(0)}$ be a non-negative load vector $x^{(0)}$ with
$\| x^{(0)} \|_{1} \leq n\cdot \ce^{-(\log n)^{1-\epsilon}}$. Then, it holds that
\[
 \Pro{ \bigcap_{t=t_0}^{t_0+n-1} \mathcal{F}^{(t)}   } \geq 1 - 2 n^{-1}.
\]
\end{lemma}
\begin{proof}
Since $t_0 = \tau_{\cont}(n,n^{-2})$ by definition, the time-interval
$[0,t_0]$ is $(n,n^{-2})$--smoothing with probability at least $1-n^{-1}$. For the rest of the proof, assume that this event happens. Clearly, this implies that for any round $t \geq t_0$, the time-interval $[0,t]$ is also $(n,n^{-2})$--smoothing.
Let us first lower bound
$\Pro{   \mathcal{F}^{(t)}}$ for any fixed round $t \geq t_0$.
Consider any round $t \geq t_0$ and fix a node $u \in V$. Let
\[
  Z := \sum_{v \in N(u)} x_v^{(t)} = \sum_{i=1}^{\left\|x^{(0)}\right\|_1} \chi_{w_i^{(t)} \in N(u)}.
\]
Since $[0,t]$ is $(n,n^{-2})$--smoothing, every token is located at a fixed node in $S$ with probability at most $2/n$ (cf.~\lemref{boundbymixinglemma}), and therefore,
\[
  \Ex{Z} \leq |N(u)| \cdot \|x^{(0)}  \|_1 \cdot \frac{2}{n} \leq 2 d \cdot \ce^{-(\log n)^{1-\eps}}.
\]
By \lemref{chernofftoken},
$
 \Pro{ Z \geq (1+ \delta) \Ex{Z} } \leq \left( \frac{\ce}{\delta} \right)^{\delta \Ex{Z}},
$
and choosing $\delta = d/(4 \Ex{Z})$ yields
\begin{align*}
 \Pro{ Z \geq d/2 } &\leq 	\left( \frac{4\cdot\mathrm{e}\cdot \Ex{Z}}{d} \right)^{\frac{d}{4} } = n^{-\omega(1)}.
\end{align*}
By the union bound and recalling that with probability at least $1-n^{-1}$, all the intervals $[0,t], t \geq t_0$, are $(n,n^{-2})$-smoothing, and
\[
  \Pro{ \bigcap_{t=t_0}^{t_0+n-1}  \mathcal{F}^{(t)}  } = 1 - n^{-\omega(1)} - n^{-1} \geq 1 - 2n^{-1},
\]
which completes the proof.
\end{proof}

The next lemma is similar to \lemref{lowdegree} from the sparse graph case.
\begin{lemma}\label{lem:highdegree}
Let $G$ be any $d$-regular graph with  $d \geq \ce^{(\log n)^{1/2}}$, and consider the random matching model. Assume that $x^{(t)}=x$ is any non-negative load vector that satisfies $\mathcal{E}^{(t)}$. Then, it holds that
\[
 \Ex{ \Lambda^{(t+1)} \, \mid \, x^{(t)}=x } \leq \left( 1-\frac{c_{\min}}{4} \right) \cdot \Lambda^{(t)}.
\]
\end{lemma}
\begin{proof}
By the second statement of \lemref{potential},
\begin{align*}
 \Ex{ \Lambda^{(t)} - \Lambda^{(t+1)} \, \mid \, x^{(t)}=x}
 &\geq
 \sum_{\substack{u \in V \colon \\ x_u^{(t)} \geq 2}} \sum_{\substack{v \in N(u) \colon \\ x_v^{(t)} = 0}} \Pro{ \{u,v\} \in \M^{(t)} }
 \cdot \left(	\Lambda_u^{(t)} \cdot \left(1 - \frac{2}{\ce^{\frac{1}{8} (\log n)^{1-\eps} }} \right)  \right) \\
 &\geq \frac{1}{2} \cdot \sum_{\substack{u \in V \colon \\ x_u^{(t)} \geq 2}} \sum_{\substack{v \in N(u) \colon \\ x_v^{(t)} = 0}} \frac{c_{\min}}{d}
 \cdot \Lambda_u^{(t)}.
\end{align*}
Since $x^{(t)}$ satisfies $\mathcal{E}^{(t)}$, we know that for every node $u \in V$, at least half of the neighbors $v \in N(u)$ satisfy $x_v^{(t)} = 0$. Therefore,
\begin{align*}
 \Ex{ \Lambda^{(t)} - \Lambda^{(t+1)}  \, \mid \, x^{(t)}=x } &\geq \frac{c_{\min}}{4} \cdot \sum_{\substack{u \in V \colon \\ x_u^{(t)} \geq 2}} \Lambda_u^{(t)}  = \frac{c_{\min}}{4} \cdot \Lambda^{(t)},
\end{align*}
and thus,
$
  \Ex{ \Lambda^{(t+1)} \, \mid \, x^{(t)}=x} \leq (1-\frac{c_{\min}}{4}) \cdot \Lambda^{(t)}
  $.
\end{proof}

Finally, we are able to prove \thmref{final}.

\begin{proof}
We give the proof for the random matching model  first (where $d$ denotes the degree of the graph) and consider the balancing circuit model at the end of the proof (where $d$ represents the number  of matchings applied periodically).

\textbf{Case~1: (Random Matching Model) $d \leq \ce^{(\log n)^{1/2}}$ and $\frac{1}{1-\lambda(\mathbf{P})} \leq (\log n)^{1/4}$.}
By the first statement of~\lemref{potential}, it holds with probability $1-2n^{-1}$ that
$ \Lambda^{(t_0)} \leq 9n^{2} $,
where $t_0 := \tau_{\cont}(n,n^{-2})$. Let $\mathcal{E} := \bigcap_{t=t_0}^{t_0+n-1} \mathcal{E}^{(t)}$.
By \lemref{smallt}, it holds that
\[
  \Pro{ \mathcal{E} } \geq 1 - 2n^{-1}.
\]
Let $\beta:=\frac{(\log n)^{2 \epsilon}}{1-\lambda(\mathbf{P})}$ and divide the time-interval $\left[t_0,t_0+\frac{C \log n}{1-\lambda(\mathbf{P})}\right]$ into $C \log n/(\log n)^{2 \epsilon}$ consecutive sections of length $\beta$ each, where $C > 0$ is
 a large constant.
By \lemref{lowdegree}, we have for any section $k \in \left\{1,\ldots,C \log n/(\log n)^{2 \epsilon}\right\}$,
\[
  \Ex{ \Lambda^{(t_0+k\cdot \beta)} \, \mid \, x^{(t_0+(k-1) \cdot \beta)}=x } \leq \ce^{- \frac{c}{4}\cdot (\log n)^{2 \epsilon} } \cdot \Lambda^{(t_0+(k-1) \cdot \beta)},
\]
where $x$ is any load vector that satisfies $\mathcal{E}^{(t+(k-1) \cdot \beta)}$. Hence, by Markov's inequality,
\begin{align}
 \Pro{  \Lambda^{(t_0+k\cdot \beta)} \leq 2 \cdot \ce^{- \frac{c}{4}\cdot (\log n)^{2 \epsilon} } \cdot \Lambda^{(t_0+(k-1) \cdot \beta)} \, \mid \, x^{(t_0+(k-1) \cdot \beta)}=x }  \geq \frac{1}{2}. \label{eq:goodlower}
\end{align}
In the following, we call a section $k \in \left\{1,\ldots,C \log n/(\log n)^{2 \epsilon}\right\}$ good if (i) $\mathcal{E}^{(t_0+(k-1) \cdot \beta)}$ does not occur, or (ii) $\Lambda^{(t_0+k\cdot \beta)} \leq 2 \ce^{- \frac{c}{4}\cdot (\log n)^{2 \epsilon} } \cdot \Lambda^{(t_0+(k-1) \cdot \beta)}$. Then, by \eq{goodlower}, every section $k$ is good with probability at least $1/2$, independently of all previous sections. Let $X$ be the random variable denoting the number of good sections. Then, $\Ex{X} \geq \frac{C \log n}{(\log n)^{2 \epsilon}} \cdot \frac{1}{2}$, and by the Chernoff bound,
\[
 \Pro{X \geq \frac{C \log n}{4 (\log n)^{2 \epsilon}} } \geq \Pro{ X \geq \frac{1}{2} \cdot \Ex{X} } \geq 1 - \ce^{-\frac{\Ex{X}}{8}} \geq 1 - \ce^{-(\log n)^{1-3 \eps}},
\]
and by the union bound,
\[
 \Pro{ \left( X \geq \frac{1}{2} \cdot \Ex{X} \right) \bigcap \mathcal{E} } \geq 1 - 2n^{-1} - \ce^{-(\log n)^{1-3 \eps}} \geq 1 - 2 \ce^{-(\log n)^{1-3 \eps}}.
\]
If $\left( X \geq \frac{1}{2} \cdot \Ex{X} \right) \cap \mathcal{E}$ occurs, then in at least $C \log n/(4 (\log n)^{2 \epsilon})$ sections, the potential $\Lambda$ decreases by a factor of $2\cdot\ce^{- \frac{c}{4}\cdot (\log n)^{2 \epsilon} }$, and the potential at the end of the last section can be bounded from above by
\[
  \Lambda^{\left(t_0 + \frac{C \log n}{(\log n)^{2 \eps}} \cdot \beta\right)} \leq \Lambda^{(t)} \cdot \left(2 \ce^{- \frac{c}{4}\cdot (\log n)^{2 \epsilon}} \right)^{ \frac{C \log n}{4 (\log n)^{2 \epsilon}}} \leq 9n^{2} \cdot 2^{ \frac{C \log n}{4 (\log n)^{2 \epsilon}}  } \cdot \ce^{- \frac{C \cdot c}{16}\cdot \log n},
\]
which is smaller than $1$ for $C:=48/c$. Recalling that with probability at least $1-2n^{-1}$, $\Lambda^{(t_0)} \leq 9n^2$, we conclude by the union bound that
\[
  \Pro{ \Lambda^{\left(t_0 + \frac{ 48/c\cdot \log n}{1-\lambda(\P)}\right)} = 0} \geq 1 - 2 \ce^{-(\log n)^{1-3 \eps}} - 2 n^{-1} \geq 1 - 3 \ce^{-(\log n)^{1-3 \eps}},
\]
which completes the proof of the first case, as $\Lambda^{\left(t_0 + \frac{ 48/c\cdot \log n}{1-\lambda(\P)}\right)} = 0$ implies that the maximum load of $x^{\left(t_0 + \frac{ 48/c\cdot \log n}{1-\lambda(\P)}\right)}$ is $1$.

%
\textbf{Case~2: (Random Matching Model) $d \leq \ce^{(\log n)^{1/2}}$ and $\frac{1}{1-\lambda(\mathbf{P})} > (\log n)^{1/4}$.}
The proof of this case is very similar to Case~1. Here, we choose $\beta:=(\log n)^{48 \eps}$. Divide the time-interval $[t_0,t_0+ \beta \cdot \log n]$ into $\log n$ sections of length $\beta$. By the second statement of \lemref{lowdegree}, it holds for  any section $k \in \{1,\ldots, \log n \}$ that
\[
  \Ex{ \Lambda^{(t_0+k \cdot \beta)} \, \mid \, x^{(t_0+(k-1) \cdot \beta)}=x } \leq (\log n)^{-\eps } \cdot \Lambda^{(t_0+(k-1) \cdot \beta)},
\]
where $x$ is any load vector that satisfies $\mathcal{E}^{(t_0+(k-1) \cdot \beta)}$. Hence, by Markov's inequality,
\begin{align*}
  \Pro{ \Lambda^{(t_0+k \cdot \beta)} \leq 2 \cdot (\log n)^{-\eps} \cdot \Lambda^{(t_0+(k-1) \cdot \beta)} \, \mid \, x^{(t_0+(k-1) \cdot \beta)}=x  } \geq \frac{1}{2}.
\end{align*}
 In the following, we call a section $k \in \{1,\ldots,\log n\}$ good if
(1) $\mathcal{E}^{(t_0+(k-1) \cdot \beta)}$ does not occur or (2) $\Lambda^{(t_0+k\cdot \beta)} \leq 2 (\log n)^{-\eps} \cdot \Lambda^{(t_0+(k-1) \cdot \beta)}$.
 Applying a Chernoff bound to the random variable $X$ counting the number of good sections,
\[
  \Pro{ X \geq \frac{\log n}{4} } \geq \Pro{ X \geq \frac{1}{2} \cdot \Ex{X} } \geq 1 - \ce^{-\frac{\log n}{8}}.
\]
Since $\Pro{ \mathcal{E} } \geq 1 - 2n^{-1}$, we conclude by the union bound that with probability at least $1 - 2 \ce^{-\frac{\log n}{8}}$,
the potential at the end of the last section is bounded from above by
\[
 \Lambda^{(t_0+\beta \cdot \log n)} \leq \Lambda^{(t_0)} \cdot (2 \log n)^{-\eps \cdot \frac{\log n}{4} } = 9 n^{2} \cdot n^{-\omega(1)}.
\]
Since $\epsilon \leq \frac{1}{192}$, we have
$
  \beta \cdot \log n \leq (\log n)^{1/4} \cdot \log n \leq \frac{\log n}{1-\lambda(\mathbf{P})},
$ and the analysis of the second case is complete.

\textbf{Case~3: (Random Matching Model) $d > \ce^{(\log n)^{1/2}}$.}
The proof of this case is the same as Case~1 and 2. In fact, it is even slightly
simpler, because \lemref{highdegree} implies an exponential drop of the potential $\Lambda$ for every single round.

\textbf{Case~4: (Balancing Circuit Model).}
Finally, we consider the balancing circuit model with a sequence of $d=\Oh(1)$ matchings. Since $d$ is a constant, it holds for every node $u \in V$ that at most $d$ neighbors of $u$ in $G$ appear in the matching matrices $\M^{(1)},\ldots,\M^{(d)}$. For this reason, we may assume that the underlying graph $G$ has bounded maximum degree. Moreover,
we can apply \lemref{wellknown} and \lemref{circuitlocalexpansion} to conclude that for all nodes $v \in V$ and any $\beta \in \N$,
\begin{align*}
  \left\| \bM_{v,\cdot}^{\beta} \right\|_{\infty} \leq \frac{1}{n} +  \min \left\{ (\lambda(\M))^{\beta/2},\Oh(\beta^{-1/8}) \right\}, 
\end{align*}
which corresponds to the inequalities \eq{goodexpansion} and \eq{badexpansion} for the random matching model.
Therefore, the same statements as in \lemref{smallh}, \lemref{smallt}, and \lemref{lowdegree} hold for the balancing circuit model with the only difference that any time-interval containing $x$ rounds in the random matching model is replaced by $x \cdot d$ rounds in the balancing circuit model. Consequently, the analysis of the balancing circuit model is the same as the analysis of the random matching model (Case 1 and Case 2).
\end{proof}

\section{The Diffusion Model}\label{sec:diffusion}

In the diffusion model, the (continuous) load vector $\xi^{(t)} \in \mathbb{R}^n$ in round $t \geq 1$ is given by the recursion
$
\xi^{(t)}=\xi^{(t-1)}\mathbf{P}, $
 where, for any $\gamma\geq 1$, the {\em diffusion matrix} $\mathbf{P}=\P(\gamma)$  of graph $G$ is defined as follows:  For any pair of nodes $u$ and $v$,
$\P_{u,v}=\frac{1}{\gamma\Delta}$ if $\{u,v\}\in E$, $\P_{u,v}=1-\frac{\deg(u)}{\gamma\Delta}$ if $u=v$, and $\P_{u,v}=0$, otherwise. Hence,
\begin{align*}
\xi^{(t)}_u
&=  \xi^{(t-1)}_u \cdot \left(1 - \frac{\deg(u)}{\gamma \Delta} \right) + \sum_{v \colon \{u,v\}\in E}
\xi^{(t-1)}_v \cdot \frac{1}{\gamma \Delta}
\\ &=\xi^{(t-1)}_u+\sum_{v \colon \{u,v\}\in E} \frac{\xi^{(t-1)}_v-\xi^{(t-1)}_u}{\gamma\Delta}.
\end{align*}

Common choices are $\gamma=2$ (resulting in a loop probability at least $1/2$) or $\gamma=1+1/\Delta$, both ensuring convergence also on bipartite graphs.
Since $\P$ is symmetric, it has $n$ real eigenvalues
$\lambda_1(\P)\geq\dots
 \geq\lambda_n(\P)$, and we define $\lambda(\P):=\max \{|\lambda_2(\P)|,|\lambda_n(\P)|\}$. Note that $\lambda(\P)$ depends on $\gamma$.
 Let
 $\P^{t}$ be the $t$-th power of $\P$ and $\P^{0}$ be the $n$ by $n$ identity matrix.

%

As for the matching model, there is a natural upper bound on the convergence in terms of the spectral gap of $\bP$.
\begin{theorem}[{\cite[Thm.~1]{RSW98}}]\label{thm:continuous}
    Let $G$ be any graph, and consider the continuous case. Then, for any $\epsilon > 0$, the discrepancy is at most $\epsilon$ after $\frac{ 2}{1-\lambda(\P)}\cdot \log( \frac{K\,n^2}{ \epsilon}  )
    $ rounds for any initial load vector with discrepancy at most $K$.
\end{theorem}

\begin{theorem}\label{thm:continuouslower}
Let $G$ be any $d$-regular graph, and consider the continuous case with $\gamma \geq 2$. Then, for any $\epsilon > 0$ and $K > 0$, there is an initial load vector with discrepancy $K$ so that it takes at least $ \frac{1}{8\cdot(1-\lambda(\P))} \cdot \log( \frac{K}{4\epsilon n})$ rounds to achieve a discrepancy smaller than $\epsilon$.
\end{theorem}
\begin{proof}
The proof is very similar to the proof of \thmref{continuouslowermatching} in the randomized matching model, except that here the transition matrix is $\P$ instead of $\mathbf{Q}$.

Define
\[
\Delta_w(t):=\frac{1}{2}\sum_{v\in V}\left|
\mathbf{P}_{v,w}^t-\frac{1}{n}\right|,
\]
and $\tau_w(\tilde{\varepsilon}):=\min\{
t:\Delta_w(t')\leq \tilde{\varepsilon} \mbox{ for all } t'\geq t
\}$.
By
\cite[Proposition~1, Part~(ii)]{Si92}, it holds that
\[
  \tau := \max_{w \in V} \tau_{w}(\tilde{\epsilon}) \geq \frac{1}{2} \cdot \frac{\lambda(\mathbf{P})}{1-\lambda(\mathbf{P})} \cdot \log \left( \frac{1}{2 \tilde{\epsilon}} \right).
\]
Hence, for iteration $\tau-1$, there exists a vertex $w \in V$ with
\[
\frac{1}{2}\sum_{v\in V}\left|
\mathbf{P}_{v,w}^{\tau-1}-\frac{1}{n}\right|\geq \tilde{\varepsilon}.
\]
By setting $\tilde{\varepsilon}=2\varepsilon\cdot n/K$, we have that $\tau=\frac{1}{2} \cdot \frac{\lambda(\P)}{1-\lambda(\P)} \cdot \log( \frac{K}{4\epsilon n})$, and
\[
\frac{1}{2}\sum_{v\in V}\left|
\mathbf{P}_{v,w}^{\tau-1}-\frac{1}{n}\right|\geq \frac{2\varepsilon\cdot n}{K}.
\]
Hence,  there exists
a vertex $v \in V$ such that
\begin{align*}
  \left| \P^{\tau-1}_{v,w} - \frac{1}{n} \right| > \frac{\epsilon}{K}. 
\end{align*}
Now define a load vector $\xi^{(0)}$ by $\xi^{(0)}_u = K$ if $u=v$, and $\xi^{(0)}_u=0$, otherwise. Then, for any vertex $w \in V$, $\xi^{(\tau-1)}_w = \xi^{(0)}_v  \cdot \P_{v,w}^{\tau-1} = K \cdot \P_{v,w}^{\tau-1}$. Therefore,  the discrepancy in round $\tau-1$ is larger than $K \cdot \frac{\epsilon}{K} = \epsilon$.

Since the sum of all eigenvalues of $\frac{1}{d} \cdot \mathbf{A}$ equals $0$ and $\lambda_{1}(\frac{1}{d} \cdot \mathbf{A})=1$, we have that
\[
\lambda_2(\P) \geq \frac{1}{2} + \frac{1}{2}\cdot \lambda_2\left(\frac{1}{d} \cdot \mathbf{A}\right) \geq \frac{1}{2} - \frac{1}{2\cdot(n-1)},\]
which completes the proof.
\end{proof}

\subsection{The Discrete Case}

We study two natural protocols in the discrete case of the diffusion model. The first protocol is the vertex-based protocol \cite{BCFFS11}, where excess tokens are allocated by vertices. The second is the edge-based protocol \cite{FGS10}, where every edge performs an independent randomized rounding.

The {\em vertex-based protocol} from \cite{BCFFS11} works for $d$-regular graphs as follows. In every round $t$, every node $u$ first sends $\left\lfloor \frac{x_u^{(t-1)}}{d+1}\right \rfloor$ tokens to each neighbor and keeps the same amount of tokens for itself. After that, the remaining $x_u^{(t-1)} - (d+1) \cdot \left\lfloor \frac{x_u^{(t-1)}}{d+1}\right\rfloor$ tokens at node $u$ are randomly distributed (without replacement) among node $u$ and its $d$ neighbors. This corresponds to a diffusion matrix $\bP$ with $\gamma=1+1/d$.

Consider now the {\em edge-based protocol} from \cite{FGS10}, where the load sent along each edge is obtained by randomly rounding the flow that would be sent in the
continuous case to the nearest integer. While the protocol itself is quite natural, it may happen that negative loads occur during the execution. Let $x^{(0)} = \xi^{(0)} \in \mathbb{Z}^n$. Similar to the matching model, we now derive an expression for the deviation between the discrete and continuous models at some node $w$ in round $t$:
\begin{align}
  x_w^{(t)} - \xi_w^{(t)}
  &= \sum_{s=1}^{t} \sum_{u\in V} e_u^{(s)} \P^{t-s}_{u,w} = \sum_{s=1}^{t} \sum_{u\in V}
     \sum_{v \colon \{u,v\} \in E} e_{u,v}^{(s)} \P^{t-s}_{u,w} \notag\\
   &= \sum_{s=1}^{t} \sum_{ [u:v] \in E}
     e_{u,v}^{(s)} \, \left( \P^{t-s}_{u,w} - \P^{t-s}_{v,w} \right),
  \label{eq:diffstd}
\end{align}
where $e_{u,v}^{(s)}$ is the rounding error for each edge $[u:v] \in E$ in round $s$, defined by
\[
e_{u,v}^{(s)}=\left\{ \begin{aligned}
         \left\lceil \frac{ \xi_{v}^{(s-1)} - \xi_{u}^{(s-1)}}{\gamma \Delta}   \right\rceil - \frac{ \xi_{v}^{(s-1)} - \xi_{u}^{(s-1)}}{\gamma \Delta} &\quad\quad \mbox{w.\,p. }\ \frac{ \xi_{v}^{(s-1)} - \xi_{u}^{(s-1)}}{\gamma \Delta} - \left\lfloor \frac{ \xi_{v}^{(s-1)} - \xi_{u}^{(s-1)}}{\gamma \Delta}   \right\rfloor, \\
                  \left\lfloor \frac{ \xi_{v}^{(s-1)} - \xi_{u}^{(s-1)}}{\gamma \Delta}   \right\rfloor - \frac{ \xi_{v}^{(s-1)} - \xi_{u}^{(s-1)}}{\gamma \Delta}&\quad\quad \mbox{w.\,p. }\ \left\lceil \frac{ \xi_{v}^{(s-1)} - \xi_{u}^{(s-1)}}{\gamma \Delta}  \right\rceil - \frac{ \xi_{v}^{(s-1)} - \xi_{u}^{(s-1)}}{\gamma \Delta}.
                          \end{aligned} \right.
\]
Moreover, let $e_{u,v}^{(s)}=0$ if
$\frac{ \xi_{v}^{(s-1)} - \xi_{u}^{(s-1)}}{\gamma \Delta}$ is an integer.
By definition, we have $\Ex{ e_{u,v}^{(s)} } = 0$.  Further, for any set of different (not necessarily disjoint) edges, the rounding errors within the same round are mutually independent.

\subsection{Local Divergence and Discrepancy}

Based on the deviation between the discrete and continuous cases \eq{diffstd},
we now recall the definition of the (refined) local $p$-divergence for the diffusion model from \cite{FS09,BCFFS11}, which extends the original definition of \cite{RSW98}.

\begin{defi}[Local $p$-Divergence for Diffusion]
For any graph $G$ and $p \in \mathbb{Z}_{+}$,  the local $p$-divergence is
\[
\Psi_p(\P)=\max_{w\in V} \left( \sum_{t=0}^{\infty}\sum_{[u:v]\in E} \left|\P^{t}_{u,w}-\P^{t}_{v,w} \right|^p  \right)^{1/p},
\]
and the refined local $p$-divergence is
\[
\Upsilon_p(\P)=\max_{w\in V} \left( \frac{1}{2} \sum_{t=0}^{\infty} \sum_{u\in V} \max_{v \in N(u)} \left|\P^{t}_{u,w}-\P^{t}_{v,w} \right|^p  \right)^{1/p}.
\]
\end{defi}
Clearly, $\Upsilon_p(\P) \leq \Psi_p(\P)$.
We now present our bounds on the (refined) local $2$-divergence.



\begin{theorem}\label{thm:divergencediffusion}
Let $G$ be any graph, and $\gamma>1$, which is not necessarily constant. Then, it holds that
\[
\Upsilon_2(\P) \leq \Psi_2(\P) \leq \sqrt{ \frac{\gamma\cdot \Delta}{2-2/\gamma} }.
\]
Moreover, for any $\gamma \geq 1$, we have $\Psi_2(\P) \geq \sqrt{\Delta}$  and
\[\Upsilon_2(\P) \geq \sqrt{ \frac{1 + \Delta}{2} }.\]
\end{theorem}

The upper bound on $\Psi_2(\P)$ is minimized for $\gamma=2$ and becomes, in that case, $\sqrt{2 \cdot \Delta}$. This result significantly improves over the previous bounds in \cite{BCFFS11}, which all depend on the spectral gap $1-\lambda(\bP)$, or are restricted to special networks. The analysis of the edge-based algorithm in \cite{FGS10} did not use $\Psi_2(\P)$ or $\Upsilon_2(\bP)$, but the bound on the discrepancy depends on the spectral gap.

\begin{proof}
The proof of \thmref{divergencediffusion} uses a similar approach based on the same potential function as in \cite[Lemma~1]{BFZ09}.  However, we  perform a  more precise analysis to handle the case where $\gamma$ is very close to $1$, for instance for the vertex-based protocol, we have $\gamma=1+1/d$. By contrast, the proof in \cite[Lemma~1]{BFZ09} is based on a sequential exposure of the edges and only works if $\gamma \geq 4$.

Fix any node $w\in V$.
Define the potential function in round $t$ by
\[
  \Phi^{(t)} := \sum_{u\in V}\left( \bP_{u,w}^{t} - \frac{1}{n} \right)^2,
\]
so $\Phi^{(0)} = 1- \frac{1}{n}$.

We first prove the upper bound on $\Phi^{(t)}$. Let  $y_u := \bP_{u,w}^{t-1}$. By definition of the diffusion model, we have
\begin{align*}
  \Phi^{(t)} &= \sum_{u\in V} \left( \bP_{u,w}^{t} - \frac{1}{n} \right)^2 \\
         &= \sum_{u\in V} \left( \Big(1 - \frac{\deg(u)}{\gamma \Delta} \Big) \cdot y_u + \left( \sum_{v \in N(u)} \frac{1}{\gamma \Delta} \cdot y_v \right) - \frac{1}{n} \right)^2 \\
         &= \sum_{u\in V} \left( \Big(1 - \frac{\deg(u)}{\gamma \Delta} \Big) \cdot \Big( y_u - \frac{1}{n} \Big) + \sum_{v \in N(u)} \frac{1}{\gamma \Delta} \Big(y_v - \frac{1}{n} \Big) \right)^2 \\
         &= \sum_{u\in V} \left[ \sum_{v \in N(u)}  \frac{1}{\deg(u)} \cdot \left( \Big(1 - \frac{\deg(u)}{\gamma \Delta} \Big) \cdot \Big( y_u - \frac{1}{n} \Big) + \frac{\deg(u)}{\gamma \Delta} \Big(y_v - \frac{1}{n} \Big) \right) \right]^2.
\intertext{Using the notation $\EXX{v \in N(u)}{X(v)}$ for the expection of a random variable $X(v)$, where $v \in N(u)$ is chosen uniformly at random, we can rewrite the above expression as}
         \Phi^{(t)}&= \sum_{u\in V} \left( \EXX{v \in N(u)}{   \Big(1 - \frac{\deg(u)}{\gamma \Delta} \Big) \cdot \Big( y_u - \frac{1}{n} \Big) + \frac{\deg(u)}{\gamma \Delta}\cdot \Big(y_v - \frac{1}{n} \Big)     } \right)^2, \\
         \intertext{and applying Jensen's inequality gives}
      \Phi^{(t)}   &\leq \sum_{u\in V}  \EXX{v \in N(u)}{   \left( \Big(1 - \frac{\deg(u)}{\gamma \Delta} \Big) \cdot \Big( y_u - \frac{1}{n} \Big) + \frac{\deg(u)}{\gamma \Delta} \cdot\Big(y_v - \frac{1}{n} \Big) \right)^2     } \\
         &= \sum_{u\in V} \sum_{v \in N(u)} \frac{1}{\deg(u)} \cdot \left( \Big(1 - \frac{\deg(u)}{\gamma \Delta} \Big) \cdot \Big( y_u - \frac{1}{n} \Big) + \frac{\deg(u)}{\gamma \Delta}\cdot \Big(y_v - \frac{1}{n} \Big) \right)^2  \\
         &= \sum_{[u:v] \in E} \Biggl\{ \frac{1}{\deg(u)} \cdot \left( \Big(1 - \frac{\deg(u)}{\gamma \Delta} \Big) \cdot \Big( y_u - \frac{1}{n} \Big) + \frac{\deg(u)}{\gamma \Delta}\cdot \Big(y_v - \frac{1}{n} \Big) \right)^2  \\
         &\quad\quad\quad\quad \, + \frac{1}{\deg(v)} \cdot \left( \Big(1 - \frac{\deg(v)}{\gamma \Delta} \Big) \cdot \Big( y_v - \frac{1}{n} \Big) + \frac{\deg(v)}{\gamma \Delta}\cdot \Big(y_u - \frac{1}{n} \Big) \right)^2 \Biggr\}.
\end{align*}
Note that
\[
  \Phi^{(t-1)} = \sum_{u\in V} \Big( y_u - \frac{1}{n} \Big)^2 = \sum_{[u:v] \in E} \left\{ \frac{1}{\deg(u)}\cdot \Big( y_u - \frac{1}{n} \Big)^2 + \frac{1}{\deg(v)}\cdot \Big( y_v - \frac{1}{n} \Big)^2 \right\},
  \]
and using the upper bound on $\Phi^{(t)}$ from above, we obtain
\begin{align*}
\lefteqn{ \Phi^{(t-1)} - \Phi^{(t)} }  \\ &\geq \sum_{[u:v] \in E} \left\{ \frac{1}{\deg(u)}\cdot \Big( y_u - \frac{1}{n} \Big)^2 + \frac{1}{\deg(v)}\cdot \Big( y_v - \frac{1}{n} \Big)^2 \right\}  \\ & \quad -\sum_{[u:v] \in E} \Biggl\{ \frac{1}{\deg(u)} \cdot \left( \Big(1 - \frac{\deg(u)}{\gamma \Delta} \Big) \cdot \Big( y_u - \frac{1}{n} \Big) + \frac{\deg(u)}{\gamma \Delta}\cdot \Big(y_v - \frac{1}{n} \Big) \right)^2  \\
         &\quad\quad\quad\quad\quad \, + \frac{1}{\deg(v)} \cdot \left( \Big(1 - \frac{\deg(v)}{\gamma \Delta} \Big) \cdot \Big( y_v - \frac{1}{n} \Big) + \frac{\deg(v)}{\gamma \Delta}\cdot \Big(y_u - \frac{1}{n} \Big) \right)^2 \Biggr\} \\
       &= \sum_{[u:v] \in E} \Biggl\{ \underbrace{ \frac{1}{\deg(u)}\cdot \Big( y_u - \frac{1}{n} \Big)^2  -  \frac{1}{\deg(u)} \cdot \left( \Big(1 - \frac{\deg(u)}{\gamma \Delta} \Big) \cdot \Big( y_u - \frac{1}{n} \Big) + \frac{\deg(u)}{\gamma \Delta}\cdot \Big(y_v - \frac{1}{n} \Big) \right)^2 }_{=:A} \\
         &\quad\quad\quad\quad \, + \underbrace{ \frac{1}{\deg(v)}\cdot \Big( y_v - \frac{1}{n} \Big)^2 - \frac{1}{\deg(v)} \cdot \left( \Big(1 - \frac{\deg(v)}{\gamma \Delta} \Big) \cdot \Big( y_v - \frac{1}{n} \Big) + \frac{\deg(v)}{\gamma \Delta}\cdot \Big(y_u - \frac{1}{n} \Big) \right)^2 }_{=:B} \Biggr\} \\
\end{align*}
We first compute $A$.
\begin{align*}
 A &= \frac{1}{\deg(u)} \cdot \Bigg[ \Big( y_u - \frac{1}{n} \Big)^2  -   \left( \Big(1 - \frac{\deg(u)}{\gamma \Delta} \Big) \cdot \Big( y_u - \frac{1}{n} \Big) + \frac{\deg(u)}{\gamma \Delta}\cdot \Big(y_v - \frac{1}{n} \Big) \right)^2 \Bigg] \\
   &= \frac{1}{\deg(u)}\cdot \left( y_u - \frac{1}{n} + \Big(1 - \frac{\deg(u)}{\gamma \Delta} \Big) \cdot \Big( y_u - \frac{1}{n} \Big) + \frac{\deg(u)}{\gamma \Delta} \Big(y_v - \frac{1}{n} \Big)      \right) \\ &\quad\, \cdot \left(y_u - \frac{1}{n} -    \Big(1 - \frac{\deg(u)}{\gamma \Delta} \Big) \cdot \Big( y_u - \frac{1}{n} \Big) - \frac{\deg(u)}{\gamma \Delta} \Big(y_v - \frac{1}{n} \Big)   \right),
   \intertext{where the last equality follows from $p^2 - q^2 = (p+q) \cdot (p-q)$. Further,}
  A  &= \frac{1}{\deg(u)} \left( \Big( 2 - \frac{\deg(u)}{\gamma \Delta} \Big) \cdot \Big( y_u - \frac{1}{n} \Big) + \frac{\deg(u)}{\gamma \Delta}\cdot \Big(y_v - \frac{1}{n} \Big)  \right)  \\ &~~~~~ \cdot \left(  \frac{\deg(u)}{\gamma \Delta} \cdot \Big( y_u - \frac{1}{n} \Big) - \frac{\deg(u)}{\gamma \Delta}\cdot \Big(y_v - \frac{1}{n} \Big)  \right) \\
   &=  \frac{1}{\deg(u)} \cdot \left( 2 y_u - \frac{2}{n} - \frac{\deg(u)}{\gamma \Delta} (y_u - y_v)     \right) \cdot  \frac{\deg(u)}{\gamma \Delta} \cdot (y_u-y_v) \\
&= \frac{1}{\gamma \Delta} \cdot \Big( y_u - y_v \Big) \cdot \Big(   2 y_u - \frac{2}{n}  - \frac{\deg(u)}{\gamma \Delta} \cdot ( y_u - y_v ) \Big).
\end{align*}
Similarly, we get
\begin{align*}
 B &= \frac{1}{\gamma \Delta} \cdot \Big( y_v - y_u \Big) \cdot \Big( 2  y_v - \frac{2}{n} - \frac{\deg(v)}{\gamma \Delta} \cdot ( y_v - y_u ) \Big),
\end{align*}
and thus,
\begin{align*}
 A + B &= \frac{1}{\gamma \Delta} \cdot \Big( y_u - y_v \Big) \cdot \left( 2 y_u - \frac{2}{n} - \frac{\deg(u)}{\gamma \Delta} \cdot \Big( y_u - y_v \Big) - 2 y_v + \frac{2}{n} + \frac{\deg(v)}{\gamma \Delta} \cdot ( y_v - y_u ) \right) \\
&= \frac{1}{\gamma \Delta} \cdot \Big( y_u - y_v \Big) \cdot  \Big(y_u - y_v\Big) \cdot \Big( 2 - \frac{\deg(u)+\deg(v)}{\gamma \Delta} \Big)  \\
&\geq \frac{1}{\gamma \Delta} \cdot \Big( y_u - y_v \Big)^2 \cdot \left( 2 - \frac{2 \Delta}{\gamma \Delta} \right) \\
&= \frac{1}{\gamma \Delta} \cdot \Big( y_u - y_v \Big)^2 \cdot \left( 2 - \frac{2}{\gamma} \right)\ .
\end{align*}
Therefore,
\begin{align*}
 \Phi^{(t-1)} - \Phi^{(t)} &\geq \frac{1}{\gamma \Delta} \cdot \left( 2 - \frac{2}{\gamma} \right)  \cdot \sum_{[u:v] \in E}  \Big( y_u - y_v \Big)^2,
\end{align*}
i.e.,
\[
\sum_{[u:v]\in E}\left(\P^{t-1}_{u,w}- \P^{t-1}_{v,w} \right)^2 \leq
\frac{\gamma\Delta}{2-2/\gamma}\left( \Phi^{(t-1)}-\Phi^{(t)} \right)\ .
\]
Finally, summing over all rounds gives
\[
\sum_{t=1}^{\infty}\sum_{[u:v]\in E}\left( \P^{t-1}_{u,w}- \P^{t-1}_{v,w}\right)^2\leq \frac{\gamma \Delta}{2 - 2/\gamma}\sum_{t=1}^{\infty} \left( \Phi^{(t-1)}-\Phi^{(t)} \right)\leq\frac{\gamma\Delta}{2-2/\gamma}\cdot\Phi^{(0)} =\frac{\gamma \Delta}{2 - 2/\gamma} \left(1-\frac{1}{n}\right),
\] and thus, $\Psi_2(\bP) <
 \sqrt{\frac{\gamma\cdot\Delta}{2-2/\gamma}}$.

For the lower bound on $\Psi_2(\bP)$, we consider a node $w \in V$ with $\deg(w)=\Delta$ to obtain that
\begin{align*}
 \Psi_2(\bP) &\geq \sqrt{ \sum_{[u:v] \in E} \left( \bP_{u,w}^{0} - \bP_{v,w}^{0}      \right)^2 } \geq \sqrt{ \Delta \cdot \left( 1 - 0      \right)^2 } = \sqrt{\Delta}.
\end{align*}
In the same way, we prove the lower bound on $\Upsilon_2(\bP)$:
\begin{align*}
\Upsilon_2(\bP) &\geq \sqrt{ \frac{1}{2} \sum_{u\in V} \max_{v \in N(u)} \left(  \mathbf{P}_{u,w}^0 - \mathbf{P}_{v,w}^0  \right)^2 }
= \sqrt{ \frac{1}{2} (\Delta+1) \cdot 1 }\enspace.
\end{align*}
\end{proof}

\begin{lemma}\label{lem:martingalebounddiffusion}
Consider the edge-based diffusion model.
Fix two rounds $0 \leq t_1 < t_2$ and the load
vector $x^{(t_1)}$ at the end of round $t_1$. For any family of numbers $g^{(s)}_{u,v}\ ([u:v]\in E,~t_1+1 \leq s \leq t_2)$, define the random variable
\[
   Z := \sum_{s=t_1+1}^{t_2} \sum_{[u:v] \in E} g_{u,v}^{(s)}\cdot e_{u,v}^{(s)}.
\]
Then, $\Ex{Z}=0$, and for any $\delta > 0$, it holds that
\begin{align*}
 \Pro{ \left| Z - \Ex{Z}           \right|  \geq \delta       }
 \leq
2 \exp \left(-  \frac{ \delta^2}{8 \sum_{s=t_1+1}^{t_2} \sum_{[u:v] \in E} \left(g^{(s)}_{u,v}\right)^2 }         \right).
\end{align*}
\end{lemma}
The proof of \lemref{martingalebounddiffusion} is the same as \lemref{martingalebound}, except that the inner sum runs over all edges of the graph, and we have to use the slightly weaker inequality $ \left|e_{u,v}^{(s)} \right| < 1$ instead of $ \left|e_{u,v}^{(s)} \right| \leq 1/2$, which results in an extra factor of $4$ in the denominator.

We now use the above machinery to derive upper bounds on the discrepancy for the edge-based and vertex-based protocol.
\begin{theorem}\label{thm:edgebased}
Let $G$ be an arbitrary graph, and consider the edge-based protocol. Then, the following statements hold:
\begin{itemize}
\item For any round $t$, it holds that
\[
  \Pro{ \max_{w \in V} \left| x_w^{(t)} - \xi_w^{(t)} \right| \leq 4  \sqrt{\log n} \cdot \sqrt{\frac{\gamma \cdot \Delta}{2 - 2/\gamma}} } \geq 1 - 2n^{-1}.
\]
\item After $\Oh\big(\frac{\log (Kn)}{1-\lambda(\bP)}\big)$ rounds, the discrepancy is at most $8  \sqrt{\log n} \cdot \sqrt{\frac{\gamma \cdot \Delta}{2 - 2/\gamma}} + 1$  with probability at least $1-n^{-1}$. In particular, for any constant $\gamma > 1$, after $\Oh\big(\frac{\log (Kn)}{1-\lambda(\bP)}\big)$ rounds, the discrepancy is at most $\Oh\big(\sqrt{\log n \cdot \Delta}\big)$  with probability at least $1-n^{-1}$.
    \end{itemize}
Moreover, consider the vertex-based protocol on a $d$-regular graph $G$. Then, the following statements hold:
\begin{itemize}
\item For any round $t$, it holds that
\[
  \Pro{ \max_{w \in V} \left| x_w^{(t)} - \xi_w^{(t)} \right| = \Oh\left( d^2 \sqrt{\log n} \right)} \geq 1 - 2n^{-1}.
\]
\item After $\Oh\big(\frac{\log (Kn)}{1-\lambda(\bP)}\big)$ rounds, the discrepancy is at most $\Oh\left( d^{2} \sqrt{\log n}\right)$ with probability at least $1-n^{-1}$.
    \end{itemize}
\end{theorem}

\begin{proof}
We prove this result in the same way as \thmref{deviationmatching}, but now we invoke
\lemref{martingalebounddiffusion} instead of \lemref{martingalebound}.
Fix any node $w\in V$, round $t$ and define $Z_w:=x_w^{(t)}-\xi_w^{(t)}$. By \eq{diffstd},
\[Z_w = x_w^{(t)} - \xi_w^{(t)} = \sum_{s=1}^{t} \sum_{[u:v] \in E} \left(  \P_{u,w}^{t-s} - \P_{u,v}^{t-s}       \right) \cdot e_{u,v}^{(s)}.\]
 Applying \lemref{martingalebounddiffusion}, we have $\Ex{Z_w}=0$, and for any $\delta > 0$, it holds that\begin{align*}
 \Pro{ | Z_w | \geq \delta } &\leq 2 \exp \left( - \frac{ \delta^2 }{   8 \sum_{s=1}^{t} \sum_{[u:v] \in E} \left( \P_{u,w}^{t-s} - \P_{u,v}^{t-s}    \right)^2   } \right).
\end{align*}
By definition of the local $2$-divergence, the denominator above is upper bounded by $8\cdot \Psi_2(\P)^2$, and we obtain for $\delta = 4 \sqrt{\log n} \cdot \Psi_2(\P)$ that
\begin{align*}
 \Pro{ |Z_w| \leq 4 \sqrt{\log n }\cdot \Psi_2(\P) } \geq 1 - 2 n^{-2},
\end{align*}
and the first statement follows by using the union bound and the upper bound on $\Psi_2(\P)$ from \thmref{divergencediffusion}. The second statement follows directly by applying~\thmref{continuous}.

For the vertex-based algorithm, it was shown in \cite[Proof of Thm.~1.1]{BCFFS11} that
\begin{align*}
 \Pro{ \max_{w \in V} \left| x_w^{(t)} - \xi_w^{(t)} \right| = \Oh\left( \sqrt{ \log n} \cdot d\cdot \Upsilon_2(\bP) \right)} \geq 1 -n^{-1}.
\end{align*}
Using \thmref{divergencediffusion} with $\gamma=1+1/d$ (and $\Delta=d$) gives
$
  \Upsilon_2(\bP) \leq \Psi_2(\bP) = \Oh(d),
$
which yields the third statement.
Finally, the bound on the discrepancy follows immediately from~\thmref{continuous}.
\end{proof}

\paragraph{Acknowledgements:} We are grateful to Hoda Akbari for the comments on an earlier version of this work.

%

\bibliographystyle{abbrv}
\bibliography{antrag}

\appendix

\section{Concentration Inequalities}

The following concentration inequality is known as ``Azuma's Inequality''.
\begin{theorem}[{\cite[page~92]{MR98}}]\label{thm:martingale}
    Let $X_1,\ldots,X_n$ be a martingale sequence such that for each~$\ell
    \in \{1,\ldots,n\}$, there is a non-negative~$c_{\ell}$ such that
    \[
        \big| X_{\ell} - X_{\ell-1} \big| \leq c_\ell.
    \]
    Then for any~$\delta > 0 $,
    \[
      \Pro{ |X_n - X_0 | \geq \delta } \leq 2 \,
      \exp \bigg( - \frac{\delta^2}{2 \sum_{\ell=1}^n c_\ell^2} \bigg).
    \]
\end{theorem}

\begin{lemma}[{\cite[Theorem~3.4]{PS97}}]\label{lem:panconesi}
Let $X_1,X_2,\ldots,X_n$ be $0/1$-random variables with $\Pro{X_i=1}=p_i$.
Let $X:= \sum_{i=1}^n X_i$ and $\mu=\Ex{X}=\sum_{i=1}^n p_i$.
If for all subsets $S \subseteq \{1,\ldots,n\}$,
\[
  \Pro{ \bigcap_{i \in S} \left\{ X_i = 1 \right\} } \leq \prod_{i \in S} \Pro{ X_i = 1},
\]
then it holds for all $\delta > 0$ that
\begin{align*}
 \Pro{ X \geq (1+\delta) \mu } &\leq \left(   \frac{\ce^{\delta}}{(1+\delta)^{1+\delta}}   \right)^{\mu}.
\end{align*}
\end{lemma}

The following lemma is a standard Chernoff bound for the sum of independent, identically distributed geometric random variables.
\begin{lemma}\label{lem:chernoffgeo}
Consider some fixed $0< p <1$. Suppose that $X_1,\ldots,X_n$ are
independent geometric random variables on $\mathbb{N}$ with $\Pro{X_i = k} =
(1-p)^{k-1} p$ for every $k \in \mathbb{N}$. Let $X = \sum_{i=1}^n
X_i$, $\mu=\Ex{X}$. Then, it holds for all $\beta > 0$ that
\[
  \Pro { X \ge  (1+\beta) \mu } \le \exp\left(-\frac{\beta^2\cdot n}{2 \, (1+\beta) }\right). \]
\end{lemma}

We continue to define the notion of negative regression.
\begin{defi}[{\cite[Definition~21]{DR98}}]\label{def:regression}
A random vector $X=(X_1,\ldots,X_n) \in \{0,1\}^{n}$ is said
to satisfy the {\em negative regression condition} if for any two disjoint subsets $\mathcal{I}$ and $\mathcal{J}$ of $\{1,\ldots,n\}$ and any non-decreasing function $f \colon \{0,1\}^{|\mathcal{I}|}  \to \mathbb{R}$,
\[
  \Ex{ f(X_i, i \in \mathcal{I}) \, \mid \, X_j = \sigma_j, j \in \mathcal{J}}
\]
is non-increasing in each $\sigma_j \in \{0,1\}, j \in \mathcal{J}$.
\end{defi}

\begin{lemma}[{\cite[Lemma~26]{DR98}}] \label{lem:lemmaA8}
Consider a random vector $X=(X_1,\ldots,X_n) \in \{0,1\}^{\gamma}$  that satisfies the negative regression condition. Then for any index set $\mathcal{I} \subseteq \{1,\ldots,n\}$ and any non-decreasing functions $f_i, i \in \mathcal{I}$,
\[
 \Ex{ \prod_{i \in \mathcal{I}} f_i(X_i) } \leq \prod_{i \in \mathcal{I}} \Ex{ f_i(X_i) }.
\]
\end{lemma}

\end{document}